\documentclass[a4paper,UKenglish]{lipics-v2016}
\usepackage[normalem]{ulem}
\usepackage{multirow}
\usepackage{amssymb}
\usepackage{verbatim}
\usepackage{authblk}
\usepackage{amsmath,amsfonts}
\usepackage{graphicx}
\usepackage{color}
\usepackage{xcolor,pgf,tikz,pgflibraryarrows,pgffor,pgflibrarysnakes}
\usepackage{amsmath}
\usepackage{amsfonts}
\usepackage{amssymb}
\usepackage{tikz}
\usepackage{slashbox}
\usetikzlibrary{fit} % fitting shapes to coordinates
\usetikzlibrary{backgrounds} % drawing the background after the foreground

\usetikzlibrary{calc}
\tikzstyle{RectObject}=[rectangle,fill=white,draw,line width=0.5mm]
\tikzstyle{line}=[draw]
\tikzstyle{arrow}=[draw, -latex]

\newcommand{\cnt}{\mathsf{C}}
\newcommand{\pnu}{\mathsf{Pn}}
\newcommand{\UT}{\mathsf{UT}}

\newcommand{\icmet}{\mathsf{ICMET}}
\newcommand{\iecm}{\mathsf{IECM}}

\newcommand{\mcnt}{\mathsf{MC}}
\newcommand{\re}{\mathsf{re}}
\newcommand{\UM}{\mathsf{UM}}

\newcommand{\Ss}{\mathsf{S}}

\newcommand{\Cc}{\mathsf{C}}
\newcommand{\Pp}{\mathsf{P}}
\newcommand{\singl}{\mathsf{single}}

\newcommand{\regm}{\mathsf{Rat}}
\newcommand{\uregm}{\mathsf{URat}}
\newcommand{\reg}{\mathsf{Rat}}
\newcommand{\ureg}{\mathsf{URat}}
\newcommand{\regmtl}{\mathsf{RatMTL}}
\newcommand{\sfmtl}{\mathsf{SfrMTL}}
\newcommand{\po}{\mathsf{po}}
\newcommand{\N}{\mathcal{N}}

\newcommand{\regmitl}{\mathsf{RatMITL}}
\newcommand{\at}{\mathsf{atom}}
\newcommand{\Threads}{\mathsf{Threads}}
\newcommand{\Th}{\mathsf{Th}}
\newcommand{\Next}{\mathsf{Nxt}}
\newcommand{\ovs}{\mathsf{ovs}}
\newcommand{\BD}{\mathsf{BD}}
\newcommand{\BDset}{\mathsf{BDSet}}

\newcommand{\Merge}{\mathsf{Merge}}
\newcommand{\merge}{\mathsf{merge}}
\newcommand{\In}{\mathsf{In}}

\newcommand{\wB}{\Box^{\mathsf{ns}}}
\newcommand{\wU}{\until^{\mathsf{ns}}}
\newcommand{\wF}{\fut^{\mathsf{ns}}}

\newcommand{\until}{\:\mathsf{U}}
\newcommand{\since}{\:\mathsf{S}}
\newcommand{\weaku}{\:\mathsf{W}}

\newcommand{\R}{\:\mathbb{R}}
\newcommand{\fut}{\Diamond}

\mathchardef\mhyphen="2D
\mathchardef\mhyph="2D

\newcommand{\nex}{\mathsf{O}}
\newcommand{\nx}{\mathsf{O}}

\newcommand{\mtl}{\mathsf{MTL}}
\newcommand{\ltl}{\mathsf{LTL}}

\newcommand{\Beh}{\mathsf{Beh}}

\newcommand{\combine}{\mathsf{combine}}

\newcommand{\isEven}{\mathsf{iseven}}

\newcommand{\tptl}{\mathsf{TPTL}}

\newcommand{\mitl}{\mathsf{MITL}}

\newcommand{\oomit}[1]{}

\newcommand{\optptl}{\mbox{$1-\mathsf{OpTPTL}$}}

\begin{document}
\title{Making Metric Temporal Logic Rational}

\titlerunning{Making Metric Temporal Logic Rational}

\author[1]{S. Krishna}
\author[1]{Khushraj Madnani}
\author[2]{P. K. Pandya}
\affil[1]{
%IIT Bombay, Mumbai, India\\
  \text{krishnas,khushraj@cse.iitb.ac.in}}
\affil[2]{
%Tata Institute of Fundamental Research, Mumbai, India\\
  \text{pandya@tifr.res.in}}
\authorrunning{Krishna, Madnani, Pandya}

\maketitle

\begin{abstract} 
We study an extension of $\mtl$ in pointwise time with rational expression guarded modality
$\reg_I(\re)$ where $\re$ is a rational expression over subformulae. We study the decidability and expressiveness of this extension ($\mtl$+$\varphi \ureg_{I, \re} \varphi$+$\reg_{I,\re}\varphi$), called $\regmtl$, as well as its fragment 
$\sfmtl$ where only star-free rational expressions are allowed. Using the technique of temporal projections, we show that $\regmtl$ has decidable satisfiability by giving an equisatisfiable reduction to $\mtl$. We also identify a subclass $\mitl+\ureg$ 
of $\regmtl$ for which our equi-satisfiable reduction gives rise to formulae of $\mitl$, yielding elementary decidability. As our second main result, we show 
a tight automaton-logic connection between $\sfmtl$ and partially ordered (or very weak) 1-clock alternating timed automata. 
%We also identify natural fragments of logic $\tptl$ which correspond to $\sfmtl$. 

\end{abstract}

\section{Introduction}
Temporal logics provide constructs to specify qualitative ordering between events in time. Real time logics are quantitative extensions of temporal logics with the ability to specify real time constraints amongst events. Logics $\mtl$ and $\tptl$ are amongst the prominent real time logics \cite{AH93}.  Two notions of $\mtl$ semantics have been studied in the literature : continuous and pointwise \cite{fst05}. The expressiveness and decidability results vary 
considerably with the semantics used : while the satisfiability checking of $\mtl$ is undecidable in the continuous semantics even for finite timed words  \cite{AFH96}, it is decidable in pointwise semantics with non-primitive recursive complexity  \cite{Ouaknine05}. 
Due to limited expressive power of $\mtl$, several additional modalities have been proposed : the $\mathsf{threshold~ counting}$ modality \cite{count} $\cnt_I^{\geq n} \phi$ states that in time interval $I$ relative to current point, $\phi$ occurs at least $n$ times. The $\mathsf{Pnueli}$ modality \cite{count} $\pnu_I(\phi_1, \ldots, \phi_n)$ states that there is a subsequence of $n$ time points inside interval $I$ where at $i$th point the formula $\phi_i$ holds. In a recent result,  Hunter \cite{hunter} showed that, in continuous time semantics, $\mtl$ enriched with  $\cnt$  modality (denoted $\mtl+\cnt$) is as expressive as $\mathsf{FO}$ with distance $\mathsf{FO}[<,+1]$, which is as expressive as $\tptl$. Unfortunately, satisfiability and model checking of all these logics are  undecidable. This has led us to 
focus on the pointwise case with only the future modality, i.e. logic  $\mtl[\until_I]$, which we abbreviate as $\mtl$ in rest of the paper.
  Also, $\mtl+op$ means $\mtl$ with modalities $\until_I$ as well as $op$. 

In  pointwise semantics, it can be shown that $\mtl \subset \mtl+\cnt \subset \mtl+\pnu$ (see \cite{fossacs16}).  In this paper, we propose
a generalization of threshold counting and Pnueli modalities by a $\mathsf{rational~expression}$ modality $\regm_I \re(\phi_1,\ldots,\phi_k)$, which 
specifies that the truth of the subformulae, $\phi_1,\ldots,\phi_k$, at the set of points within interval $I$ is in accordance with the 
rational expression $\re(\phi_1,\ldots,\phi_k)$. The resulting logic is called $\regmtl$ and is the subject of this paper.
The expressive power of logic $\regmtl$ raises several points of interest.
It can be shown that $\mtl+\pnu \subset \regmtl$, and it can express several new and interesting properties:
(1) Formula $\reg_{(1,2)}((aa)^*)$ states that within time interval $(1,2)$ there is an even number of occurrences of $a$. We will define
 a derived modulo counting modality which states this directly as the formula  $\mcnt_{(1,2)}^{0\%2} a$.
 \oomit{
(2) An exercise regime of 600 to 610 seconds consists  an arbitrary many repetitions of 3 pushup CYCLES where each CYCLE  lasts between 55 and 65 seconds
and it consists of repeated rounds of the action sequence up, strech, down  where the up-down sequence must 
be completed in at most 2 seconds and rounds are followed by rest of exactly 5 seconds. The following formula specifies the routine. (We
assume that all events are mutually exclusive in time. This is easy to state and omitted.)
\[
\begin{array}{l}
 \reg_{[600,610]} ~~((STCYCLE (UPP.strech.down)^* \reg_{[0,5]} (rest))^* end) ~~~\mbox{where} \\
 UPP ~=~ up \land Freg_{(0,2]} (up.strech.down) \\
 STCYCLE ~=~ st \land  Ureg_{[55,65],~(\neg st \land \neg end)^*} ~(st + end)
\end{array}
\]
}
(2) An exercise regime lasting between 60 to 70 seconds  consists of arbitrary many repetitions of three pushup cycles which must be completed 
within 2 seconds. There is no restriction on delay between two cycles to
accomodate weak athletes. This is given by 
$\reg_{[60,70]}((UPP.up.up)^*)$ where 
 $UPP=(up~ \ureg_{(0,2], up} up)$.  
\oomit{
(3) Consider bus interface between processor and memory which specifies timing constraints on address and data signals for 
    memory to be read properly.
    Address should be stable (denoted by proposition $adr-stable$ for $t_{setup}$ time before the enable signal goes high and must continue to remain stable till enable is high.
    The data becomes available from $t_{ready}$ time and it will remain stable till enable is high.
    This can be naturally specified using $\reg$ operator as follows.
       \[
           adr-stable       
       \]
%\end{itemize}
}
The inability to specify rational expression constraints has been an important lacuna of LTL and its practically useful extensions
such as  PSL sugar \cite{psl}, \cite{psl1} (based on Dymanic Logic \cite{DL}) which extend LTL with both counting and 
rational expressions. This indicates that our logic $\regmtl$ is a natural and useful logic for specifying properties. 
%In the timed setting, Asarin   \cite{TRE} has studied timed regular expressions and shown its equivalence to nondeterministic timed automata. 
However, to our knowledge, impact of rational expression constraints on metric temporal modalities have not been studied before. 
As we show in the paper, timing and regularity constraints interact in a fairly complex manner.

As our first main result, we show that satisfiability of $\regmtl$ is decidable by giving an equisatisfiable reduction to $\mtl$. The reduction
makes use of the technique of \emph{oversampled temporal projections} which was previously proposed \cite{time14}, \cite{fossacs16}  and used for proving the decidability of $\mtl+\cnt$. The reduction given here has several novel features such as an $\mtl$  encoding of the run tree of an alternating automaton which restarts the DFA of a given rational expression at each time point (section \ref{sec:eqsat-dfa}).
We  identify two syntactic subsets of $\regmtl$ denoted $\mitl+\ureg$ with 2$\mathsf{EXPSPACE}$ hard satisfiability, and its further subset  $\mitl+\UM$ with $\mathsf{EXPSPACE}$-complete satisfiability. 
As our second main result, we show that the star-free fragment $\sfmtl$ of $\regmtl$ characterizes exactly the class of partially ordered 1-clock alternating timed automata, thereby giving a tight logic automaton connection. The most non-trivial part
of this proof is the construction  of $\sfmtl$ formula 
equivalent to a given partially ordered 1-clock alternating timed automaton $\mathcal{A}$ (Lemma \ref{lem:poata-sfmtl}).

%
%This  construction describes the behaviours 
%starting at each location $s$ of the  $\mathcal{A}$,  by giving LTL formulae in each of the intervals $[0,0], \dots, (K, \infty)$ 
%where  $K$ is the maximal constant used in $\mathcal{A}$. Each of these LTL formulae  can be written as a star-free expression that describes 
%the behaviour of location $s$ over the chosen time interval $I$. This is then used in obtaining a $\sfmtl$  
%formula where $\reg_I(\re)$ describes the star-free expression $\re$ that is true on interval $I$. The partial order of locations of $\mathcal{A}$
% helps in stitching together the behaviours as one moves from one location to another, yielding at the end, 
% a  $\sfmtl$ formula for the initial location $s_0$ that results in acceptance. 

%Secondly, previous work of Ouaknine and Worrell \cite{LICS05} showed that  $\mtl$ (over finite pointwise models) can be 
%reduced to equivalent partially ordered 1-clock alternating timed automata. However, the converse was not clear.
%%Moreover, they showed that 1-clock $\tptl$ augmented with 
%%fixed point operator is expressively equivalent to 1-clock alternating timed automata.
%%Generalising this, it is fairly easy to show that $\regmtl$ can also be reduced to equivalent  1-clock alternating timed automata.
%Let $\sfmtl$ be subset  of $\regmtl$ where the rational expressions are star-free.

%\input{intro-k.tex}
\section{Timed Temporal Logics}
\label{sec:mtlmc}
This section describes the syntax and semantics of the timed temporal logics 
needed in this paper : $\mathsf{MTL}$ and  $\mathsf{TPTL}$.
%\subsection{Timed Temporal Logics}
  Let $\Sigma$ be a finite set of propositions. A finite timed word over $\Sigma$ is a tuple
$\rho = (\sigma,\tau)$.   $\sigma$ and $\tau$ are sequences $\sigma_1\sigma_2\ldots\sigma_n$ and  $\tau_1\tau_2\ldots \tau_n$ respectively, with $\sigma_i \in \mathcal{P}(\Sigma)-\emptyset$,  and $\tau_i \in \R_{\geq 0}$
 for $1 \leq i \leq n$ and $\forall i \in dom(\rho)$,  $\tau_i \le \tau_{i+1}$, where $dom(\rho)$ is the set of positions $\{1,2,\ldots,n\}$ in the timed word. For convenience, we assume $\tau_1=0$. 
 The $\sigma_i$'s can be thought of as labeling positions $i$ in $dom(\rho)$.  
 For example, given $\Sigma=\{a,b,c\}$,  $\rho=(\{a,c\},0)(\{a\},0.7)(\{b\},1.1)$ is a timed word.
$\rho$ is strictly monotonic iff $\tau_i < \tau_{i+1}$ for all $i,i+1 \in dom(\rho)$. 
Otherwise, it is weakly monotonic. 
The set of finite timed words over $\Sigma$ is denoted $T\Sigma^*$. 
Given $\rho=(\sigma,\tau)$ with $\sigma=\sigma_1\dots \sigma_n$,  
 $\sigma^{\singl}$ denotes the set of 
words $\{w_1w_2 \dots w_n \mid w_i \in \sigma_i\}$. For $\rho$ as above,
$\sigma^{\singl}$ consists of $(\{a\},0)(\{a\},0.7)(\{b\},1.1)$
and $(\{c\},0)(\{a\},0.7)(\{b\},1.1)$. 
%Metric Temporal Logic ($\mathsf{MTL}$) extends  linear temporal logic ($\mathsf{LTL}$) by adding timing constraints 
%to the ``until'' modality of $\mathsf{LTL}$. $\mathsf{MTL}$ is parametrized by using a  permitted 
Let $I\nu$ be a set of open, half-open or closed time intervals. 
 The  end points of these intervals are  in $\mathbb{N}\cup \{0,\infty\}$.   For example, 
$[1,3), [2, \infty)$. For $\tau \in \R_{\geq 0}$ and interval  $\langle a, b\rangle$, with $<\in \{(,[\}$ and $>\in \{ ],) \}$, 
$\tau+\langle a, b\rangle$ stands for the interval 
$\langle \tau+a, \tau+b\rangle$. \\
%\subsection*{Metric Temporal Logic}
%\label{prelim}
\noindent{\bf Metric Temporal Logic}($\mtl$). Given a finite alphabet $\Sigma$,  the formulae of $\mathsf{MTL}$ are built from $\Sigma$  using boolean connectives and 
time constrained version of the modality $\until$ as follows:\\
$\varphi::=a (\in \Sigma)~|true~|\varphi \wedge \varphi~|~\neg \varphi~|
~\varphi \until_I \varphi$, 
where  $I \in I\nu$.    
\label{point}
For a timed word $\rho=(\sigma, \tau) \in T\Sigma^*$, a position 
$i \in dom(\rho)$, and an $\mathsf{MTL}$ formula $\varphi$, the satisfaction of $\varphi$ at a position $i$ 
of $\rho$ is denoted $(\rho, i) \models \varphi$, and is defined as follows: (i)
\noindent $\rho, i \models a$  $\leftrightarrow$  $a \in \sigma_{i}$, (ii) $\rho,i  \models \neg \varphi$ $\leftrightarrow$  $\rho,i \nvDash  \varphi$, 
(iii) $\rho,i \models \varphi_{1} \wedge \varphi_{2}$   $\leftrightarrow$ 
$\rho,i \models \varphi_{1}$ 
and $\rho,i\ \models\ \varphi_{2}$, (iv)
$\rho,i\ \models\ \varphi_{1} \until_{I} \varphi_{2}$  $\leftrightarrow$  $\exists j > i$, 
$\rho,j\ \models\ \varphi_{2}, \tau_{j} - \tau_{i} \in I$, and  $\rho,k\ \models\ \varphi_{1}$ $\forall$ $i< k <j$. 

%We assume the existence of  a special point called 0, outside $dom(\rho)$. The time stamp 
%of this point is 0 ($\tau_0=0$).\footnote{All distances in $\mtl$ are relative to the starting point; $\mtl$ cannot check the time stamp of the first action point. For instance, the words $(a,1.2)(b,2.1)$ and $(a,0)(b,0.9)$ are indistinguishable by $\mtl$.  
%This is not the case with (alternating) timed automata. The assumption of $\tau_0=0$ is to bridge this gap, especially while going from automata to logic, as in  lemma \ref{lem:ata-1tptl}.}    
%\noindent $\rho$ satisfies $\varphi$ denoted $\rho \models \varphi$ 
%iff $\rho,1 \models \varphi$.
 The language of a $\mathsf{MTL}$ formula $\varphi$ is $L(\varphi)=\{\rho \mid \rho, 1 \models \varphi\}$. 
Two formulae $\varphi$ and $\phi$ are said to be equivalent denoted as $\varphi \equiv \phi$ iff $L(\varphi) = L(\phi)$.
Additional temporal connectives are defined in the standard way: 
we have the constrained future eventuality operator $\fut_I a \equiv true \until_I a$ 
and its dual 
$\Box_I a \equiv \neg \fut_I \neg a$.
We also define the next operator as $\nex_I \phi \equiv \bot \until_I  \phi$. 
Non strict versions of  operators 
are defined as  $\wF_I a=a \vee \fut_I a, 
\wB_I a\equiv a \wedge \Box_I a$, $a \wU_I b\equiv b \vee [a \wedge (a \until_I b)]$ if $0 \in I$, and 
$[a \wedge (a \until_I b)]$ if $0 \notin I$. Also, $a \weaku b$ is a shorthand for $\Box a \vee (a \until b)$. 
 The subclass of $\mathsf{MTL}$ obtained by restricting the intervals $I$ in the until modality 
to non-punctual intervals is denoted  $\mathsf{MITL}$.  \\
\noindent{\bf Timed Propositional Temporal Logic} ($\mathsf{TPTL}$). 
$\tptl$ is a prominent real time extension of $\mathsf{LTL}$, where timing constraints are specified with the help of freeze clocks. 
The set of $\mathsf{TPTL}$ formulas are defined inductively as 
$\varphi::=a (\in \Sigma)~|true~|\varphi \wedge \varphi~|~\neg \varphi~|~\varphi \until \varphi~|~y.\varphi~|~y\in I$. 
$\mathcal{C}$ is a set of
clock variables progressing at the same rate, 
$y\in \mathcal{C}$, and $I$ is an interval as above. %of the form ${<}a,b{>}$ $a,b\in \mathbb{N}$ with $<\in \{(,[\}$ and $>\in \{ ],) \}$.
%$\mathsf{TPTL}$ is interpreted over finite timed words over  $\Sigma$.
%The truth of a formula is interpreted at a position
%$i\in \mathbb{N}$ along the word.
 For a timed word $\rho=(\sigma_1,\tau_1)\dots(\sigma_n,\tau_n)$, we define the satisfiability relation, $\rho,
i, \nu \models \phi$ saying that  the formula $\phi$ is true at position
$i$ of the timed word $\rho$ with valuation $\nu$ of all the clock
variables as follows: 
(1) 	$\rho, i, \nu \models a$   $\leftrightarrow$  $a \in \sigma_{i}$,
(2) $\rho,i,\nu  \models \neg \varphi$  $\leftrightarrow$  $\rho,i,\nu \nvDash  \varphi$,
(3)		$\rho,i,\nu \models \varphi_{1} \wedge \varphi_{2}$  $\leftrightarrow$   $\rho,i,\nu \models \varphi_{1}$ 
		and $\rho,i,\nu\ \models\ \varphi_{2}$,
(4) 		$\rho,i,\nu \models x.\varphi $  $\leftrightarrow$ $\rho,i,\nu[x \leftarrow \tau_i] \models \varphi$,
(5)		$\rho,i,\nu \models x \in I $   $\leftrightarrow$  $\tau_i - \nu(x) \in I$,
(6)		$\rho,i,\nu\ \models\ \varphi_{1} \until \varphi_{2}$  $\leftrightarrow$  $\exists j > i$, 
		$\rho,j,\nu \ \models\ \varphi_{2}$,  and 
		 $\rho,k,\nu \ \models\ \varphi_{1}$ $\forall$ $i < k < j$.
$\rho$ satisfies $\phi$ denoted $\rho \models \phi$ iff $\rho,1,\bar{0}\models \phi$. Here $\bar{0}$ 
is the valuation obtained by setting all clock variables to 0.  
We denote by $k{-}\mathsf{TPTL}$ the fragment of  $\mathsf{TPTL}$ using at most $k$ clock variables.
%The fragment of $\mathsf{TPTL}$ with $k$ clock variables is denoted $k{-}\mathsf{TPTL}$. 
 \begin{theorem}[\cite{Ouaknine05}]
 \label{thm-basic}
	 $\mathsf{MTL}$ satisfiability is decidable over finite timed words and is non-primitive recursive. 
\end{theorem}
\paragraph*{$\mtl$ with Rational Expressions($\regmtl$)}
We propose an  extension of $\mathsf{MTL}$ with rational expressions, that forms the core of the paper.  
These modalities can assert the truth of a rational expression (over subformulae) within a particular time interval with respect to the present point. 
For example, $\regm_{(0,1)}(\varphi_1.\varphi_2)^+$ when evaluated at a point $i$, asserts the existence of  $2k$ points 
$\tau_i < \tau_{i+1} < \tau_{i+2} < \dots < \tau_{i+2k} < \tau_{i}+1$, $k >0$, 
such that $\varphi_1$ evaluates to true at $\tau_{i+2j+1}$, and $\varphi_2$ evaluates to true at 
$\tau_{i+2j+2}$, for all $0 \leq j <k$. 
%Since we have to check the truth of $\varphi_1, \varphi_2$ 
%at (intermediate) time points, we write $\regm_{(0,1)}(\varphi_1.\varphi_2)^+$ as 
%$\regm^{\{\varphi_1, \varphi_2\}}_{(0,1)}(\varphi_1.\varphi_2)^+$. In general, this superscript will be a set of formulae $S$, which will be useful to evaluate the truth of the rational expression. 

\noindent{$\regmtl$ \bf{Syntax}:} Formulae of $\regmtl$ are built from $\Sigma$ (atomic propositions) as follows:\\
	$\varphi::=a (\in \Sigma)~|true~|\varphi \wedge \varphi~|~\neg \varphi~|~ \regm_I \re(\Ss)~|
	~\varphi \uregm_{I,\re(\Ss)} \varphi$, 
	where $I \in I\nu$
	and $\mathsf{S}$ is a finite set of formulae of interest, and  $\re(\Ss)$ is defined as a rational expression over $\Ss$. 
	 $\re(\Ss)::= \varphi (\in \Ss)~|~\re(\Ss).\re(\Ss)~|~\re(\Ss)+\re(\Ss)~|~[\re(\Ss)]^*$. Thus, $\regmtl$ 
	 is $\mtl + \ureg+ \reg$. 
An \emph{atomic}  rational expression $\re$ is any well-formed formula $\varphi \in \regmtl$. 
%For a rational expression $\re$,
%let $\Gamma$ be the set of 
%all subformulae and their negations appearing in $\re$. 
%For example, if $\re=a \ureg^\{\}_{(0,1), \reg_{(1,2)}[\reg_{(0,1)}b]} b$, then 
%$\Gamma$ consists of 
%$\reg_{(1,2)}[\reg_{(0,1)}b],\reg_{(0,1)}b,b$ and their negations. \\
%Let $\mathsf{Cl}(\Gamma)$ denote consistent sets\footnote{a set $S$ is consistent iff $\varphi \in S\leftrightarrow \neg \varphi \notin S$} in $\mathcal{P}(\Gamma)$. 
%$L(\re)$ is the set of strings over $\mathsf{Cl}(\Gamma)$  defined as follows.  Let $S \in \mathsf{Cl}(\Gamma)$. 
%    $$   
%  L(\re)=
%  \begin{cases}
%    \{S \mid a \in S\} & \text{if } \re=a,  \\
%    \{S \mid \varphi_1, \varphi_2 \in S\}  & \text{if } \re=\varphi_1 \wedge \varphi_2, \\
%    \{S \mid \varphi \notin S\} & \text{if } \re=\neg \varphi, \\
%     L(\re_1).L(\re_2) & \text{if } \re=\re_1.\re_2, \\
%    L(\re_1) \cup L(\re_2)   & \text{if } \re=\re_1+\re_2, \\
%    [L(\re_1)]^*  & \text{if } \re=(\re_1)^*. 
%    \end{cases}
%  $$
%If $\re$ is not an atomic rational expression, but has the form $\re_1+\re_2$ or $\re_1.\re_2$ 
%or $(\re_1)^*$, then we use the standard definition of $L(\re)$ as $L(\re_1) \cup L(\re_2)$, 
%$L(\re_1).L(\re_2)$ and $[L(\re_1)]^*$ respectively. 

%We assume that 
%	 $\Ss$ is \emph{exclusive}; that is, no two formulae of $\Ss$ evaluate 
%	 to true at the same point. If $\Ss=\{\varphi_1, \dots, \varphi_n\}$ is not exclusive, we 
%	 can make it exclusive by constructing $\Exc(\Ss)=\{\bigwedge_{i \in K} \varphi_i \wedge \bigwedge_{i \notin K} \neg \varphi_i \mid K \subseteq \{1,2,\dots,n\}\}$.
	
\noindent{$\regmtl$ {\bf Semantics}:} 
For a timed word $\rho=(\sigma, \tau) \in T\Sigma^*$, a position 
$i \in dom(\rho)$, and a $\regmtl$ formula $\varphi$, a finite set $\Ss$ of formulae, we define the satisfaction of $\varphi$ at a position $i$ 
as follows. For positions $i < j \in dom(\rho)$, let $\mathsf{Seg}(\Ss, i, j)$ denote 
the untimed word over $\mathcal{P}(\mathsf{S})$ 
obtained by marking  the positions $k \in \{i+1, \dots, j-1\}$ of $\rho$ with 
$\psi \in \mathsf{S}$ iff $\rho,k \models \psi$.  
For  a position $i{\in} dom(\rho)$ and an interval $I$, 
 let $\mathsf{TSeg}(S, I, i)$ denote 
the untimed word over $\mathcal{P}(\mathsf{S})$
obtained by marking all the positions $k$ 
such that $\tau_k - \tau_i \in I$ 
  of $\rho$ with 
$\psi \in \mathsf{S}$ iff $\rho,k \models \psi$. 
\begin{enumerate}
\item %$\varphi=\varphi_{1} \ureg_{I,\re(\Ss)} \varphi_{2}$. 
%\begin{itemize}
%\item 
 $\rho,i \models \varphi_{1} \ureg_{I,\re(\Ss)} \varphi_{2}$  $\leftrightarrow$  $\exists j {>} i$, 
$\rho,j {\models}\ \varphi_{2}, \tau_{j} - \tau_{i} {\in} I$, $\rho,k\ {\models}\ \varphi_{1}$ ${\forall} i{<} k {<}j$ and, 
$[\mathsf{Seg}(\Ss, i, j)]^{\singl} \cap  
 L(\re(S)) \neq \emptyset$, where $L(\re(\Ss))$ is the language  
of the rational expression $\re$ formed over the set $\Ss$. 
The subclass of $\regmtl$ using only the $\ureg$ modality is denoted 
$\regmtl[\ureg]$ or $\mtl+\ureg$ and if only non-punctual intervals are used, then it is denoted $\regmitl[\ureg]$
or $\mitl+\ureg$.
%Since $\mathsf{Seg}(\Ss, i, j)$ is a word over $\mathcal{P}(\mathsf{S})$, 
%and $L(\re(\Ss))$ is a language over $\Ss$, we clarify the notation 
%$\mathsf{Seg}(\Ss, i, j) \in L(\re(\Ss))$.  
%  If a position $p$ in the word $\mathsf{Seg}(\Ss, i, j)$ is marked with formulae $\varphi_1, \dots, \varphi_j$, then  we obtain $j$ words, $w_1, \dots, w_j$, where the $p$th position of $w_i$  
% is marked only $\varphi_i$.  Thus, if $\mathsf{Seg}(\Ss, i, j)=A_{i+1} \dots A_{j-1}$, with each $A_k \in \mathcal{P}(\mathsf{S})$, then 
% we obtain  the set of words 
% $\mathcal{W}=\{\psi_{i+1}  \dots \psi_{j-1} \mid \psi_k \in A_k\}$.   
%Then  $\mathsf{Seg}(\Ss, i, j) \in L(\re(\Ss))$ iff  $\mathcal{W} \cap L(\re(\Ss)) \neq \emptyset$. 
% For example, %for $\Ss=\{\varphi_1 \wedge \neg \varphi_2, \varphi_2 \wedge \neg \varphi_1, \varphi_1 \wedge \varphi_2, \neg \varphi_1 \wedge \neg \varphi_2\}$, 
%$\mathsf{Seg}(\Ss, 1, 4)$ is the untimed word obtained by marking positions 2,3 of the timed word. 
% \end{itemize}

\item %$\varphi=\reg_I \re$. 
$\rho,i \models \reg_I \re$ $\leftrightarrow$ 
$[\mathsf{TSeg}(S, I, i)]^{\singl} \cap L(\re(S)) \neq \emptyset$. 
\end{enumerate}
The language accepted by a $\regmtl$ formula $\varphi$ is given by 
$L(\varphi)=\{\rho \mid \rho, 0 \models \varphi\}$. \\
\noindent \emph{Example 1}. 
Consider the formula $\varphi=a \uregm_{(0,1), ab^*} b$. Then $\re{=}ab^*$, and the subformulae 
of interest are $a,b$. 
For $\rho{=}(\{a\},0)(\{a,b\},0.3)(\{a,b\},0.99)$, 
 $\rho, 1 \models \varphi$, since $a {\in} \sigma_2, b {\in} \sigma_3$, $\tau_3 {-} \tau_1 {\in} (0,1)$ and 
 $a \in [\mathsf{Seg}(\{a,b\}, 1, 3)]^{\singl} \cap L(ab^*)$. 
On the other hand, for the word $\rho=(\{a\},0)(\{a\},0.3)(\{a\},0.5)(\{a\},0.9)(\{b\},0.99)$, we know that 
$\rho, 1 \nvDash \varphi$, since even though $b \in \sigma_5, a \in \sigma_i$ for $i <5$, 
$[\mathsf{Seg}(\{a,b\}, 1, 5)]^{\singl}=aaa$ and $aaa \notin L(ab^*)$. \\
\noindent \emph{Example 2}.
Consider the formula $\varphi=\regm_{(0,1)}[\neg \regm_{(0,1)}a]$.
% \begin{enumerate}
 %\item 
 For the word $\rho=(\{a,b\}, 0)$ $(\{a,b\},0.91)(\{a\},1.2)$, 
 to check $\varphi$ at position 1, we check position 2 of the word, since  $\tau_2-\tau_1 \in (0,1)$. The formulae of interest for marking 
 is $\{\neg \regm_{(0,1)}a\}$. Position 2 is not marked, since 
 $\rho,2 \models \regm_{(0,1)}a$.  
 Then $[\mathsf{TSeg}(\Ss, (0,1), 1)]^{\singl}=\emptyset \notin L(\neg \regm_{(0,1)}a)$. However, for the word
   $\rho=(\{a,b\}, 0)$ $(\{a,b\},0.91)(\{b\},1.1)$,
 $\rho, 1 \models \varphi$, since position 2 is marked with 
$\neg \regm_{(0,1)}a$, and 
 $\neg \regm_{(0,1)}a \in L(\neg \regm_{(0,1)}a) \cap [\mathsf{TSeg}(S, (0,1), 1)]^{\singl}$.

%\item  For   $\rho=(\{a,b\}, 0)(\{b\},0.6)(\{a,b\},0.91)(\{a\},1.7)$, positions 2,3 are marked.
%Position 2 is marked $\{\}$
%and position 3 is marked $a$. 
%$\mathsf{TSeg}(S, (0,1), 1) {=}\{\}a$. Thus, no word obtained 
%from $\mathsf{TSeg}(S, (0,1), 1)$ is in $L(\neg \regm_{(0,1)}a)$.
%Hence $\rho,1 \nvDash \varphi$.  
 %\end{enumerate}
\noindent \emph{Example 3}.
Consider the formula $\varphi=\regm_{(0,1)}[\regm_{(0,1)}a]^*$. \\
 For $\rho=(\{a,b\},0)$$(\{a,b\},0.7)(\{b\},0.98)(\{a,b\},1.4)$, we have 
 $\rho, 1 {\nvDash}  \regm_{(0,1)}[\regm_{(0,1)}a]^*$, since 
point 2 is not marked $\regm_{(0,1)}a$, even though point 3 is.

\noindent{\bf{Generalizing Counting, Pnueli \& Mod Counting Modalities}}
The following reductions show that $\regmtl$ subsumes most of the extensions of $\mtl$ studied in the literature.\\
%\begin{itemize}
 (1) \textbf{Threshold Counting} constraints \cite{count}, \cite{cltl}, \cite{fossacs16} specify the number of times a property holds within some time region is at least (or at most) $n$. These  can be expressed in $\regmtl$:
  (i) $\cnt^{\ge n}_I \varphi \equiv \reg_I(\re_{th})$,
  (ii) $ \phi_1 \UT_{I,\varphi \ge n } \phi_2 \equiv \phi_1\ureg_{I,\re_{th}} \phi_2$, where 
  %$\re_{th}=true^*(\varphi.true^*.\ldots.\varphi.true^*)^n$.\\
  $\re_{th}={true^*\! \underbrace{\varphi.true^*.\ldots.\varphi.true^*}_\text{$n$ times}}$. \\
  (2) \textbf{Pnueli Modalities}  $\cite{count}$, which  enhance the expressiveness of $\mitl$ in continuous semantics preserving the complexity, can be 
  written in $\regmtl$: 
  $\mathsf{Pn}(\phi_1,\phi_2,\ldots,\phi_k)$ can be written as  $\reg_I(true^*.\phi_1.true^*\phi_2.\ldots.true^*.\phi_k.true^*)$. \\
  (3) \textbf{Modulo Counting} constraints \cite{DBLP:conf/lics/BaziramwaboMT99}, \cite{atva10}
  specify the number of times a property holds modulo $n \in \mathbb{N}$, in some region. 
   We extend these to the timed setting by proposing two modalities $\mcnt^{k\%n}_I$ and 
   $\UM_{I,\varphi=k\%n}$. $\mcnt^{k\%n}_I \varphi$ checks if the number of times $\varphi$ 
  is true in interval $I$ is   $M(n)+k$, where $M(n)$ denotes a non-negative integer multiple of $n$, and 
  $0 \leq k \leq n-1$, while $\varphi_1 \UM_{I,\#{\psi}=k\%n} \varphi_2$ when asserted at a point $i$, checks 
  the existence of $j >i$ such that $\tau_j-\tau_i \in I$,  $\varphi_2$ is true at $j$, $\varphi_1$ holds 
  between $i, j$,   and the number of times $\psi$ is true 
  between $i, j$ is $M(n)+k$, $0 \leq k \leq n-1$.   As an example, %$\varphi=\wB(a \rightarrow \mcnt_{(0,1)}^{0\%2}b)$ says that whenever 
%there is an $a$ at  time point $\tau$, the number of $b$'s in interval $(\tau, \tau+1)$ is even. 
$\psi=true \UM_{(0,1), \#b=1\%2}(a \vee b)$, when asserted at a point $i$, checks the existence 
of a point $j >i$ such that $a$ or $b \in \sigma_j$, $\tau_j-\tau_i \in (0,1)$, 
 and the number of points between $i, j$ where $b$ is true is odd.  
  Both these modalities can be rewritten equivalently in $\regmtl$ as follows:
    $\mcnt^{k\%n}_I \varphi \equiv \reg_I(\re_{mod})$ and 
   $ \phi_1 \UM_{I,\varphi= k \% n } \phi_2 \equiv \phi_1\ureg_{I,\re_{mod}} \phi_2$
    where
  $\re_{mod} = ([\underbrace{(\neg \varphi)^*.\varphi.\ldots.(\neg \varphi)^*.\varphi}_\text{$n$ times}]^*.[\underbrace{(\neg \varphi)^*.\varphi.\ldots.(\neg \varphi)^*.\varphi}_\text{$k$ times}]$. 
 The extension of $\mtl$ ($\mitl$) with only $\UM$ 
is denoted $\mtl+\UM$ ($\mitl+\UM$) while $\mtl+\mcnt$ ($\mitl+\mcnt$) denotes the extension
using  $\mcnt$.

\vspace{-.4cm}
\section{Satisfiability of $\regmtl$ and  Complexity}
  The main results of this section are as follows.  
  \begin{theorem}
 \label{mitl-ureg}
% \begin{enumerate}
(1) Satisfiability of $\regmtl$ is decidable.
(2)  Satisfiability of $\mitl+\UM$ is $\mathsf{EXPSPACE}$-complete.
 (3)  Satisfiability of $\mitl+\ureg$ is in $\mathsf{2EXPSPACE}$.\\
 (4)  Satisfiability of $\mitl+\mcnt$ is $\bf{F}_{\omega^{\omega}}$-hard.
  % \end{enumerate}
 \end{theorem}
  %The proof of Theorem \ref{mitl-ureg}.1 is given below.
  % follows from Lemmas \ref{elm-reg} and \ref{elm-ureg}, and from Theorem \ref{thm-basic}. 
Details of Theorems \ref{mitl-ureg}.2, \ref{mitl-ureg}.3, \ref{mitl-ureg}.4 can be found in Appendices \ref{app:th-um}, \ref{app:th-ureg} and  \ref{app:th-ack}. 
\begin{theorem}
\label{thm:exp}
  $\mtl+\ureg \subseteq \mtl+\reg$, $\mtl+\UM \subseteq \mtl+\mcnt$. 
  \end{theorem}
 Theorem \ref{thm:exp} shows that the $\reg$ modality can capture $\ureg$ (and likewise, $\mcnt$ captures 
 $\UM$). Thus, $\regmtl \equiv \mtl+\reg$. Observe that any $\re$  can be decomposed into finitely many factors, i.e. 
 $\re = \sum \limits_{i=1}^n R^i_1.R^i_2$. 
 Given $true \ureg_{[l,u),\re} \phi_2$,  
 we assert  $R^i_1$
 within interval $(0,l]$ and $R^i_2$ in the prefix of the 
 latter part within $[l,u)$, followed by $\phi_2$. 
$true \ureg_{[l,u),\re} \phi_2 \equiv \bigvee \limits_{i\in \{1,2\ldots,n\}}\reg_{(0,l)}R^i_1 \wedge \reg_{[l,u)}R^i_2.\phi_2.\Sigma^*$. 
 The proofs can be seen in Appendix \ref{app:exp}.
% \begin{lemma}
% 		\label{reg-exp}
% 		Given any regular expression $R$, there exist 
% 	finitely many regular expressions $R^1_1, R^1_2, \dots, R^n_1, R^n_2$
% 		such that $R=\bigcup_{i=1}^n R^i_1. R^i_2$. That is, 
% for any string $\sigma \in R$ and for any decomposition of $\sigma$ as  $\sigma_1.\sigma_2$, 
% there exists some $i\le n$ such that $\sigma_1 \in R^i_1$ and  $\sigma_2 \in R^i_2$.
% 		\end{lemma}

 \subsection{Proof of Theorem \ref{mitl-ureg}.1}
 \label{sec:eqsat-dfa}
 
\noindent {\bf Equisatisfiability} We will use the technique of equisatisfiability modulo oversampling \cite{time14}
in the proof of Theorem \ref{mitl-ureg}. 
 Using  this technique, formulae $\varphi$ in one logic (say $\regmtl$) can be   
 transformed into formulae $\psi$ over a simpler logic (say $\mtl$) such that 
 whenever $\rho \models \varphi$ for a timed word $\rho$ over alphabet $\Sigma$, 
 one can construct a timed word $\rho'$ over an extended set of positions and an extended alphabet 
 $\Sigma'$ such that $\rho' \models \psi$ and vice-versa \cite{time14}, \cite{fossacs16}.  
   %The extension is applied both to  the domain of $\rho$ as well as 
 %the alphabet $\Sigma$ and is called \emph{oversampling}. 
 In \emph{oversampling}, (i) $dom(\rho')$ is 
 extended by adding  some extra positions between the first and last point of $\rho$, (ii) the labeling of a position 
 $i \in dom(\rho)$ is over the extended alphabet $\Sigma' \supset \Sigma$ and can be a superset of the previous labeling over $\Sigma$, 
 while  the new positions are labeled using only the new symbols $\Sigma'-\Sigma$.
% This allows us to come up with a formula $\psi$ such that  whenever $\rho \models \varphi$, we have 
 %$\rho' \models \psi$.  
 We can recover $\rho$ from $\rho'$ by erasing the new points and the new symbols.
   A restricted use of oversampling, when one only extends 
 the alphabet and not the set of positions of a timed word $\rho$ is called \emph{simple extension}. 
 In this case, if $\rho'$ is a simple extension of $\rho$, then $dom(\rho)=dom(\rho')$, and by erasing the new symbols from $\rho'$, we obtain $\rho$.
  See Figure \ref{eg-os} for an illustration. 
  The formula $\psi$ over the larger alphabet $\Sigma' \supset \Sigma$ such that $\rho' \models \psi$ iff $\rho \models \varphi$
    is said to be equisatisfiable modulo temporal projections to $\varphi$.  In particular, $\psi$ is equisatisfiable to 
 $\varphi$  modulo simple extensions or modulo oversampling, depending on how the word $\rho'$ is constructed  
 from the word $\rho$. \begin{figure}[h]
\includegraphics[scale=0.3]{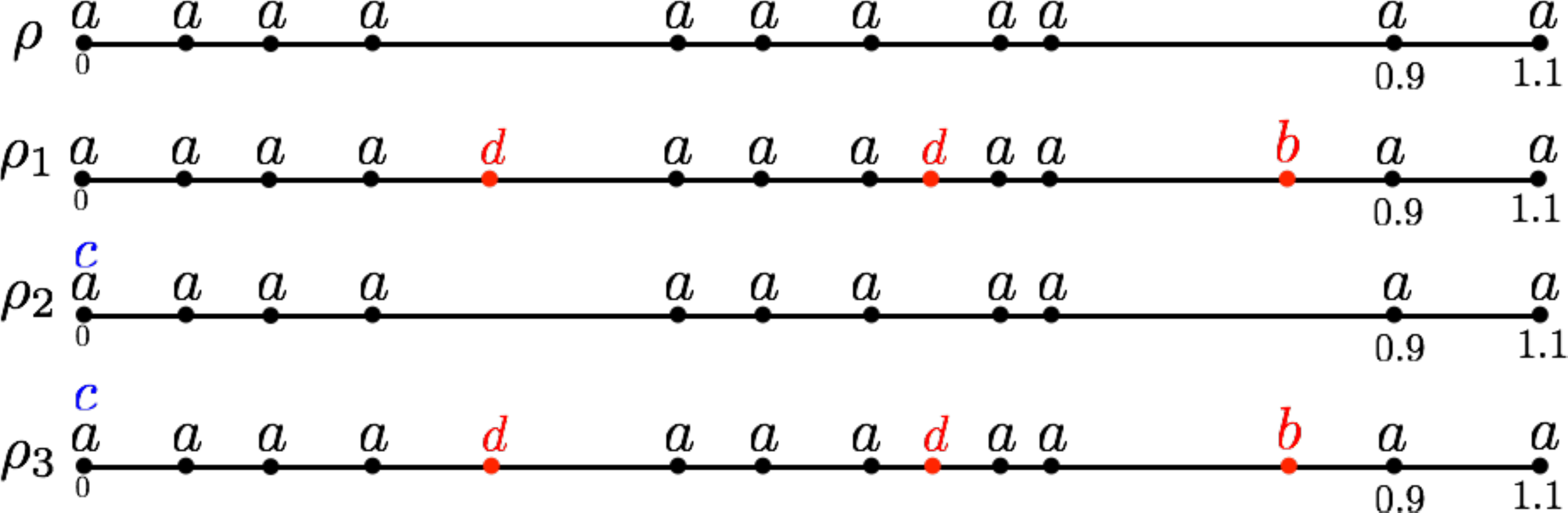}
\caption{$\rho$ is over  $\Sigma=\{a\}$ and satisfies $\varphi=\Box_{(0,1)}a$.
$\rho_1$ is an oversampling of $\rho$ over an extended alphabet $\Sigma_1=\Sigma \cup \{b,d\}$ and 
satisfies  $\psi_1=\Box(b \leftrightarrow \neg a) \wedge (\neg b \until_{(0,1)} b)$. The red points in $\rho_1$ are the oversampling  points.
$\rho_2$ is a simple extension of $\rho$ over an extended alphabet $\Sigma_2=\Sigma \cup \{c\}$ and 
satisfies  $\psi_2=\Box(c \leftrightarrow \Box_{(0,1)}a) \wedge c$. It can be seen that 
$\psi_1$ is equivalent to $\varphi$ modulo oversampling, and $\psi_2$ is equivalent to $\varphi$ modulo simple extensions
using the (respectively oversampling, simple) extensions $\rho_1, \rho_2$ of $\rho$. However, $\rho_3$ above, obtained 
by merging $\rho_1, \rho_2$, eventhough an oversampling of $\rho$, is not a good model for the formula 
$\psi_1 \wedge \psi_2$ over $\Sigma_1 \cup \Sigma_2$. However, we can relativize 
$\psi_1$ and $\psi_2$ with respect to $\Sigma$ as 
$\Box(act_1 {\rightarrow} (b {\leftrightarrow} {\neg a})) {\wedge} [(act_1 {\rightarrow} {\neg b}) {\until_{(0,1)}} (b {\wedge} act_1)]$, and 
$\Box(act_2 \rightarrow (c \leftrightarrow \Box_{[0,1)} (act_2 \rightarrow a))) \wedge (act_2 \wedge c)$
where 
$act_1=\bigvee \Sigma_1, act_2=\bigvee \Sigma_2$.
 The relativized formula $\kappa=Rel(\psi_1, \Sigma) \wedge Rel(\psi_2, \Sigma)$ 
is then equisatisfiable to $\varphi$ modulo oversampling, and $\rho_3$ is indeed an oversampling  of $\rho$ satisfying 
$\kappa$. This shows that while combining formulae $\psi_1, \psi_2$ which are equivalent to formulae $\varphi_1, \varphi_2$
modulo oversampling, we need to relativize $\psi_1, \psi_2$ to obtain a conjunction which 
will be equisatisfiable  to $\varphi_1 \wedge \varphi_2$ modulo oversampling. See \cite{time14} for details.
}
\label{eg-os}
\end{figure}
The oversampling technique is used in the proofs of 
parts \ref{mitl-ureg}.1,
\ref{mitl-ureg}.3 and \ref{mitl-ureg}.4.
\vspace{-.3cm}
 \paragraph*{Equisatisfiable Reduction : $\regmtl$ to $\mtl$}
 Let $\varphi$ be a  $\regmtl$ formula. To obtain equisatisfiable $\mtl$ formula $\psi$, we do the following.
 
  \begin{enumerate}
 % 	\item Convert $\varphi$ to ExNF. 
 	\item  We ``flatten'' the reg modalities to simplify the formulae, eliminating nested 
 	reg modalities. Flattening results in extending the alphabet.
 	Each of the modalities $\reg_I, \ureg$ that appear in the formula $\varphi$ are replaced with fresh witness propositions to obtain a flattened formula.  	  	 For example, 
  	 if $\varphi=\reg_{(0,1)}[a \ureg_{(1,2),\reg_{(0,1)}(a+b)^*}b]$, then flattening 
  	 yields the formula $\varphi_{flat}=w_1 \wedge \wB[w_1 \leftrightarrow \reg_{(0,1)} w_2] \wedge \wB[w_2 \leftrightarrow a \ureg_{(1,2),w_3}b] \wedge \wB[w_3 \leftrightarrow  \reg_{(0,1)}(a+b)^*]$, where $w_1, w_2, w_3$ are fresh witness propositions. 
  	 Let $W$ be the set of fresh witness propositions such that $\Sigma \cap W =\emptyset$. 
  	 	 After flattening, the modalities $\reg_I, \ureg$ appear only in this simplified form as $\wB[w \leftrightarrow \reg_I, \ureg]$. 
  	 	 This simplified appearance of reg modalities are called 
  	 	 \emph{temporal definitions} 
  	 	 %Temporal definitions (referred as $T$ or $T_i$ when there is more than one) are 
  	 	and have the form 
  	$\wB[w \leftrightarrow \reg_I \at]$ or $\wB[w \leftrightarrow x \ureg_{I',\at} y]$, where $\at$ is a rational expression over $\Sigma \cup W_i$, $W_i$ being  the set of fresh witness propositions used in the flattening, 
  	   	 and $I'$ is either a unit length interval or an unbounded interval. 
  	   	
  	   	\item The elimination of reg modalities is achieved by obtaining equisatisfiable $\mtl$ formulae $\psi_i$ 
  	   	over $X_i$, possibly a larger set of propositions
  	   	 than $\Sigma \cup W_i$ corresponding to
  	   	 each temporal definition $T_i$
  	   	 of $\varphi_{flat}$.  Relativizing these $\mtl$ formulae and conjuncting them, we obtain an $\mtl$   
  	   	 formula $\bigwedge_i Rel(\psi_i, \Sigma)$  that is equisatisfiable to $\varphi$ (see Figure \ref{eg-os} for relativization).

%  	   \item 	Consider any temporal definition $T$ of the form  
%  	   	$\wB[w \leftrightarrow \reg_I \at]$ or $\wB[w \leftrightarrow x \ureg_{I',\at} y]$
%  	   	and a timed word $\rho$ over $\Sigma \cup W$. 
%  	   	Each of the rational expression $\at$ has a corresponding minimal DFA recognizing it. To assert $w$ 
%  	at any point of $\rho$, we first check the truth of $\at$ at that point, by running the minimal DFA 
%  	corresponding to $\at$ starting at each point of $\rho$. This requires us to remember certain information 
%  	about the runs of the DFA at each point of $\rho$, giving rise to a \emph{simple extension} $\rho'$ of 
%  	of $\rho$. To mark $w$ correctly, we need to check the truth of $\at$ in a particular time interval $I$ starting at each point 
%  	$i$ of $\rho$. This might require certain time points that are not already present among the time stamps in $dom(\rho)$, leading to \emph{oversampling}. 
%  	We oversample $\rho'$ by introducing some extra points,  and mark a point of $\rho'$ 
%  	   with $w$  iff $\reg_I \at$ or 
%  	   	$x \ureg_{I',\at} y$ holds good at that point. The construction of the simple extension $\rho'$ is right below,  while 
%  	   	 details of the elimination of $\reg_I \at$,  
%  	   	$x \ureg_{I',\at} y$ using oversampling are in the lemmas \ref{elm-reg} and \ref{elm-ureg}.
 	   	  \end{enumerate}
 	   	  The above steps are routine \cite{time14}, \cite{fossacs16}. What remains is to handle the temporal definitions. 
    \vspace{-.2cm}
  \paragraph*{Embedding the Runs of the DFA}
  %Construction of Simple Extension $\rho'$}
  \label{reg-xtnd}
    For any given $\rho$ over $\Sigma \cup W$, where $W$ is the set of witness propositions used in the temporal definitions 
  $T$ of the forms $\wB[w \leftrightarrow \reg_I \at]$ or $\wB[w \leftrightarrow x \ureg_{I',\at} y]$, 
  the rational expression $\at$ has a corresponding minimal DFA recognizing it. We
  define an LTL formula $\mathsf{GOODRUN}(\phi_e)$ which takes a formula $\phi_e$ as a parameter 
  with the following behaviour. $\rho, i \models \mathsf{GOODRUN}(\phi_e)$ iff for all $k >i$, 
  $(\rho, k \models \phi_e) \rightarrow (\rho[i,k] \in L(\at))$. To achieve this,  
%  To assert $w$ at any point of $\rho$, we first check the truth of $\at$ at that point, by running the minimal DFA 
%	corresponding to $\at$ starting at each point of $\rho$.
%  For this, we construct a simple extension $\rho'$ that marks points of $\rho$ with the run information of the minimal DFA 
%  accepting $\at$. 
  we use two new sets of symbols $\mathsf{Threads}$ and $\mathsf{Merge}$ for this information. 
    This results in the extended alphabet   
   $\Sigma \cup W  \cup \mathsf{Threads} \cup \mathsf{Merge}$ for the simple extension $\rho'$ of $\rho$. The behaviour of $\mathsf{Threads}$ and 
   $\mathsf{Merge}$ are explained below. 
  
    Consider $\at=\re(\Ss)$. 
%    The truth of $\at$ at a point depends on sets of formulae 
%	 from $\Ss$, and hence the minimal DFA to check the truth of $\at$ is over $2^{\Ss}$. The rational expression $\at'$
%	 accepted by this DFA is over $2^{\Ss}$ and is equivalent to $\at$ (see Appendix \ref{app:exnf}). 
	Let $\mathcal{A}_{\at} = (Q,2^{\Ss},\delta, q_1, Q_F)$ be the minimal DFA for $\at$ 
	and let $Q = \{q_1,q_2,\ldots,q_m\}$. Let $\In=\{1,2,\dots,m\}$ be the indices of the states.
%	 We mark every point $i$ in $dom(\rho')$ with $w$ or $\neg w$ depending on the truth 
%	 of $\at$ at $i$. To do this, we ``run''  $\mathcal{A}_{\at}$ starting 
%	 from each point $i$ in $dom(\rho')$. At any point $i$ of $dom(\rho')$, we have the 
%	 states reached in $\mathcal{A}_{\at}$ starting from all the $i-1$ prefixes of $\rho'$, and 
%	 we also start a new thread which starts a run at position $i$. 
%Since the number of states of $\mathcal{A}_{\at}$ is a fixed number $m$, we ``merge'' threads $i, j$ if the states reached 
%	 at points $i, j$ are the same, and maintain the information of the merge. 
Conceptually, we consider multiple runs of $\mathcal{A}_\at$ with a new run (new thread) started at each point in $\rho$. $\mathsf{Threads}$ records the state of each previously started run. At each step,  each thread is updated from it previous value according to the transition function $\delta$ of $\mathcal{A}_\at$ and also augmented with a new run in initial state. Potentially, the number of threads would grow unboundedly in size but notice that once two runs are the same state at position $i$ they remain identical in future.  Hence they can be merged into single thread (see 
Figure\ref{fig:runinfo-1}). As a result, $m$ threads suffice. We record whether threads are merged in the current state using variables $\mathsf{Merge}$. An LTL formula records the evolution of $\mathsf{Threads}$ and $\mathsf{Merge}$ over any behaviour $\rho$.  We can define formula $\mathsf{GOODRUN}(\phi_e)$ in LTL  over 
$\mathsf{Threads}$ and $\mathsf{Merge}$.

		\begin{enumerate}
		\item At each position, let  $\Th_i(q_x)$ be a proposition that denotes that the $i$th thread is active and is in state $q_x$, 
		while $\Th_i(\bot)$ be a proposition that denotes that the $i$th thread is not active. 
		The set $\Threads$ consists of propositions $\Th_i(q_x),\Th_i(\bot)$ for $1 \leq i,x \leq m$.
		\item If at a position $e$, we have $\Th_i(q_x)$ and $\Th_j(q_y)$ for $i < j$, and if $\delta(q_x, \sigma_e)=\delta(q_y, \sigma_e)$, 
		then we can merge the threads $i, j$ at position $e+1$. Let $\merge(i,j)$ be a proposition that signifies that threads $i,j$ have been merged. 
		In this case, $\merge(i,j)$ is true at position $e+1$. Let $\Merge$ be the set of all propositions 
		$\merge(i,j)$ for $1 \leq i <j \leq m$. %At most $m$ threads can be running at any point $e$ of the word.
		\end{enumerate}
We now describe the conditions to be checked in $\rho'$. 
    \begin{itemize}
    	\item \textbf{Initial condition}($\varphi_{init}$)- At the first point of the word, we start the first thread and initialize all other threads as $\bot$ : 
    	$\varphi_{init} = ((\Th_1(q_1)) \wedge \bigwedge \limits_{1<i\leq m} \Th_i(\bot))$.
    	\item \textbf{Initiating runs at all points}($\varphi_{start}$)- To check the rational expression within an arbitrary interval, we need to start a new run from every point.
    	     	$\varphi_{start} = \wB(\bigvee\limits_{i \le m} \Th_i(q_1))$
    	\item \textbf{Disallowing Redundancy}($\varphi_{no-red}$)- At any point of the word, if $i < j$ and $\Th_i(q_x)$ and $\Th_j(q_x)$ are both true, $q_x \ne q_y$.
    	 $\varphi_{no-red} = \bigwedge \limits_{x \in \In}\wB[\neg \bigvee \limits_{1\le i < j \le m} (\Th_i(q_x) \wedge \Th_j(q_x)) ]$
  
 \begin{figure}[h]
  \includegraphics[scale=0.3]{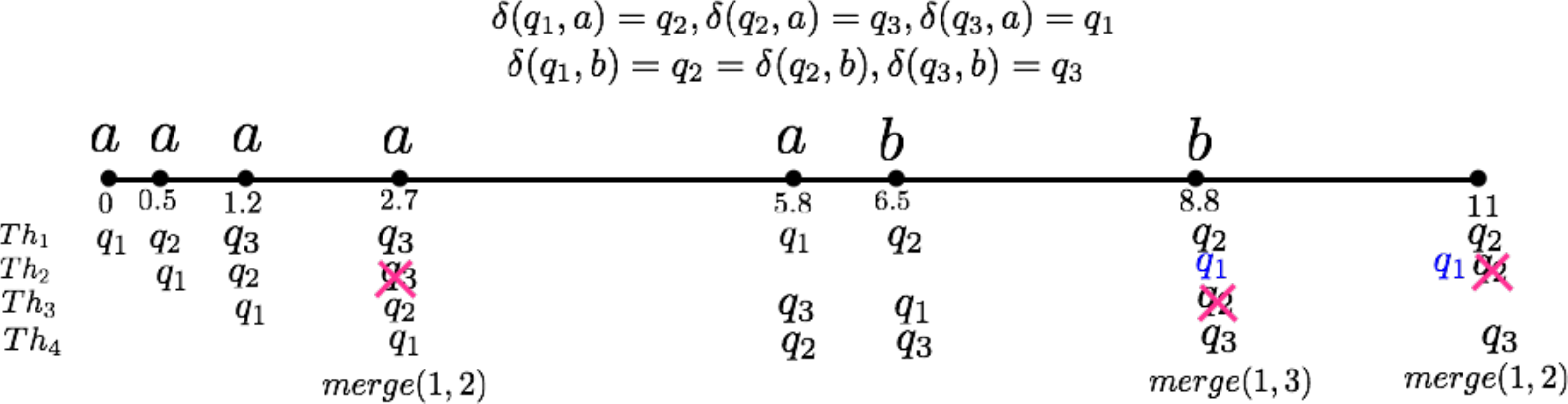} 		
\caption{Depiction of threads and merging. At time point 2.7, thread 2 is merged with 1, since they both 
had the same state information. This thread remains inactive till time  point 8.8,
where it becomes active, by starting a new run in state $q_1$. At time point 8.8, 
thread 3 merges with thread 1, while  at time point 11, thread 2 merges with 1, but is reactivated in state $q_1$.
 }
 \label{fig:runinfo-1}
 \end{figure}

    	\item \textbf{Merging Runs}($\varphi_{\merge}$)- If two different threads $\Th_i,\Th_j (i<j)$  reach  the same state $q_x$ on reading the input at the present point, 
    	then we merge thread $\Th_j$ with $\Th_i$. We remember the merge with the proposition $\merge(i,j)$.
    	 We define a macro $\Next(\Th_i(q_x))$ which is true at a point $e$ if and only if $\Th_i(q_y)$ is true at $e$ and $\delta(q_y, \sigma_e) = q_x$, 
    	 where $\sigma_e \subseteq AP$ is the maximal set of propositions true at  $e$: 
    	% $\Next(\Th_i(q_x))$ is true at $e$ iff thread $\Th_i$  reaches  state $q_x$ after reading the input at $e$.  \\
    	 ${\bigvee \limits_{\{(q_y,prop) \in (Q, 2^{AP}) | \delta (q_y,prop) {=} q_x\}}}[prop {\wedge} {\Th_i(q_y)}]$.

    	    	Let  $\psi(i,j,k,q_x)$ be a formula that says that at the next position, $\Th_i(q_x)$ and $\Th_k(q_x)$ are true for $k>i$, but for all $j<i$, 
    	 $\Th_j(q_x)$ is not. $\psi(i,j,k,q_x)$ is given by \\
    	  	$\Next(\Th_i(q_x)) {\wedge} {\bigwedge \limits_{j<i}} {\neg \Next(\Th_j(q_x))} {\wedge} {\Next(\Th_k(q_x))}$. 
       	 In this case, we merge threads $\Th_i, \Th_k$, and either restart $\Th_k$ in the initial state, or deactivate the $k$th thread at the next position. 
    	 This is given by the formula $\mathsf{NextMerge(i,k)}=\nx[\merge(i,k) \wedge (\Th_k(\bot) \vee \Th_k(q_1)) \wedge \Th_i(q_x)]$.
       	$\varphi_{\merge}=\bigwedge \limits_{x,i,k \in \In \wedge k>i}\wB[\psi(i,j,k,q_x) \rightarrow \mathsf{NextMerge(i,k)}]$.

    	\item \textbf{Propagating runs}($\varphi_{pro},\varphi_{NO-pro}$)- If $\Next(\Th_i(q_x))$ is true at a point, and if for all $j < i$,  $\neg \Next(\Th_j(q_x))$ is true, 
    	 then at the next point, we have $\Th_i(q_x)$.  Let $\mathsf{NextTh}(i,j,q_x)$ denote the formula  $\Next(\Th_i(q_x)) \wedge \neg \Next(\Th_j(q_x))$.
    	     	  The formula 
    	 $\varphi_{pro}$ is given by \\
    	  $ \bigwedge \limits_{i,j \in \In\wedge i<j} \wB[ \mathsf{NextTh}(i,j,q_x){\rightarrow}
    		{\nx [\Th_i(q_x) {\wedge} {\neg \merge(i,j)}]]}$. 
    	    	  If $\Th_i(\bot)$ is true at the current point, then at the next point, either $\Th_i(\bot)$ or $\Th_i(q_1)$. 
    	 The latter condition corresponds to starting a new run on thread $\Th_i$. 
    	     		$\varphi_{NO-pro}{=} \bigwedge \limits_{i\in \In} \wB\{\Th_i(\bot) {\rightarrow} {\nx (\Th_i(\bot) \vee \Th_i(q_1))}\}$
     \end{itemize}
     Let $\mathsf{Run}$ be the formula obtained by conjuncting all formulae explained above.
          Once we construct the simple extension $\rho'$, checking whether the rational expression $\at$ holds in some interval $I$ in the timed word $\rho$, 
          is equivalent to checking that 
%          if a thread $\Th_i$ is at $q_1$ at the first action point in $I$, 
%         then, at the last point in interval $I$, there is some thread $\Th_j$ such that 
%         (i)$\Th_j(q_f)$ is true at this last point, (ii)$\Th_j$ is obtained from $\Th_i$ after some possible merges. 
         if $u$ is the first action point within  $I$,  and if $\Th_i(q_1)$ holds at $u$, then after a series of merges of 
         the form $\merge(i_1,i)$,$\merge(i_2,i_1)$, 
         $\ldots \merge(j,i_n)$, at the last point $v$ in the interval $I$, $\Th_j(q_f)$ is true, 
         for some final state $q_f$. 
         This is encoded as $\mathsf{GOODRUN}(q_f)$. 
         It can be seen that 
         the number of possible sequences of merges are bounded. 
   Figure \ref{fig:runinfo-1}  illustrates the threads and merging.
%       This captures the run information 
%          of $\mathcal{A}_{\at}$, and $\mathsf{Run}$ correctly captures the run information 
%          on $\rho$.   
%         
 We can easily write a 1- $\mathsf{TPTL}$ formula that will check the truth of  $\reg_{[l,u)} \at$ at a point $v$ on the simple extension 
 $\rho'$ (see Appendix \ref{app:1-tptl}). However, to write an $\mtl$ formula that checks the truth of    
 $\reg_{[l,u)} \at$ at a point $v$, we need to oversample $\rho'$ as shown below.  
 \begin{figure}[h]
\includegraphics[scale=0.4]{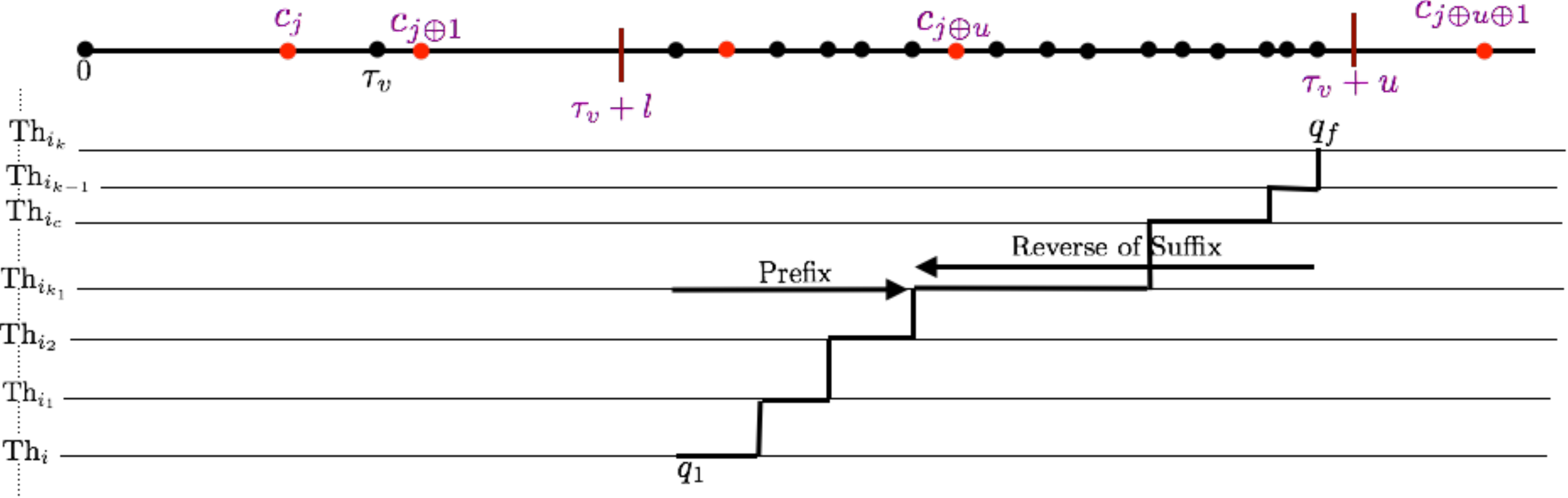}	
\caption{The linking thread at $c_{j \oplus u}$. The points in red are the oversampling integer points, and so are 
$\tau_v+l$ and $\tau_v+u$.}
\end{figure}
\vspace{-.2cm}
	 	\begin{lemma}
 		 		Let $T=\wB[w \leftrightarrow \reg_{I}\at]$  be a temporal definition built from $\Sigma \cup W$. 
 		Then we synthesize a formula $\psi \in \mtl$ 
 		over  $\Sigma \cup W \cup X$ such that 
 		$T$ is equivalent to $\psi$ modulo oversampling.
 \label{elm-reg}	\end{lemma}
 	
 	\begin{proof}
 	Lets first consider the case when the interval $I$ is bounded of the form $[l, u)$.  
Consider a point in $\rho'$ with time stamp $\tau_v$.  To assert  $w$ 
at $\tau_v$, 
we look at the 
%run that starts in some thread $\Th_i$ of the 
first action point after time point 
$\tau_v+l$, and check that $\mathsf{GOODRUN}(last(q_f))$ holds, where $last(q_f)$ identifies the 
%a final state $q_f$ is reached  at the 
last action point 
just before $\tau_v+u$.
% in a thread $\Th_j$, such that $\Th_j$ is obtained from $\Th_i$ after a sequence of merges. 
The first difficulty is the possible absence of time points $\tau_v+l$ and $\tau_v+u$. 
To overcome this difficulty, we oversample $\rho'$ by introducing 
points at times $t+l, t+u$, whenever $t$ is a time point in $\rho'$. These new points are labelled with a new proposition 
$\ovs$.  Sadly, $last(q_f)$ cannot be written in $\mtl$. 

%We can write a nested until formula that captures the possible sequence of merges 
%of threads starting at the first action point after $\tau_v+l$; however, to measure 
%the time elapse and check the presence of $q_f$ in a merged thread at the last action point before $\tau_v+u$ is 
%not possible using nested untils. 
To address this, we introduce new time points at every integer 
point of $\rho'$. The starting point 0 is labelled $c_0$. Consecutive integer time points 
are marked $c_i, c_{i \oplus 1}$, where $\oplus$ is addition modulo 
the maximum constant used in the time interval in the $\regmtl$ formula. This helps in measuring the time elapse 
since the first action point after $\tau_v+l$, till the last action point before $\tau_v+u$ as follows:
if $\tau_v+l$ lies between points marked $c_j, c_{j \oplus 1}$, then 
the last integer point before $\tau_v+u$ is {\bf uniquely} marked $c_{j \oplus u}$. 
\begin{itemize}
\item 	
Anchoring at $\tau_v$, we assert the following at distance $l$:
no action points are seen until the first action point where $\Th_i(q_1)$ is true for some 
thread $\Th_i$. 
% $\Th_i$ remains active until it is merged with some thread $\Th_{i_1}$, and so, 
Consider the next point where  
$c_{j \oplus u}$ is seen. Let $\Th_{i_{k_1}}$ be the thread to which $\Th_i$ 
has merged at the last action point just before $c_{j \oplus u}$. 
Let us call $\Th_{i_{k_1}}$ the ``last merged thread'' before $c_{j \oplus u}$. 
The sequence of merges from $\Th_i$ till $\Th_{i_{k_1}}$ 
asserts a prefix of the run 
that we are looking for between $\tau_v+l$ and $\tau_v+u$. 
%We are sure of the time elapse because of $c_{j\oplus u}$. 
To complete the run %between points $\tau_v+l$ and $\tau_v+u$, 
we mention the  sequence of merges from $\Th_{i_{k_1}}$ 
which culminates in some $\Th_{i_k}(q_f)$ at the last action point before
$\tau_v+u$.  
%such that $\Th_{i_{k_1}}$ merges into $\Th_{i_{k}}$. 
\item 
Anchoring at $\tau_v$, we assert the following at distance $u$:
we see no action points since $\Th_{i_k}(q_f)$ at the action point 
before $\tau_v+u$ for some thread $\Th_{i_k}$, and there is a path linking 
thread 	$\Th_{i_{k_1}}$ to $\Th_{i_k}$ 
%and $\Th_{i_k}$ was active since 
%some previous action point when some other thread $\Th_{i_{k-1}}$ merged into  $\Th_{i_k}$, 
%and 
%so on 
since the point $c_{j \oplus u}$.  We assert that the ``last merged thread'', 
$\Th_{i_{k_1}}$ is active at $c_{j\oplus u}$ :  this is the linking thread which is last merged 
into before $c_{j\oplus u}$, and which is the first thread which merges into another thread after $c_{j\oplus u}$.
\end{itemize}
These two formulae thus ``stitch'' the actual run observed between points $\tau_v+l$ and $\tau_v+u$. The formal technical details can be seen 
in Appendix \ref{app:pref-suf}.
If $I$ was an unbounded interval of the form $[l, \infty)$, then
we will go all the way till the end of the word, and assert $\Th_{i_k}(q_f)$ at the last action point of the word. 
Thus, for unbounded intervals, we do not need any oversampling at integer points.  
 \end{proof}
In a similar manner, we can eliminate the $\ureg$ modality, the proof of which can be found in Appendix \ref{app:ureg}. 
 If we choose to work on logic 	$\mitl+\ureg$, we obtain a $\mathsf{2EXPSPACE}$ upper bound for satisfiability checking, 
 since elimination of $\ureg$ results in an equisatisfiable $\mathsf{MITL}$ formula. This is an interesting consequence 
 	of the oversampling technique; without oversampling,  	we can eliminate 
  $\ureg$ obtaining 1-$\mathsf{TPTL}$ (Appendix \ref{app:1-tptl}). However, 1-$\mathsf{TPTL}$ does not enjoy 
  the benefits of non-punctuality, and is  non-primitive recursive (Appendix \ref{app:1-tptl-hard}).

\section{Automaton-Metric Temporal Logic-Freeze Logic Equivalences}
\label{a-il-fl}
The focus of this section is to obtain equivalences between automata,
temporal and freeze logics. 
First of all, we identify a fragment of $\regmtl$ denoted 
$\sfmtl$, where the rational  
expressions in the formulae are all star-free.
We then show the equivalence between $\po$-1-clock ATA, 1${-}\tptl$,
 and $\sfmtl$ ($\po$-1-clock ATA $\subseteq \sfmtl \subseteq 1{-}\tptl \equiv \po$-1-clock ATA).
 The main result of this section gives a tight automaton-logic connection in Theorem \ref{thm:tptl-1-mtl}, and is  
proved using Lemmas \ref{aut-tptl-1}, \ref{lem:sf-1tptl} and 
\ref{lem:poata-sfmtl}. 
\begin{theorem}
\label{thm:tptl-1-mtl}
1${-}\tptl$, $\sfmtl$ and $\po$-1-clock ATA are all equivalent. 
\end{theorem}

%\paragraph*{Automaton-Freeze Logic ($\po$-1-clock ATA $\equiv 1{-}\tptl$)}
%\label{a-fl}
\noindent We first show that partially ordered 1-clock alternating timed automata ($\po$-1-clock ATA) capture exactly the same class of languages as $1{-}\tptl$. We also show that $1{-}\tptl$ is equivalent to the subclass $\sfmtl$  of $\regmtl$ where the rational expressions 
$\re$ involved in the formulae are such that $L(\re)$ is star-free.  %This also shows for the first time in pointwise timed logics, an equivalence between 
 %freeze point logics and logics with interval constraints. 
 
A 1-clock ATA \cite{Ouaknine05} is a tuple $\mathcal{A}=(\Sigma, S, s_0, F, \delta)$, where 
$\Sigma$ is a finite alphabet, $S$ is a finite set of locations,   
  $s_0 \in S$  is the initial location and $F \subseteq S$ is the set 
  of final locations.   Let $x$ denote the clock variable in the 1-clock ATA, and $x \bowtie c$ denote a clock constraint 
  where $c \in \mathbb{N}$ and $\bowtie \in \{<, \leq, >, \geq \}$. Let $X$ denote a finite set 
  of clock constraints of the form $x \bowtie c$. The transition function is defined as   
   $\delta: S \times \Sigma \rightarrow \Phi(S \cup \Sigma \cup X)$  where
  $\Phi(S \cup \Sigma \cup X)$ is a set of formulae defined by the grammar  $\varphi::=\top|\bot|\varphi_1 \wedge \varphi_2|\varphi_1 \vee \varphi_2|s|x \bowtie c|x.\varphi$ where  
    $s \in S$, and $x. \varphi$ is a binding construct corresponding to resetting the clock $x$ to 0. 
   
      The notation $\Phi(S \cup \Sigma \cup X)$ thus allows boolean combinations 
   as defined above of locations, symbols of $\Sigma$, clock constraints and $\top, \bot$, with or without  
   the binding construct $(x.)$. A configuration of a 1-clock ATA is a set consisting of locations along with their clock valuation. 
   Given a configuration $C$, 
   we denote by $\delta(C,a)$ the configuration $D$ obtained by applying 
   $\delta(s,a)$ to each location $s$ such that $(s, \nu) \in C$. 
      A run of the 1-clock ATA starts from the initial configuration 
   $\{(s_0,0)\}$, and proceeds with alternating time elapse transitions and  
   discrete transitions obtained on  reading a symbol from $\Sigma$. A configuration is accepting iff it is either empty, or 
   is of the form $\{(s, \nu) \mid s \in F\}$.   
    The language accepted  by a 1-clock ATA $\mathcal{A}$, denoted 
    $L(\mathcal{A})$ is the set of all timed words $\rho$ such that starting 
    from $\{(s_0,0)\}$, reading $\rho$ leads to an accepting configuration.  
       A $\po$-1-clock ATA is one in which 
  % \begin{itemize}
   (i) there is a  partial order denoted $\prec$ on the locations, such that whenever 
   $s_j$ appears  in $\Phi(s_i)$, $s_j \prec s_i$, or $s_j=s_i$.    
   Let $\downarrow s_i=\{s_j \mid s_j \prec s_i\}$, (ii) $x.s$ does not appear in $\delta(s,a)$ for all $s \in S, a \in \Sigma$. 
      % \end{itemize}
   
% \noindent{\em{Example}}. 
 \begin{example}
 \label{eg1}
 Consider  the $\po$-1-clock ATA $\mathcal{A}=(\{a,b\}, \{s_0,s_a,s_{\ell}\}, s_0, \{s_0, s_{\ell}\},\delta)$ with transitions 
 $\delta(s_0,b)=s_0,\delta(s_0,a)=(s_0 \wedge x.s_a) \vee s_{\ell},$
 $\delta(s_a,a)=(s_a \wedge x<1) \vee (x>1)=\delta(s_a,b),$ and 
 $\delta(s_{\ell},b)=s_{\ell}, \delta(s_{\ell},a)=\bot$.       
  The automaton accepts all strings where every non-last $a$ 
  has no symbols at distance 1 from it, and has some symbol 
  at distance $>1$ from it. 
 \end{example}

  \begin{lemma}
 $\po$-1-clock ATA and $1{-}\tptl$ are equivalent in expressive power. 
\label{aut-tptl-1}
 \end{lemma}
The translation from $1{-}\tptl$ to $\po$-1-clock ATA is easy, as in the 
translation from $\mtl$ to  $\po$-1-clock ATA. For the reverse direction, 
we start from the lowest location (say $s$) in the partial order, and
 replace the transitions of $s$ by a 1-$\tptl$ formula that  models timed words which are accepted, 
when started in $s$. The accepting behaviours 
of each location $s$, denoted $\Beh(s)$ is computed bottom up. The 1-$\tptl$ formula 
that we are looking for is $\Beh(s_0)$ where $s_0$ is the initial location.  
In example \ref{eg1}, $\Beh(s_{\ell})=\wB b$, 
$\Beh(s_a){=}(x <1) \wU  (x>1)$, 
$\Beh(s_0)=[(a \wedge x.\nex \Beh(s_a)) \vee b]  \weaku (a \wedge \nex \Beh(s_{\ell}))$
=$((a \wedge (x.\nex [(x <1) \wU  x>1])) \vee b) \weaku (a \wedge \nex \wB b).$
Step by step details for Lemma \ref{aut-tptl-1} can be seen in Appendix \ref{app:ata-tptl-1}.

We next show that starting from a $\sfmtl$ formula $\varphi$, we can construct 
an equivalent $1{-}\tptl$ formula $\psi$.  
The proof of Lemma \ref{lem:sf-1tptl} can be found in Appendix \ref{app:sf-1tptl}. 
\begin{lemma}
\label{lem:sf-1tptl}
				$\sfmtl \subseteq 1-\tptl$
				\end{lemma}
The idea is to iteratively keep replacing the $\reg$ modality level by level, starting with the innermost one, 
until we have eliminated the topmost one.  
Consider the $\sfmtl$ formula $\varphi=\reg_{(0,1)}[\reg_{(1,2)}(a+b)^*]$. 
To eliminate 
 $\reg_{(1,2)}(a+b)^*$ at a point, we freeze a clock $x$, and wait till $x \in (1,2)$, 
 and  check $(a+b)^*$ on this region. The LTL formula for $(a+b)^*$ is $\Box(a \vee b)$.
Treating $x \in (1,2)$ as a proposition, we obtain 
  $\zeta=x.(x \notin(1,2) \until [\psi_1 \wedge \psi_2])$ where 
  $\psi_1=x \in(1,2) \wedge [x \in (1,2) \until \nex(\Box(x \notin (1,2)))]$ and 
  $\psi_2=  \Box[x \in(1,2) \rightarrow (a \vee b)]$.
%  \Box\{(a \wedge x \in(1,2)) \vee (b \wedge x \in(1,2))\}
%(x \in (1,2) \until \nex(\Box(x \notin (1,2))))]]$ 
$\zeta$ asserts $\Box(a \vee b)$
exactly on the region (1,2), eliminating the modality $\reg_{(1,2)}$. To eliminate 
the outer $\reg_{(0,1)}$, we 
assert the existence of a point in (0,1) where $\reg_{(1,2)}(a+b)^*$ is true by saying 
$x.(x \notin(0,1) \until ([x \in(0,1) \wedge \zeta \wedge \nex(\Box(x \notin(0,1)))]))$.
This is   1-$\tptl$ equivalent to $\varphi$.

\begin{lemma}($\po$-1-clock ATA to $\sfmtl$)
\label{lem:poata-sfmtl}
	Given a 	$\po$-1-clock ATA $\mathcal{A}$, we can construct a $\sfmtl$ formula 
	$\varphi$ such that $L(\mathcal{A})=L(\varphi)$.				\end{lemma}
\begin{proof}(Sketch)
	We give a proof sketch here, a detailed proof can be found in Appendix \ref{app:ata-1tptl}.
%\noindent{\bf The Main Idea}:
Let $\mathcal{A}$ be a $\po$-1-clock ATA with locations $S=\{s_0,s_1,  \dots,s_n\}$. 
Let $K$ be the maximal constant used in the guards $x \sim c$ occurring in the transitions. 
Let $R_{2i}=[i,i], R_{2i+1}=(i, i+1), 0 \leq i < K$ and $R^+_K=(K, \infty)$ be the regions 
$\mathcal{R}$ of $x$. Let $R_h \prec R_k$ denote that region $R_h$
precedes region $R_k$.
For each location $s$, $\Beh(s)$ as computed in Lemma \ref{aut-tptl-1} is a 1-$\tptl$ formula that gives the timed behaviour starting at $s$, 
using constraints $x \sim c$ since the point where $x$ was frozen. 
In example \ref{eg1}, $\Beh(s_a){=}(x <1) \wU  (x>1)$, allows symbols $a,b$ as long as $x<1$ 
keeping the control in $s_a$, has no behaviour at $x=1$, and allows 
 control to leave $s_a$ when $x>1$.
 For any $s$, we ``distribute'' $\Beh(s)$ across regions by untiming it. In example \ref{eg1}, 
  $\Beh(s_a)$ is $\wB(a \vee b)$ for regions $R_0, R_1$, it is $\bot$ for $R_2$ and 
  is $(a \vee b)$ for  $R^+_1$.  Given any $\Beh(s)$, and a pair of regions $R_j \preceq R_k$, such that
   $s$ has a non-empty behaviour in region $R_j$, and control leaves $s$ in $R_k$,
    the untimed behaviour of $s$ between regions $R_j, \dots, R_k$ is written as LTL formulae 
   $\varphi_j, \dots, \varphi_k$. This results in a ``behaviour description'' (or $\BD$ for short)
    denoted $\BD(s, R_j,R_k)$ : this is a $2K+1$ tuple with $\BD[R_l]=\varphi_l$
    for $j \leq l \leq k$, and $\BD[R]=\top$ denoting ``dont care'' for the other regions.
   Let $\BDset(s)$ denote the set of all $\BD$s for a location $s$.  
    For the initial location $s_0$, consider all $\BD(s_0, R_j, R_k) \in \BDset(s_0)$
that have a behaviour starting in $R_j$, and ends in an accepting configuration in $R_k$. 
        Each LTL formula $\BD(s_0, R_j,R_k)[R_i]$ (or $\BD[R_i]$ when $s, R_j, R_k$ are clear)
        is replaced with a star-free rational expression
    denoted $\re(\BD(s_0, R_j,R_k)[R_i])$. Then  $\BD(s_0, R_j,R_k)$ 
    is transformed into a $\sfmtl$ formula $\varphi(s_0,R_j,R_k)= \bigwedge_{j \leq g \leq k}\reg_{R_g}\re(\BD(s_0, R_j, R_k)[R_g])$. 
The language accepted by the $\po$-1-clock ATA $\mathcal{A}$ is then given by  
$\bigvee_{0 \leq j \leq  k \leq 2K} \varphi(s_0,R_j, R_k)$.

%   
%    
%    
%    
%    
%    by writing 
%LTL formulae $\varphi_j, \dots, \varphi_k$. This gives rise to a ``behaviour description'' $\BD$, for each $\Beh(s_i)$ and a pair of regions 
%$R_j, R_k$. A $\BD$ is a $2K+1$ tuple consisting of untimed behaviours in all regions. 
%The $\BD$ of $s$ starting in region $R_j$ 
%and leaving $s$ in region  $R_k$ denoted $\Beh(s,R_j, R_k)$, is a $2K+1$ tuple of LTL formulae
%where (i) the LTL formulae 
%for regions $R_0, \dots, R_{j-1}$ as well as $R_{k+1}, \dots, R_{2K}$ are $\top$, representing ``dont care'', 
%(ii)$\Beh(s)$ allows a transition from $s$ in region $R_i$, and (iii) $\Beh(s)$ allows control to leave $s$ 
%in region $R_k$. The LTL formula in regions $R_j,\dots, R_k$ is obtained as explained below. 
% Each LTL formula in the $\BD$ can be written as a star-free expression. 
%   Let $\Beh(s, R_j, R_k)[R_i]$ (or $BD[R_i]$ when $s,R_j, R_k$ are clear) represent the LTL formula in the region $R_i$ 
%  and let   
%   $\re(\Beh(s, R_j, R_k)[R_i])$ represent the star-free rational expression
%  that is equivalent to this LTL formula. 
% The $\BD$ $\Beh(s,R_j, R_k)$
%is transformed into a $\sfmtl$ formula $\varphi(s,R_j,R_k)= \bigwedge_{j \leq g \leq k}\reg_{R_g}\re(\Beh(s, R_j, R_k)[R_g])$. 
%The language accepted by the $\po$-1-clock ATA $\mathcal{A}$ is then given by  
%$\bigvee_{0 \leq j \leq  k \leq 2K} \varphi(s_0,R_j, R_k)$ where $s_0$ is the initial location, and 
%the word is accepted while in region $R_k$.   This disjunction allows all possible accepting behaviours 
%from the initial location $s_0$. 

\begin{figure}
\includegraphics[scale=0.45]{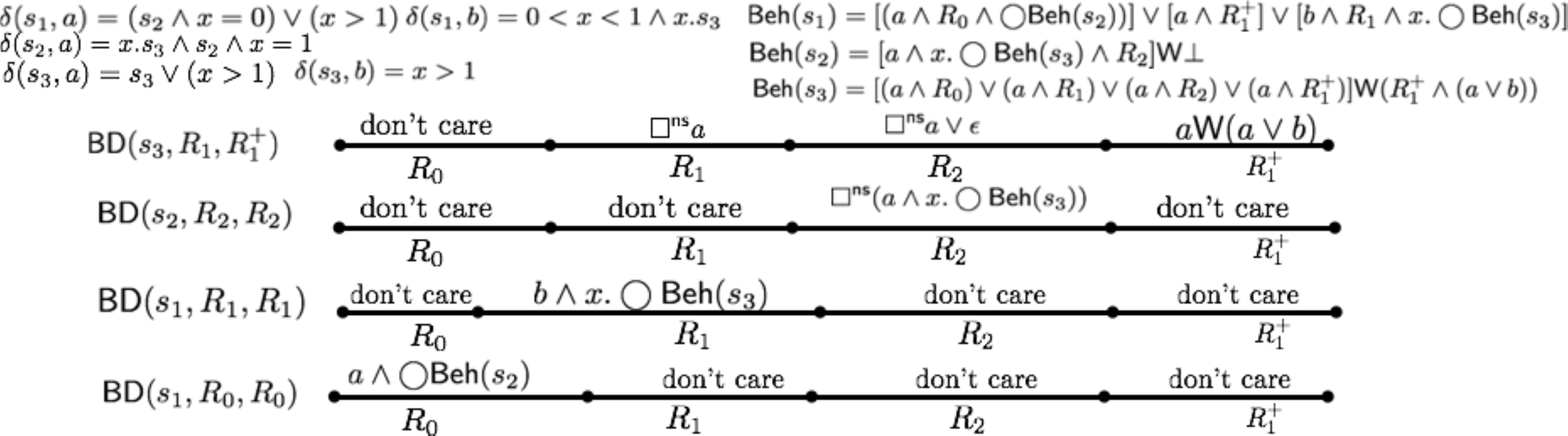}	
\caption{A $\po$-1-clock ATA with initial location $s_1$ and  $s_2,s_3$ are accepting.}
\label{last-eg}
\end{figure}

\noindent{\bf Computing $\BD(s,R_i, R_j)$ for a location $s$ and pair of regions $R_i \preceq R_j$}. 
We first compute  $\BD(s,R_i, R_j)$ for locations $s$ which are lowest in the partial order, followed 
by computing $\BD(s',R_i, R_j)$ for locations $s'$ which are higher in the order.   
For any location $s$, $\Beh(s)$ has the form $\varphi$ or $\varphi_1 \weaku \varphi_2$ 
or $\varphi_1 \wU \varphi_2$, where $\varphi,\varphi_1, \varphi_2$ 
are disjunctions of conjunctions over $\Phi(S \cup \Sigma \cup X)$, where $S$ is the set of locations 
with or without the binding construct $x.$, and $X$ is a set of clock constraints 
of the form $x \sim c$. Each conjunct 
has the form $\psi \wedge x \in R$  where $\psi \in \Phi(\Sigma \cup S)$ and $R \in \mathcal{R}$. 
Let $\varphi_1=\bigvee(P_i \wedge C_i), \varphi_2=\bigvee(Q_j \wedge E_j)$ where $P_i,Q_j \in \Phi(\Sigma \cup S)$ and $C_i,E_j \in \mathcal{R}$.
Let $\mathcal{C}$ and $\mathcal{E}$ be shorthands  for any $C_k, E_l$.  

If $\Beh(s)$ is an expression without $\until, \mathsf{W}$ (the case of $\varphi$ above), 
then $\BD(s, R_i, R_i)$ is defined for a region $R_i$ if $\varphi=\bigvee(Q_j \wedge E_j)$ 
and there is some $E_l$ with $x \in R_i$.  It is a $2K+1$ tuple with $\BD(s, R_i, R_i)[R_i]=Q_l$, and 
the rest of the entries are $\top$ (for dont care). If $\Beh(s)$ has the form 
$\varphi_1 \weaku \varphi_2$ or $\varphi_1 \wU \varphi_2$, then 
 for $R_i \preceq R_j$, and a location $s$, $\BD(s,R_i,R_j)$ 
%is empty if $\Beh(s)$ has no constraint $x \in R_i$ occurring in $\mathcal{C}, \mathcal{E}$. 
%Otherwise, it is non-empty and 
%is computed as follows: (%(1) there is a  constraint $x \in R_i$ in $\mathcal{C}$ or $\mathcal{E}$ (this allows us to start observing the behaviour in region $R_i$) 
%(2) there must be a constraint $x \in R_j$ in some $\mathcal{E}$ (this allows us to exit the control 
% location $s$ while in region $R_j$).  
is a $2K+1$ tuple with 
  (i) formula $\top$ in regions $R_0, \dots, R_{i-1},R_{j+1}, \dots, R_{K}^+$, 
    (ii) If $C_k=E_l=(x \in R_j)$ for some $C_k, E_l$, then the LTL formula in region $R_j$ is $P_k \until Q_l$ if $s$ is not  accepting, 
  and is $P_k \weaku Q_l$ if $s$ is  accepting, (iii) If no $C_k$ is equal to any $E_l$, and  
  if $E_l=(x \in R_j)$ for some $l$, then the formula in region $R_j$ is $Q_l$. If $C_m=(x \in R_i)$ for some $m$, then 
  the formula for region $R_i$ is $\wB P_m$. If there is some $C_h=(x \in R_w)$ for $i < w < j$, then 
  the formula in region $R_w$ is $\wB P_h \vee \epsilon$, where $\epsilon$ signifies that 
  there may be no points in regions $R_w$. If there are no $C_m$'s such that $C_m=(x \in R_w)$  
for $R_i \prec R_w \prec R_j$, then the formula in region $R_w$ is $\epsilon$. 
$\epsilon$ is used as a \emph{special symbol} in LTL whenever there is no behaviour 
in a region.

 \noindent{\bf $\BD(s,R_i, R_j)$ for location $s$ lowest in po}.  Let $s$ be a location that is lowest in the partial order. 
% The locations $s_{\ell},s_a$ in Example \ref{eg1} are lowest in the partial order, and 
%$\Beh(s_{\ell})=b \weaku \bot=\wB b$, $\Beh(s_a){=}[(a \vee b) \wedge (x<1)]\wU [(a \vee b) \wedge (x>1)] $.
  In general, if $s$ is the lowest in the partial order, then  
   $\Beh(s)$ has the form $\varphi_1 \weaku \varphi_2$ 
or $\varphi_1 \wU \varphi_2$ or $\varphi$ where $\varphi, \varphi_1, \varphi_2$ 
are disjunctions of conjunctions over $\Phi(\Sigma \cup  X)$. Each conjunct 
has the form $\psi \wedge x \in R$  where $\psi \in \Phi(\Sigma)$ and $R \in \mathcal{R}$. 
See Figure \ref{last-eg}, with regions  $R_0, R_1, R_2, R^+_1$, and some example $\BD$s. 
% In example \ref{eg1}, the regions are $R_0=[0,0], R_1=(0,1), R_2=[1,1], R^+_1=(1, \infty)$.
%$\Beh(s_{\ell}, R_1, R^+_1)=(\top, \wB b,\wB b \vee \wB \bot,b \weaku \bot)$, and
%  $\Beh(s_a, R_0, R^+_1)=(\wB(a \vee b), \wB(a \vee b)\vee \wB \bot, \bot, (a \vee b))$. 
%  
In Figure \ref{last-eg},
using the $\BD$s of the lowest location $s_3$, we write 
the $\sfmtl$ formula for $\Beh(s_3)$ : 
 $\psi(s_3)=\varphi_{R_0}(s_3) \wedge  \varphi_{R_1}(s_3)
 \wedge \varphi_{R_2}(s_3) \wedge \varphi_{R^+_1}(s_3)$, where 
 each $\varphi_{R}$ 
 describes the behaviour of $s_3$ 
 starting from region $R$. 
For a fixed region $R_i$, $\varphi_{R_i}(s_3)$ is  
$\bigwedge_{R_g \prec R_i} \reg_{R_g} \epsilon \wedge \reg_{R_i}\Sigma^+ \rightarrow
\{\bigvee_{R_i \prec R_j} \varphi(s_3, R_i, R_j)\}$, where 
$\varphi(s_3, R_i, R_j)$ is described above. $\reg_{R_g} \epsilon$ means that 
there is no behaviour in $R_g$.
$\varphi_{R_0}(s_3)$ is given by 
$\reg_{R_0} \Sigma^+ \rightarrow
\{(\reg_{R_0}a^* \wedge \reg_{R_1}[a^*+\epsilon] \wedge \reg_{R_2}[a^*+\epsilon] \wedge \reg_{R^+_1}[a^*+a^*b])\}$. 
%Using the $\BD$s of $s_a$, we can write the $\sfmtl$ formula that describes the behaviour of $s_a$.
%This fomula is given by $\psi(s_a)=\varphi_{R_0}(s_a) \wedge  \varphi_{R_1}(s_a)
% \wedge \varphi_{R_2}(s_a) \wedge \varphi_{R^+_1}(s_a)$, where 
% each $\varphi_{R_i}$ 
% describes the behaviour starting from region $R_i$, while in location $s_a$. 
%For a fixed region $R_i$, $\varphi_{R_i}(s_a)$ is  
%$\bigwedge_{R_g \prec R_i} \reg_{R_g} \emptyset \wedge \reg_{R_i}\Sigma^+ \rightarrow
%\{\bigvee_{R_i \prec R_j} \varphi(s_a, R_i, R_j)\}$, where 
%$\varphi(s_a, R_i, R_j)$ is described above. 
% Recall that $\varphi(s_a, R_i, R_j)$ describes a possible behaviour of $s_a$ that starts at $R_i$ and ends in $R_j$.  
%For instance,
% $\varphi_{R_0}(s_a)$ is
%$\reg_{R_0} \Sigma^+ \rightarrow
%\{(\reg_{R_0}(a+b)^* \wedge \reg_{R_1}[(a+b)^*+\emptyset] \wedge \reg_{R_2}\emptyset \wedge \reg_{R^+_1}(a+b)^*)\}$
%while 
% $\varphi_{R_1}(s_a)$ is
%$\reg_{R_0} \emptyset \wedge \reg_{R_1}\Sigma^+ \rightarrow
%\{(\reg_{R_1}(a+b)^* \wedge \reg_{R_2}\emptyset \wedge \reg_{R^+_1}(a+b)^*)\}$. 
%Similarly,   $\varphi_{R_2}(s_a)$ is empty since $s_a$ has no behaviour in $R_2$.
%Finally,  $\varphi_{R^+_1}(s_a)$ is 
%$\bigwedge_{R_g \prec R^+_1} \reg_{R_g} \emptyset \wedge \reg_{R^+_1}\Sigma^+ \rightarrow
%\reg_{R^+_1}(a+b)^*$. In a similar manner, we can write the $\sfmtl$ formula $\psi_{s_{\ell}}$ that describes the behaviour 
%of $s_{\ell}$ across regions.   

 \noindent{\bf $\BD(s,R_i, R_j)$ for a location $s$ which is higher up}. 
 If $s$ is not the lowest in the partial order, then $\Beh(s)$ 
 can have locations  $s' \in \downarrow s$. $s'$ occurs as $\nex(s')$ or  $x.\nex(s')$ in 
 $\Beh(s)$.  
 For $x.\nex \Beh(s_3)$ in $\BD(s, R_i, R_j)$, since the clock is frozen, we  plug-in the $\sfmtl$ formula $\psi(s_3)$ computed above 
for $x.\nex \Beh(s_3)$ in $\BD(s_1, R_i, R_j)$.
For instance, in figure \ref{last-eg}, $x.\nex \Beh(s_3)$
appears in $\BD(s_2, R_2, R_2)[R_2]$. 
 We simply plug in the $\sfmtl$ formula $\psi(s_3)$ in its place.
Likewise, for locations $s, t$, if $\nex \Beh(t)$ occurs in  
$\BD(s, R_i, R_j)[R_k]$, 
we look up $\BD(t, R_k, R_l) \in \BDset(t)$ for all $R_k \preceq R_l$
and \emph{combine} $\BD(s, R_i, R_j), \BD(t, R_k, R_l)$ 
in a manner described below. This is done to detect 
if the ``next point'' for $t$ has a behaviour in $R_k$ or later.
 \begin{itemize}
\item[(a)]If the next point for $t$ is in $R_k$ itself, then we \emph{combine}
$\BD_1=\BD(s, R_i, R_j)$ with  
  $\BD_2 \in \{\BD(t, R_k, R_l)  \mid R_k \preceq R_l\} \subseteq \BDset(t)$  as follows.
$\mathsf{combine}(\BD_1, \BD_2)$ results in $\BD_3$ such that 
$\BD_3[R]=\BD_1[R]$ for $R  \prec R_k$, 
$\BD_3[R]=\BD_1[R] \wedge \BD_2[R]$ for $R_k \prec R$, where $\wedge$ denotes component wise conjunction. 
$\BD_3[R_k]$ is obtained by replacing $\nex \Beh(s_2)$ in $\BD_1[R_k]$
with $\BD_2[R_k]$. Doing so enables the next point 
 in $R_k$, emulating the behaviour of $t$ in $R_k$.
%For instance, we replace $\nex \Beh(s_2)$ in $\BD(s_1,R_1,R^+_1)[R_1]$ with $\bot$, since 
%$s_2$ has no behaviour in $R_1$, while 
%we replace $\nex \Beh(s_2)$ in $\BD(s_1,R_1,R^+_1)[R_2]$
%with $\nex(\BD(s_2,R_2,R_l)[R_2])$ ($R_2 \preceq R_l$). 
	\item[(b)] Assume the next point for $t$ lies in $R_b$, $R_k \prec R_b$.
The difference with case (a) is that we combine 
$\BD_1=\BD(s, R_i, R_j)$ with $\BD_2 {\in }\{\BD(t, R_b, R_l) \mid R_k \prec R_b \preceq R_l\} \subseteq \BDset(t)$. 
Then $\mathsf{combine}(\BD_1, \BD_2)$ results in a $\BD$, say $\BD_3$ such that 
$\BD_3[R]=\BD_1[R]$ for $R  \prec R_k$, 
$\BD_3[R]=\BD_1[R] \wedge \BD_2[R]$ for all $R_b \preceq R$, and 
$\BD_3[R]=\epsilon$ for $R_k \prec R \prec R_b$. The $\nex \Beh(t)$ 
in $\BD_1[R_k]$ is replaced with $\Box \bot$ to signify that 
the next point is not enabled for $t$.
 	 	 	See Figure \ref{last-2} where $R_b=R_2$. The conjunction 
 	 	with $\Box \bot$ in $R_0$ signifies that the next point for $s_2$ is not in $R_0$; 
 	 	 the $\epsilon$ in $R_1$ signifies that
there are no points in $R_1$ for $s_2$.  Conjuncting $\Box \bot$ 
in a region signifies that the next point 
does not lie in this region. 
%
% 
%
% Assume the next point with respect to $\Beh(s_{\ell})$ lies in $R_2$.
% Based on this, we combine any $\BD$ of $\Beh(s_{\ell})$, say,  $\BD_2=\Beh(s_{\ell}, R_2, R_j)$ ($R_2 \preceq R_j$)
% such that $\BD_2[R_2]=\wB b$. We combine 
% $\BD_2$ from  region $R_2$ or later,   
%with that of  $\BD_1=\Beh(s_0, R_1, R_2)$ from $R_2$ onwards. This combination of $\BD_1, \BD_2$ results in a new $\BD$ say $\BD_3$.
%$\BD_3[R_0]=\BD_1[R_0],\BD_3[R_1]=\BD_1[R_1]$, keeping the entries in $R_0, R_1$ unchanged.
%For regions $R_i$ such that $R_2 \prec R_i$, 
% $\BD_1[R_i]$ is conjuncted with $\BD_2[R_i]$ to obtain $\BD_3[R_i]$.
%$\BD_3[R_2]$ is obtained by replacing $\nex \Beh(s_{\ell})$ in $\BD_1[R_2]$ 
%with the LTL formula  $\BD_2[R_2]= \wB b$.  
% In the case of $s_{\ell}$, $\BD_2[R^+_1]$ 
% is $b \weaku \bot$ (if $R_2 \prec R_j$) or $\top$ (if $R_j=R_2$). 
%We hence obtain the $\BD$s 
%$(\top, \wB[(a \wedge \psi) \vee b], (a \wedge \psi) \vee b \weaku (a \wedge 
%\wB b), \top)$, 
%$(\top, \wB[(a \wedge \psi) \vee b], (a \wedge \psi) \vee b \weaku (a \wedge 
%\wB b), b \weaku \bot)$ by combining  $\BD_1$ with $\BD_2$, and $\psi=\psi_{s_a}$ is the $\regmtl$ formula 
%for $x.\nex F(s_a)$. 
\begin{figure}[h]
\includegraphics[scale=0.45]{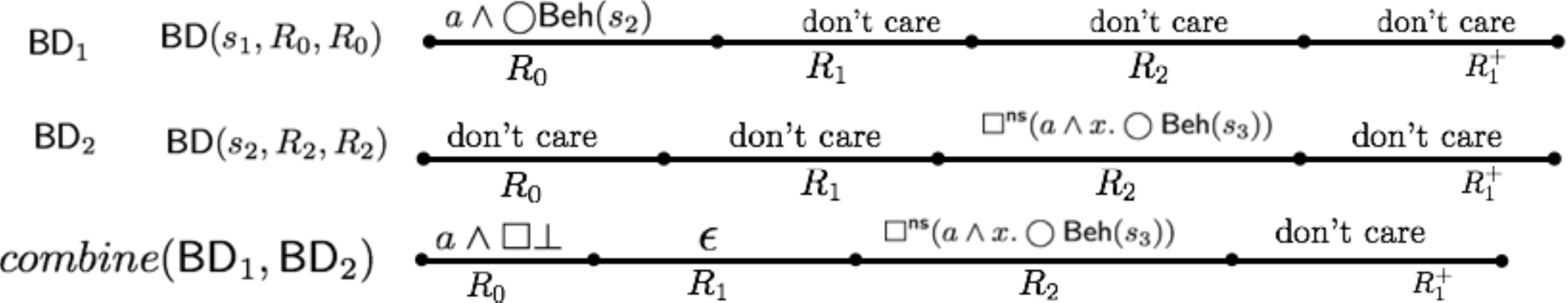}
\caption{Combining $\BD$s}	
\label{last-2}
\end{figure}
 % 	
% Then we combine $\BD_1$ as in (a) above with some $\BD_2=\Beh(s_{\ell}, R_2, R_j)$	($R_2 \preceq R_j$)
% such that $\BD_2[R_2]=\wB \bot$. 
%  $\BD_3[R_0]=\BD_1[R_0],\BD_3[R_1]=\BD_1[R_1]$, as above,
% % $\Beh(s_0, R_1, R_2)[R_0],\Beh(s_0, R_1, R_2)[R_1]$ are unchanged,
% while  $\BD_3[R_i]=\BD_1[R_i] \wedge \wB \bot=\wB \bot$ for $R_2 \preceq R_i \prec R_b$.  
% Also, $\BD_3[R_c]=\BD_1[R_c]\wedge \BD_2[R_c]$ for  $R_b \preceq R_c$. 
%% the entries $\Beh(s_0, R_1, R_2)[R_c]$ for 
%% are conjuncted with  $\Beh(s_{\ell}, R_2, R_j)[R_c]$.
%$\BD_2[R_2]=\wB \bot$, $\BD_2[R^+_1]=b \weaku \bot$.
% We obtain 
%$\BD_3=(\top, \wB[(a \wedge \psi) \vee b], \wB \bot, b \weaku \bot)$ by combining 
%$\BD_1, \BD_2$  on regions $R_2, R^+_1$. 
\end{itemize}
We look at the ``accepting'' $\BD$s in $\BDset(s_0)$, viz., all $\BD(s_0, R_j, R_k)$, 
such that acceptance happens in $R_k$, and $s_0$ has a behaviour starting in $R_j$.
The LTL formulae $\BD(s_0, R_j, R_k)[R]$ are replaced with 
star-free expressions $\re(\BD(s_0, R_j, R_k)[R])$. 
Each accepting $\BD(s_0, R_j, R_k)$ gives an $\sfmtl$ formula 
$\bigwedge_{R_j \preceq R \preceq R_k} \reg_R \re(\BD(s_0, R_j, R_k)[R])$. The disjunction of these across all accepting $\BD$s
is the $\sfmtl$ formula equivalent to $L(\mathcal{A})$. 
\end{proof}

\section{Discussion}
 We propose  $\regmtl$ which significantly increases the expressive power of $\mtl$
   and yet retains decidability over pointwise finite
 words. 
 The $\reg$ operator added to $\mtl$ syntactically subsumes several other modalities in literature including threshold counting, modulo counting 
 and the pnueli modality. Decidability of $\regmtl$ 
 is  proved by giving an  equisatisfiable reduction  to $\mtl$ 
 using oversampled temporal projections. This reduction has elementary complexity and allows us to identify two
 fragments of $\regmtl$ with $\mathsf{2EXPSPACE}$ and $\mathsf{EXPSPACE}$ satisfibility. 
 In previous work \cite{time14}, oversampled temporal projections were used to 
 reduce $\mtl$ with punctual future and non-punctual past to $\mtl$. Our reduction  
 can be combined with the one in \cite{time14} to obtain decidability of $\regmtl$ and elementary decidability 
 of $\mitl+\ureg$ + non-punctual past. These are amongst the most expressive decidable
extensions of $\mtl$ known so far.
 We also show an exact logic-automaton correspondence between the fragment $\sfmtl$ and 
 $\po$-a-clock ATA. Ouaknine and Worrell reduced $\mtl$ to $\po$-1 clock ATA. Our $\sfmtl$ achieves the
 converse too. It is not difficult to see that full $\regmtl$ can be reduced to equivalent 1 clock alternating timed automata. This  provides
 an alternative proof of decidability of $\regmtl$ but the proof will not extend to decidability of  $\regmtl+$ non-punctual past, nor prove elementary 
 decidability of $\mitl+\ureg+$non-punctual past. Hence, we believe that our proof technique has some advantages. 
  An interesting related formalism of timed regular expressions was defined by Asarin, Maler, Caspi, and shown to be expressively equivalent 
 to timed automata. Our $\regmtl$ has orthogonal expressivity, and it is boolean closed.
% The language of example \ref{eg1} cannot be captured by a classical timed automaton \cite{AD94}.
 The exact expressive power of $\regmtl$ which is between 1-clock ATA and $\po$-1-clock ATA is open. 
\bibliographystyle{plain}
\bibliography{papers}

\newpage 
\appendix

\centerline{\bf{\Large Appendix}}

\section{Rational Expressions and Star-Free Expressions}
\label{app:rat-sfexp}
We briefly introduce rational expressions and star-free expressions over an alphabet $\Sigma$. 
A rational expression over $\Sigma$ is constructed inductively 
using the atomic expressions $a \in \Sigma, \epsilon, \emptyset$ 
and combining them using concatenation, Kleene-star and  union.

A star-free expression also has the same atomic expressions, and allows combination using 
union, concatenation and complementation. 
For instance, $\Sigma^*$ is star-free since it can be written as $\neg \emptyset$.

\section{Exclusive Normal Form}
\label{app:exnf}
We  eliminate $\reg_I \at$ and $x \ureg_{I',\at} y$ respectively from temporal definitions 
 $\wB[w \leftrightarrow \reg_I \at]$ or $\wB[w \leftrightarrow x \ureg_{I',\at} y]$.
The idea is to first mark each point of the  timed word $\rho$ over $\Sigma \cup W$ 
with the information whether $\at$ is true or not at that point, obtaining a simple extension $\rho'$ 
of $\rho$, and then to refine 
this information by checking if $\at$ is true within an interval $I$.

Assume $\at=\re(\Ss)$. To say that $\re(\Ss)$ is true starting at  a point in the timed word, we have to look at the truth of subformulae in $\Ss$. The alphabet of the minimal DFA to check $\re(\Ss)$ is hence $2^{\Ss}=\Ss'$. This results in the minimal DFA 
accepting an expression $\re'(\Ss')$, and not $\re(\Ss)$. In the following, we show that 
$\re'(\Ss')$ is equivalent to $\re(\Ss)$.

The first thing we do to avoid dealing with sets of formulae of $\Ss$ being true at each point is to assume that 
the sets $\Ss$ are \emph{exclusive}: that is, at any point, exactly one formula from $\Ss$ can be true. 
If the sets $\Ss$ are all exclusive, then the formula is said to be in 
\emph{Exclusive Normal Form}.
If $\Ss$ is exclusive, then we will be  marking positions in the word 
over $\Ss$ and not $\mathcal{P}(\Ss)$. This way, the untimed words 
$\mathsf{Seg}(\Ss, i, j)$ as well as $\mathsf{TSeg}(\Ss, i, I)$ that were used in the semantics 
of  $\varphi_1 \uregm_{I, \re(\Ss)}\varphi_2$, $\regm_{I, \re(\Ss)}$ respectively 
will be words over $\Ss$. The satisfaction of $\varphi_1 \uregm_{I, \re(\Ss)}\varphi_2$, $\regm_{I, \re(\Ss)}$ at any point $i$ will then amount to simply checking 
if  $\mathsf{Seg}(\Ss,i,j),\mathsf{TSeg}(\Ss, i, I) \in L(\re(\Ss))$.  

We now show that the exclusiveness of $\Ss$ can be achieved by a simple translation. 
 
\begin{lemma}
Given any $\regmtl$ formula $\varphi$ of the form $\regm^{\Ss}_{I, \re(\Ss)}$ or $\varphi_1 \uregm^{\Ss}_{I, \re(\Ss)}\varphi_2$, 
there exists an equivalent formula $\psi \in \regmtl$ in exclusive normal form. 
\end{lemma}
 \begin{proof}
 	Let $\Ss=\{\phi_1, \dots, \phi_n\}$. Construct set $\Ss'$ consisting 
of all formulae of the form $\bigwedge_{i \in K} \phi_i \wedge \bigwedge_{
i \notin K} \neg \phi_i$ for all possible subsets $K \subseteq \{1,2,\dots,n\}$. 

		Consider any formula of the form $\reg_{I}(\re(S))$. 
			Let $W_i$ denote the set consisting of all subsets of $\{1,2,\ldots,K\}$ which contains $i$.
		The satisfaction of $\varphi_i$ is then equivalent to that of $\sum \limits_{W\in W_i} \phi_{W}$.
		We can thus replace any $\varphi_i$ occurring in $\re(\Ss)$ with $\sum \limits_{W\in W_i} \phi_{W}$. 
	This results in obtaining a rational expression $\re'$ over $\Ss'$.
	
	It can be shown that $\re(\Ss)$ is equivalent to $\re'(\Ss')$ by inducting on the structure of $\re$.
	 \end{proof}

Thus, the minimal DFA we construct for $\re(\Ss)$ in the temporal definition will end up 
accepting $\re'(\Ss')$, equivalent to $\re(\Ss)$.

\section{1-$\mathsf{TPTL}$ for  $\regm_{[l,u)} \at$}
\label{app:1-tptl}
We encode in 1-$\mathsf{TPTL}$ an accepting run going through a sequence of merges capturing $\regm_{[l,u)} \at$
at a point $e$. 
To encode an accepting run going through a sequence of merges capturing $\regm_{[l,u)} \at$
at a point $e$, we assert $\varphi_{chk1} \vee \varphi_{chk2}$ at $e$, assuming $l \neq 0$. 
If $l=0$, we assert  $\varphi_{chk3}$. Recall that $m$ is the number of states in the minimal DFA accepting $\at$.
 \begin{itemize}
 \item Let $\mathsf{cond1}=0 \le n < m$, and
 \item Let $\mathsf{cond2}=1\le i_1 <i_2 <\ldots <i_n<i\le m$.
\item   		$\varphi_{chk1} = \bigvee \limits_{\mathsf{cond1}}\bigvee \limits_{\mathsf{cond2}}
 		x.\fut (x<l \wedge \nx[(x \geq l) \wedge \mathsf{GoodRun}])$\\
 \item 		$\varphi_{chk2} = \bigvee \limits_{\mathsf{cond1}}\bigvee \limits_{\mathsf{cond2}}
 		x.(\nx[(x \geq l) \wedge \mathsf{GoodRun}])$
 \item 		$\varphi_{chk3} = \bigvee \limits_{\mathsf{cond1}}\bigvee \limits_{\mathsf{cond2}}
 		x.\mathsf{GoodRun}$

 \end{itemize}
  
 		 		 where 
 		$\mathsf{GoodRun}$ is the formula which describes the run starting in $q_1$ in thread $\Th_i$, going through a sequence of merges, and 
 		witnesses $q_f$ in a merged thread $\Th_{i_1}$  		
 		at a point when  $x \in [l, u)$, and is the maximal point in $[l, u)$. 
 		
 		  	$\mathsf{GoodRun}$ is given by 	
 $\Th_i(q_1) \wedge [\{\neg \mathsf{Mrg}(i)\} \until [\merge(i_n,i) \wedge 
 \{\neg \mathsf{Mrg}(i_n) \} \until [\merge(i_{n-1},i_n)  \ldots\\
 \{\neg \mathsf{Mrg}(i_2)\} \until [\merge(i_1,i_2) \wedge \bigvee \limits_{q \in Q_F} \Next(\Th_{i_1}(q)) \wedge x \in [l,u) \wedge \nx (x >u) ]\ldots ]] ]]$ 
 	where $\mathsf{Mrg}(i)$ is the formula $\bigvee \limits_{j<i} \merge(j,i)$.
 	
 	The idea is to freeze the clock at the current point $e$, and start checking a good run from the first point in the interval $[l,u)$. 
 $\varphi_{chk1}$ is the case when the next point after point $e$ is not at distance $[l, u)$
 from $e$, while  $\varphi_{chk2}$ 
 handles the case when the next point after $e$ is at distance $[l,u)$ from $e$. In both cases, $l>0$.  (If $l=0$, we assert  $\varphi_{chk3}$).
  	Let $\Th_i$  be the thread 
 	having the initial state $q_1$ in the start of the interval $I$. 
  Let $i_1$ be the index of the thread to which $\Th_i$ eventually merged (at the last point in the interval $[l,u)$ from $e$). 
  The next expected state of thread $\Th_{i_1}$ is one of the final states if and only if the sub-string within the interval $[l,u)$ from the point $e$ satisfies the regular expression $\at$. Note that when the frozen clock is $\geq l$, 
     we start the run with $\Th_i(q_1)$, go through the merges, and check that 
  $x \in I$ when we encounter a thread $\Th_{i_1}(q_f)$, with $q_f$ being a final state. 
  To ensure that we have covered checking all points in $\tau_e+I$, we ensure that at the next point after 
  $\Th_{i_1}(q_f)$, $x >u$.   The decidability of $1{-}\mathsf{TPTL}$ gives the decidability 
  of $\regmtl$.

 \section{Proof of Lemma  \ref{elm-reg}}
 \label{app:pref-suf}
 \begin{proof}
 		Starting with the simple extension $\rho'$ having the information about the runs of $\mathcal{A}_{\at}$, 
 	we explain the construction of the oversampled extension $\rho''$ as follows:	
 	\begin{itemize}
	\item  We first oversample $\rho'$ at all the integer timestamps and mark them with propositions in $C = \{c_0,\ldots,c_{max-1}\}$ where $max$ is the maximum constant used in timing constraints of the input formulae. 
  An integer timestamp $k$ is marked $c_i$ if and only if $k=M(max)+i$ where $M(max)$ denotes a non-negative integral multiple of $max$ and 
  $0 \leq i \leq max-1$.  This can be done easily by the formula \\
  $c_0{\wedge} {\bigwedge \limits_{i \in \{0,\ldots max-1\}}}\wB(c_i \rightarrow \neg \fut_{(0,1)} (\bigvee C) \wedge \fut_{(0,1]} c_{i \oplus 1})$
	where $x \oplus y $  is addition of $x, y$ modulo $max$.

	\item Next,  a new point marked $\ovs$ is introduced at all time points $\tau$ whenever $\tau-l$ or $\tau-u$ is marked with $\bigvee \Sigma$.  
  This ensures that for any time point $t$ in $\rho''$, the points $t+l, t+u$ are also available 
 in $\rho''$.
 \end{itemize}
After the addition of integer time points, and points marked $\ovs$, we obtain the 
 oversampled extension $(\Sigma \cup W \cup \Threads \cup \Merge, C \cup \{\ovs\})$ $\rho''$
of $\rho'$.\\
 	   	\noindent  To check the truth of $\reg_{[l,u)} \at$ at a point $v$, 
 		we need to assert the following: starting from the time point $\tau_v+l$, 
 		we have to check the existence of an accepting run $R$ in $\mathcal{A}_{\at}$ 
 		such that the run starts from the first action point in the interval 		$[\tau_v+l, \tau_v+u)$, is a valid run which goes through some possible sequence of merging of threads, and 
 		witnesses a final state at the last action point in $[\tau_v+l, \tau_v+u)$.  To capture this, we start 
 		at the first action point in $[\tau_v+l, \tau_v+u)$ with initial state $q_1$ 
 		in some thread $\Th_i$, and proceed for some time with $\Th_i$ active, until we reach a point 
 		where $\Th_i$ is merged with some $\Th_{i_1}$. This is followed by $\Th_{i_1}$ remaining active until 
 		we reach a point where $\Th_{i_1}$ is merged with some other thread $\Th_{i_2}$ and so on, until we reach 
 		the last such merge where some thread say $\Th_n$ witnesses a final state at the last action point 
 		in  $[\tau_v+l, \tau_v+u)$. A nesting of until formulae captures this sequence of merges of the threads, 
 		starting with $\Th_i$ in the initial state $q_1$. Starting at $v$,
 		we have the point marked $\ovs$ at $\tau_v+l$, which helps us to anchor there and
 		 start asserting the existence of the run.  
 		 
 		 The issue is that the nested until can not 
 		 keep track of the time elapse since $\tau_v+l$.
 		 However, note that the greatest integer point in $[\tau_v+l, \tau_v+u)$ is uniquely marked 
 		 with $c_{j \oplus u}$ whenever $c_j \leq \tau_v \leq c_{j \oplus 1}$ are the closest integer points to $\tau_v$.
 		  We make use of this by (i) asserting the run of $\mathcal{A}_{\at}$ 
 		 until we reach $c_{j \oplus u}$ from $\tau_v+l$. Let the part of the run $R$ that has been witnessed until 
 		 $c_{j \oplus u}$ be $R_{pref}$. Let $R=R_{pref}.R_{suf}$ be the accepting run. (ii) From $\tau_v+l$, we jump to $\tau_v+u$, and 
 		 assert the reverse of $R_{suf}$ till we reach  $c_{j \oplus u}$. This ensures that 
 		 $R=R_{pref}.R_{suf}$ is a valid run in the interval $[\tau_v+l, \tau_v+u)$.

 		 Let $\mathsf{Mrg}(j)=[\bigvee \limits_{k<j} \merge(k,j) \vee c_{j \oplus u}]$. 
 		
 		We first write a formula that captures $R_{pref}$. Given a point $v$, the formula 
 		captures a sequence of merges through threads $i>i_1> \dots >i_{k_1}$, 
 		and $m$ is the number of states of $\mathcal{A}_{\at}$.      
 		 
 		 Let $\varphi_{Pref, k_1}=\bigvee_{m \geq  i> i_1 > \dots > i_{k_1}} \mathsf{MergeseqPref}(k_1)$ where $\mathsf{MergeseqPref}(k_1)$ is the formula 
 		$$\fut_{[l,l]}\{\neg (\bigvee \Sigma \vee c_{i \oplus u}) \until [\Th_i(q_1) \wedge (\neg \mathsf{Mrg}(i) \until 
 		[\merge(i_1,i) \wedge$$ 
 		$$(\neg \mathsf{Mrg}(i_1) 
 		\until [\merge(i_2,i_1) \wedge \dots   
 		(\neg \mathsf{Mrg}(i_{k_1}) \until c_{i \oplus u}) ])])]\}$$
 	Note that this asserts the existence of a run till $c_{i \oplus u}$ going 	through a sequence of merges starting at $\tau_v+l$. Also, $\Th_{i_{k_1}}$ is the guessed 
 	last active thread till we reach $c_{i \oplus u}$ which will be merged 
 	in the continuation of the run from $c_{i \oplus u}$.

\begin{figure}
\includegraphics[scale=0.4]{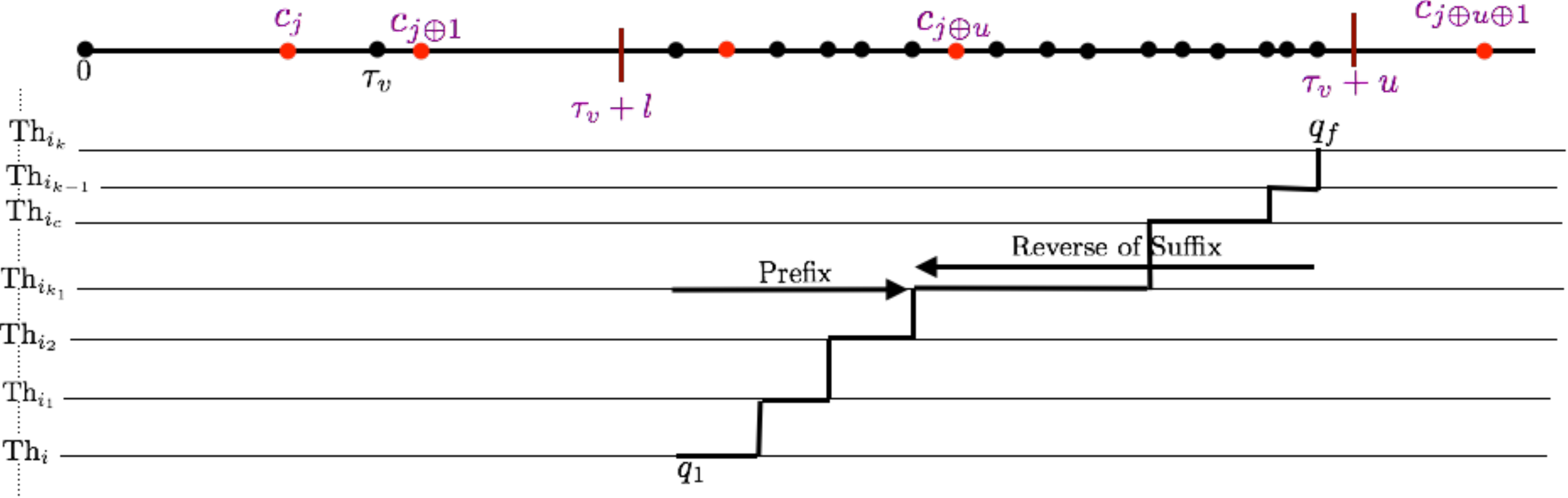}	
\caption{The linking thread at $c_{j \oplus u}$. The points in red are the oversampling integer points, and so are 
$\tau_v+l$ and $\tau_v+u$. }
\label{fig:runinfo-2}
\end{figure}

 	Now we start at $\tau_v+u$
 	and assert that we witness a final state sometime as part of some thread $\Th_{i_k}$, 
 	and walk backwards such that some thread $i_t$ got merged to $i_k$, and so on,
 	we reach a thread $\Th_{i_c}$ to which thread  $\Th_{i_{k_1}}$ merges with. 
 	Note that $\Th_{i_{k_1}}$
 	was active 
 	when we reached $c_{i \oplus u}$. This thread $\Th_{i_{k_1}}$ is thus the ``linking point'' of the forward and reverse runs.   See Figure \ref{fig:runinfo-2}. 
 		   		
 		   		Let $\varphi_{Suf, k, k_1}=\bigvee_{1\leq i_{k}<\dots<i_{k_1}\leq m} \mathsf{MergeseqSuf}(k,k_1)$
 		   		where $\mathsf{MergeseqSuf}(k,k_1)$ is the formula \\
 		   		
 		$\fut_{[u,u]}\{\neg (\bigvee \Sigma \vee c_{i \oplus u}) \since [(\Th_{i_k}(q_f)) \wedge (\neg \mathsf{Mrg}(i_k) \since $ 
 		$[\merge(i_k,i_{k-1}) \wedge (\neg \mathsf{Mrg}(i_{k-1}) 
 		\since$ \\
 		$ 		[\merge(i_{k-1},i_{k-2}) \wedge \dots  \merge(i_c, i_{k_1}) \wedge 
 		 		(\neg \mathsf{Mrg}(i_{k_1}) \since c_{i \oplus u})])])]\}$.
 		 		
  For a fixed sequence of merges, the  formula $$\varphi_{k,k_1}=\bigvee_{k \geq k_1 \geq 1} [\mathsf{MergeseqPref}(k_1) {\wedge} \mathsf{MergeseqSuf}(k,k_1)]$$ captures an accepting run 
using the merge sequence. Disjuncting over all possible sequences for a starting thread $\Th_i$, and disjuncting over all 
possible starting threads gives the required formula capturing  an accepting run.  
Note that 
this resultant formulae is also relativized with respect to $\Sigma$ and also conjuncted with 
$Rel(\mathsf{Run}, \Sigma)$ (where $\mathsf{Run}$ is the formula capturing the run information in $\rho'$ as seen in section \ref{reg-xtnd})  to obtain the equisatisfiable $\mathsf{MTL}$ formula. The relativization of $\mathsf{Run}$ with respect to $\Sigma$ 
can be done as illustrated in Figure \ref{eg-os}.) 
  Note that $\mathsf{S}$ can be eliminated obtaining an equisatisfiable $\mtl[\until_I]$ formula 
modulo simple projections \cite{deepak08}. 

If $I$ was an unbounded interval of the form $[l, \infty)$, then in formula $\varphi_{k,k_1}$, we do not require  $\mathsf{MergeseqSuf}(k,k_1)$; instead, 
we will go all the way till the end of the word, and assert $\Th_{i_k}(q_f)$ at the last action point of the word. 
Thus, for unbounded intervals, we do not need any oversampling at integer points.  
 		\end{proof}
%%%%%%%%%%%%%%%%%%%%%%%%%%%%%%%%%%%%%%%%%%%%%%% 	
 	\section{Elimination of $\ureg_{I,\re}$}
 	 \label{app:ureg}
 	\begin{lemma}
 		 		Let $T=\wB[a \leftrightarrow x \ureg_{I,\re} y]$ be a  temporal definition  built from $\Sigma \cup W$. 
 		Then we synthesize a formula $\psi \in \mtl$ 
 		over  $\Sigma \cup W \cup X$ such that $T$ is equivalent to $\psi$ modulo oversampling. 
 	%	such that $T\equiv \exists \downarrow X. \psi$. 
 	\label{elm-ureg}\end{lemma}
 		\begin{proof}
 	We discuss first the case of bounded intervals.   
 The proof technique is very similar to Lemma \ref{elm-reg}. The differences that arise are as below.
 \begin{enumerate}
 \item Checking $\re$ in $\reg_{I}\re$ at point $v$  is done at all points $j$ such that $\tau_j - \tau_v \in I$. 
   To ensure this, we needed the punctual  modalities $\fut_{[u,u]},\fut_{[l,l]}$.  
    On the other hand,  to check  $\ureg_{I, \re}$ from a point $v$, the check on $\re$ is done from the first point after $\tau_v$, and 
 ends at some point within $[\tau_v+l,\tau_v+u)$. Assuming $\tau_v$ lies between integer points $c_i, c_{i \oplus 1}$, 
  we can witness the forward run in $\mathsf{MergeseqPref}$ from the next point after $\tau_v$ till $c_{i \oplus 1}$, and 
  for the reverse run, go to some point in $\tau_v+I$ where the final state is witnessed in a merged thread, and walk back till 
 $c_{i \oplus 1}$, ensuring the continuity of the threads merged across  $c_{i \oplus 1}$.
    The punctual modalities are hence not required 
  and we do not need points marked  $\ovs$. 
   \item The formulae $\mathsf{MergeseqPref}(k_1)$, $\mathsf{MergeseqSuf}(k,k_1)$ 
   of the lemma \ref{elm-reg} are replaced as follows:
   \begin{itemize}
\item $\mathsf{MergeseqPref}(k_1) :     \{\neg (\bigvee \Sigma \vee c_{i \oplus 1}) \until
 [\Th_i(q_1) \wedge (\neg \mathsf{Mrg}(i) \until
 		[\merge(i_1,i) \wedge (\neg \mathsf{Mrg}(i_1) \until$\\
 		$ [\merge(i_2,i_1) \wedge \dots (\neg \mathsf{Mrg}(i_{k_1}) \until c_{i \oplus 1}) ])])]\}$.
  
 \item   $\mathsf{MergeseqSuf}(k,k_1) : \fut_{I}\{[(\Th_{i_k}(q_f)) {\wedge} (\neg \mathsf{Mrg}(i_k) \since$
 		$[\merge(i_k,i_{k-1}) \wedge (\neg \mathsf{Mrg}(i_{k-1}) 
 		\since$\\
 		$ [\merge(i_{k-1},i_{k-2}) \wedge$
 		 $\dots  \merge(i_c, i_{k_1}) \wedge 
 		 		(\neg \mathsf{Mrg}(i_{k_1}) \since c_{i \oplus 1})])])]\}$
 
   \end{itemize}
 		
  \end{enumerate}
  The above takes care of $\re$ in $x \ureg_{I,\re} y$ : we also need to 
  say that  $x$ holds continously from the current point to some point in $I$. This is done by pushing $x$ into $\re$ 
  (see the translation of $\varphi_1 \ureg_{I, \re} \varphi_2$ to
  $\reg_I \re'$ in Appendix \ref{app:exp}). 
  The resultant formulae is  relativized with respect to $\Sigma$ and also conjuncted with 
$Rel(\mathsf{Run}, \Sigma)$ to obtain the equisatisfiable $\mathsf{MTL}$ formula.   
 
 Now we consider unbounded intervals. 
 The major challenge for the unbounded case is that the point where we asserting $\Th_{i_k}(q_f)$ (call this point $w$) 
 may be far away from the point $v$ where we begin : that is, if $\tau_v$ is flanked between integer points 
 marked $c_i$ and $c_{i \oplus 1}$, it is possible to see multiple occurrences of 
 $c_{i \oplus 1}$ between $\tau_v$ and 
 and the point in $\tau_v +I$ which witnesses $\Th_{i_k}(q_f)$. 
 In this case, when walking back reading the reverse of the suffix, it is not easy to stitch 
 it back to the first $c_{i \oplus 1}$  seen after $\tau_v$.
  The possible non-uniqueness of $c_{i \oplus 1}$ thus poses a problem in reusing our technique  in the bounded interval case.  
 Thus we consider two cases:
\\ Case 1: In this case, we assume that our point $w$ lies within 
$[\tau_v+l, \lceil \tau_v +l \rceil)$. Note that $\lceil \tau_v +l \rceil$ is the nearest point from $v$ marked with 
$c_{i \oplus l \oplus 1}$. This can be checked by asserting $\neg c_{i\oplus l \oplus 1}$ all the way till $c_{i\oplus 1}$ while walking backward 
from $w$, where $\Th_{i_k}(q_f)$ is witnessed. 
 The formula $\mathsf{MergeseqPref}(k_1)$ does not change. $\mathsf{MergeseqSuf}(k,k_1)$ is as follows:
\\$$\fut_{[l,l+1)}\{[(\Th_{i_k}(q_f)) \wedge (\neg \mathsf{Mrg}'(i_k) \since 
 		[\merge(i_k,i_{k-1}) \wedge (\neg \mathsf{Mrg}'(i_{k-1}) $$
 		$$\since [\merge(i_{k-1},i_{k-2}) \wedge \dots  \merge(i_c, i_{k_1}) \wedge 
 		 		(\neg \mathsf{Mrg}'(i_{k_1}) \since c_{i \oplus 1})])])]\}$$

where $\mathsf{Mrg}'(i)=[\bigvee \limits_{j<i} \merge(j,i) \vee c_{i \oplus l \oplus 1}]$
\\ Case 2: In this case, we assume the complement. That is the point $w$ occurs after $\lceil \tau_v +l \rceil$. In this case, we assert the prefix till $c_{i \oplus l \oplus 1}$ and then continue asserting the suffix from this point in the forward fashion unlike other cases. The changed $\mathsf{MergeseqPref}$ and $\mathsf{MergeseqSuf}$ are as follows:
\begin{itemize}
\item 
$\mathsf{MergeseqPref}(k_1)$:
    $$\{\neg (\bigvee \Sigma \vee c_{i \oplus l \oplus 1}) \until [\Th_i(q_1) \wedge (\neg \mathsf{Mrg}(i) \until 
 		[\merge(i_1,i) \wedge$$
 		 $$(\neg \mathsf{Mrg}(i_1) 
 		\until [\merge(i_2,i_1) \wedge \dots   
 		(\neg \mathsf{Mrg}(i_{k_1}) \until c_{i \oplus l \oplus 1}) ])])]\}$$
 \item   $\mathsf{MergeseqSuf}(k,k_1)$:
$$\fut_{[l+1,l+2)}\{[c_{i \oplus l \oplus 1} \wedge (\neg \mathsf{Mrg}(i_{k_1}) \until 
 		[\merge(i_c, i_{k_1}) \wedge (\neg \mathsf{Mrg}(i_c) $$
 		$$\until [\merge(i_c,i_{k_1}) \wedge \dots  \merge(i_{k-1}, i_{k-2}) \wedge 
 		 		(\neg \mathsf{Mrg}(i_{k-1}) \until$$
 		 		 $(\Th_{i_k}(q_f))])])]\}$
where $\mathsf{Mrg}(i)=[\bigvee \limits_{j<i} \merge(j,i) ]$
  
\end{itemize}

   	\end{proof}

\subsection{Complexity of $\regmtl$ Fragments}
Given a formula $\varphi$ in ($\mitl$ or $\mtl$ or $\regmtl$), the size of $\varphi$ denoted $|\varphi|$ 
is defined by taking into consideration, the number of temporal modalities $\until_I, \nex_I$, the number of boolean connectives, as well as the maximal constant occurring in the formulae (encoded in binary). The size is  
 defined  as 
$log K \times$(the number of temporal modalities in $\varphi$ + number of boolean connectives in $\varphi$), where 
$K$ is the max constant appearing in the formulae.  
For example, $|a \until_{(0,2)} (\neg b \wedge c \until_{(0,1)}d)|=log2 \times 4$.  
In all our complexity results, we assume a binary encoding 
of all constants involved in the formulae.  

To prove the complexity results we first need the following lemma.

\begin{lemma}
	\label{expmitl}
Given any $\mitl$ formula $\varphi$ with $|\varphi|=\mathcal{O}(2^n)$ (for some $n \in \mathbb{N}$)  with maximum constant $K$ used in timing intervals, the satisfiability checking for $\varphi$ is $\mathsf{EXPSPACE}$ in $n$ and $log(K)$.
\end{lemma}
\begin{proof}
	 Given any $\mathsf{MITL}$ formula $\varphi$ with $|\varphi|=\mathcal{O}(2^n)$, there are at most $\mathsf{expn}= \mathcal{O}(2^n)$ number of temporal modalities and boolean connectives. Let $K$ be the maximal constant used in $\varphi$.
	We give a satisfiability preserving reduction from $\varphi$ to $\psi \in \mathsf{MITL[\until_{0,\infty},\since]}$.
	$\mathsf{MITL[\until_{0,\infty},\since]}$ is the fragment of $\mathsf{MITL}$ with untimed 
	past and the intervals in future modalities 
	are only of the form $\langle 0,u \rangle$ or $\langle l, \infty)$.
	The satisfiability checking for $\mathsf{MITL}[\until_{0,\infty},\since]$ is in $\mathsf{PSPACE}$ \cite{AFH96}. 
	Hence, the reduction from a  $\mathsf{MITL}$ formula with $\mathsf{expn}= \mathcal{O}(2^n)$ number of modalities
	to an $\mathsf{MITL}[\until_{0,\infty},\since]$ formula with $\mathcal{O}(poly(K.\mathsf{expn}))$ modalities preserving 
	the max constant $K$, gives an $\mathsf{EXPSPACE}$ upper bound in $n, logK$. The $\mathsf{EXPSPACE}$ hardness of $\mitl$ can be found in \cite{AFH96}.
	The reduction from $\mathsf{MITL}$ to $\mathsf{MITL[\until_{0,\infty},\since]}$ is achieved  as follows:
	 	\begin{figure}[h]
	 	\includegraphics[scale=.45]{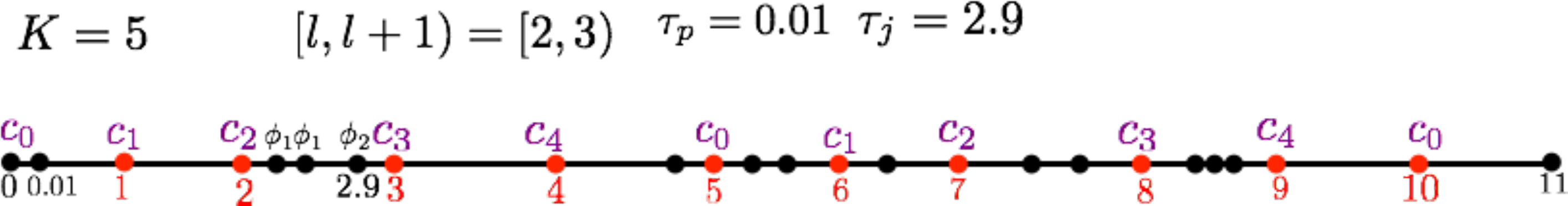}
	 	\caption{The point $p$ has $\tau_p=0.01$, and for $l=2$, $\tau_p+[l, l+1)=[2.01, 3.01)$. 
	 	$\tau_p$ lies between points $c_{i-1}=c_0$ and $c_i=c_1$. In this case, $\tau_j=2.9$ where $\phi_2$ holds 
	 	and $c_3=c_{i \oplus l}$ does not lie between $\tau_p, \tau_j$. }
	 	\label{case1}	
	 		 	\end{figure}
	\begin{figure}[h]
	 	\includegraphics[scale=.45]{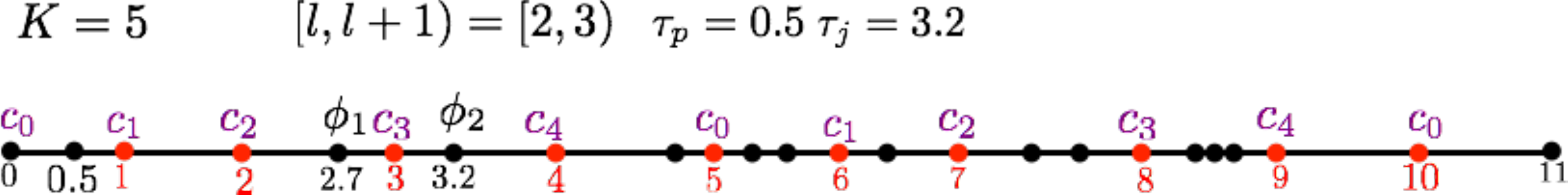}
	 	\caption{The point $p$ has $\tau_p=0.5$, and for $l=2$, $\tau_p+[l, l+1)=[2.5, 3.5)$. 
	 	$\tau_p$ lies between points $c_{i-1}=c_0$ and $c_i=c_1$. In this case, $\tau_j=3.2$ where $\phi_2$ holds 
	 	and $c_3=c_{i \oplus l}$ lies between $\tau_p, \tau_j$. In this case, $c_3$ is in the interval $[0,l+1)=[0,3)$
	 	from $\tau_p$ 	 	and $\tau_j$ lies in the unit interval between $c_{i \oplus l}=c_3$ and $c_{i \oplus l \oplus 1}=c_4$.}
	 	\label{case2}	
	 	
	 		 	\end{figure}

	 \begin{enumerate}
	 	\item[(a)] 
	 	Break each $\until_I$ formulae in $\mathsf{MITL}$ where $I$ is a bounded interval,  into disjunctions of $\until_{I_i}$ modality, where each $I_i$ is a unit length interval and union of all $I_i$ is equal to $I$. That is, $\phi_1 \until_{\langle l,u \rangle} \phi_2 \equiv \phi_1 \until_{\langle l,l+1 )} \phi_2 \vee \phi_1 \until_{[ l+1,l+2 )} \phi_2 \ldots \vee \phi_1 \until_{[u-1,u \rangle} \phi_2$. This  increases the size  from $\mathsf{expn}$ to $\mathsf{expn} \times K$. 
	 	
	 	\item[(b)] Next, we flatten all the modalities containing bounded intervals. This results in replacing subformulae of the form
	 	$\phi_1 \until_{[l,l+1)} \phi_2$ with new witness variables. This results in the conjunction 
	 	of temporal definitions of the form  $\wB[a \leftrightarrow \phi_1 \until_{[l,l+1)} \phi_2]$ to the formula,  
	 	and creates  only a linear blow up in the size.
	 	
	 %	\item[(c)] 
	 	
	 \end{enumerate}
	 
	 	Now consider any temporal definition $\wB[a \leftrightarrow \phi_1 \until_{[l,l+1)} \phi_2]$. 
	 	  We show a reduction to an equisatisfiable $\mathsf{MITL[\until_{0,\infty},\since]}$ formula by eliminating each 
	 $\wB[a \leftrightarrow \phi_1 \until_{[l,l+1)} \phi_2]$ and replacing it   with untimed
	 $\since$ and $\until$ modalities with intervals $<0,u>$ and $<l, \infty)$. 
	 
	 \begin{itemize}
	 	
	 	\item First we oversample the words at integer points $C=\{c_0, c_1, c_2, \dots, c_{K-1}\}$. 
	 	An integer timestamp $k$ is marked $c_i$ if and only if $k=M(K)+i$, where $M(K)$ denotes a non-negative integer multiple of $K$, and 
	 	$0 \leq i \leq K-1$. 
	 	This can be done easily by the formula 
	 	\begin{quote}
	 		$c_0{\wedge} {\bigwedge \limits_{i \in \{0,\ldots K-1\}}}\wB(c_i {\rightarrow} \neg \fut_{(0,1)} (\bigvee C) \wedge \fut_{(0,1]} c_{i \oplus 1})$
	 	\end{quote} 
	 	where $x \oplus y $ is $(x + y) \% K$ (recall that $(x+y) \% K=M(K)+(x+y), 0 \leq x+y \leq K-1$).

	\item Consider any point $p$  within a unit integer interval whose end points are marked $c_{i-1}, c_{i}$. 
	 	Then $\phi_1 \until_{[l,l+1)} \phi_2$ is true at that point $p$ if and only if, $\phi_1$ is true on all the action points till a point $j$ in the future, such that

	 	\begin{itemize}
	 		\item   either $j$ occurs within $[l,\infty)$ from $p$ and there is no $c_{i\oplus l}$ between $p$ and $j$ ($\tau_j \in [\tau_p+l, \lceil \tau_p+l \rceil]$ )  (see figure \ref{case1})
	 		\begin{quote}
	 		 $\phi_{C1,p} = (\phi_1 \wedge \neg c_{i \oplus l} ) \until_{[l,\infty)} \phi_2$ 	
	 		\end{quote}
	 		
	 		\item or,   $j$ occurs within $[0,l+1)$ from $p$, and $j$ is within a unit
	 		interval whose end points are marked $c_{i\oplus l}$ and $c_{i \oplus l \oplus 1}$ ($\tau_j \in [\lceil \tau_p+l \rceil, \tau_p+l+1)$ ) (see figure \ref{case2})
	 		\begin{quote}
	 				$\phi_{C2,p} = \phi_1  \until_{[0,l+1)} (\phi_2 \wedge  (\neg (\bigvee C)) \since c_{i\oplus l})$ 	
	 			\end{quote}
	 	\end{itemize}

	 	The temporal definition $\wB[a \leftrightarrow \phi_1 \until_{[l,l+1)} \phi_2]$ is then captured by \\
	 	$\bigvee \limits_{i = 1}^{K-1} \wB[\{a \wedge (\neg (\bigvee C) \until c_{i})\}  \leftrightarrow \phi_{C1,i} \vee 
	 	\phi_{C2,i}]$
	 	\end{itemize}

	 	To eliminate each bounded interval modality as seen in (a),(b) above, 
	 	we need an $\mathcal{O}(K)$ increase in size. Each temporal definition is replaced with a formula with  
	 	of size $\mathcal{O}(K)$. 	 	
	 	Thus the size of the new formula  is $\mathcal{O}(2^n) \times \mathcal{O}(K) \times \mathcal{O}(K)$, and the total number of propositions needed is $2^{\Sigma}\cup\{c_0,\ldots,c_{K-1}\}$. Assuming binary encoding for $K$, we get a $\mathsf{MITL}[\until_{0,\infty},\since]$ formulae whose size is exponential in $n$ and $log K$. As the satisfiability checking for $\mathsf{MITL}[\until_{0,\infty},\since]$ is in $\mathsf{PSPACE}$ \cite{AFH96},  we get an $\mathsf{EXPSPACE}$ upper bound
	 	in $n, log K$. 	 	The $\mathsf{EXPSPACE}$ hardness of $\mitl$ can be found in \cite{AFH96}.
	 
	 \end{proof}

\subsection{Proof of Theorem \ref{mitl-ureg}.2 : $\mitl+ \UM$ is $\mathsf{EXPSPACE}$-complete}
\label{app:th-um}
Starting from an $\mitl+\UM$ formula, we first show how to obtain an equivalent
$\mitl$ formula modulo simple projections. The constants appearing in a $\mitl+\UM$ formula 
come from those which are part of the time intervals $I$ decorating the temporal modalities, as well as 
those from counting constraints $k \% n$. If we consider some 
$\until$ modality, say $\until_{(l,u),\#b =  k \% n}$, then the number of bits needed to encode 
this modality is $(log l+ log u+ log k + log n)=\mathcal{O}(log(u)+log(n))$.  Let $n_{max}$ and 
$u_{max}$ be the maximal constants appearing in the counting constraints as well as time intervals   
of a $\mitl+\UM$ formula $\phi$. Then 
$|\phi|=(log(n_{max})+log(u_{max}))\times$(the number of temporal modalities in $\phi$+ number of boolean connectives in $\phi$).

\subsubsection*{Elimination of $\UM$}
\label{app:um}
In this section, we show how to eliminate $\UM$  
from  $\mtl+\UM$ 
over strictly monotonic timed words.
This can be extended to weakly monotonic timed words.
Given any $\mtl+\UM$  formula $\varphi$ over $\Sigma$, we first
 ``flatten" the  $\UM$ modalities 
of $\varphi$ and obtain a flattened formula. 

\noindent \emph{Example}. The formula 
$\varphi=[a \until(e \wedge (f \until_{(2,3),\#b =  2 \% 5}y))]$
 can be flattened by replacing the 
$\UM$ with a fresh witness proposition $w$ to obtain \\
$\varphi_{flat}=[a \until (e \wedge w)] {\wedge} \wB\{w \leftrightarrow (f \until_{(2,3),\#b =  2 \% 5}y)\}$.

Starting from $\chi \in \mtl+\UM$,	in the following, we now show how to 
	obtain equisatisfiable $\mtl$ formulae corresponding to each temporal projection containing a $\UM$ modality. 
 \begin{enumerate}
 \item \textbf{\it {Flattening}} :  Flatten  $\chi$ obtaining $\chi_{flat}$ over $\Sigma \cup W$, where $W$ is the set of witness propositions 
used, $\Sigma \cap W=\emptyset$.
  \item \textbf{\it{Eliminate Counting}} :  Consider, one  by one,
each temporal definition $T_i$ of $\chi_{flat}$.
  	 Let $\Sigma_i=\Sigma\cup W \cup X_i$, where $X_i$ is 
  a set of fresh propositions,  $X_i \cap X_j=\emptyset$ for $i \neq j$. 
 \begin{itemize}
 \item For each temporal projection $T_i$  containing a  $\UM$ modality
of the form $x \until_{I,\#b = k\% n}y$, Lemma \ref{um-elim}
gives  $\zeta_i \in \mtl$ over $\Sigma_i$ such that  $\zeta_i$
is equisatisfiable to $T_i$ modulo simple extensions.
%$T_i \equiv \exists X_i. \zeta_i$.
 \end{itemize}
 \item \textbf{\it{Putting it all together}} :  The formula
     $\zeta {=} \bigwedge_{i=1}^k \zeta_i  \in \mtl$ is such that
     it is equisatisfiable to modulo simple extensions, over the extra propositions 
     $X=\bigcup_{i=1}^k X_i$. 
 \end{enumerate}
 For elimination of $\UM$,    
marking witnesses correctly is ensured using an extra set of symbols $B = \{b_0,...,b_{n}\}$ which act as counters incremented in a circular fashion. Each time a witness of the formula which is being counted is encountered, the counter increments, else it remains same. 
The evaluation of the mod counting formulae can be reduced to checking the difference between indices between the first and the last symbol in the time region where the counting constraint is checked.
\subsubsection*{Construction of Simple Extension}
\label{ext-const}
  Consider a temporal definition  $T=\wB[a \leftrightarrow x \UM_{I,\#b = k \% n} y]$,  built from $\Sigma \cup W$. 
  Let $\oplus$ denote  addition modulo $n+1$. 
  \begin{enumerate}
  	\item \emph{Construction of a ($\Sigma \cup W, B)$- simple extension}. We introduce a fresh set of propositions $B = \{b_0,b_1,\ldots,b_{n-1}\}$ and construct a family of simple extensions
  	$\rho'=(\sigma', \tau)$ from $\rho=(\sigma,\tau)$ as follows:
  	\begin{itemize}
  		\item \textbf {$C1$}: $\sigma'_1= \sigma_1 \cup \{b_0\}$. If $b_k \in \sigma'_i$   
  		and if $b \in \sigma_{i+1}$, $\sigma'_{i+1}=\sigma_{i+1} \cup \{b_{k \oplus 1}\}$. 
  		\item \textbf {$C2$}:  If  $b_k \in \sigma'_i$  and $b \notin \sigma_{i+1}$, then $\sigma'_{i+1}=\sigma_{i+1} \cup \{b_k\}$.
  		
  		\item \textbf{$C3$}: $\sigma'_{i}$ has exactly one symbol from $B$ for all $1 \leq i \leq |dom(\rho)|$. 
  	\end{itemize}
  	\item \emph{Formula specifying the above behaviour}. 
  	The variables in $B$ help in counting the number of $b$'s in $\rho$. 
  	$C1,C2$ and $C3$ are written in $\mathsf{MTL}$ as follows: 
  	\begin{itemize}
  		\item  $\delta_1 {=} \bigwedge \limits_{k=0}^n \wB[(\nex b\wedge b_k) \rightarrow \nex b_{k\oplus1}]$ and 
  	\item	$\delta_2 {=} \bigwedge \limits_{k=0}^n \wB[(\nex \neg b \wedge b_k) \rightarrow \nex b_k]$
  	\item $\delta_3 {=} \bigwedge \limits_{k=0}^n \wB[ b_k \rightarrow  \bigwedge \limits_{j \neq k} \neg b_j]$
  	  	\end{itemize}
  \end{enumerate}

 \begin{lemma}
	\label{um-elim}
	Consider a temporal definition 
	$T=\wB[a \leftrightarrow x \until_{I,\#b =  k \% n}y]$,  built from $\Sigma \cup W$. 
	Then we synthesize a formula $\psi \in \mtl$ 
	over  $\Sigma \cup W \cup X$ which is equivalent to $T$ modulo simple extensions.
%	such that $T\equiv \exists X. \psi$. 
\end{lemma}
\begin{proof}
	\begin{enumerate}
		\item Construct a simple extension $\rho'$ as shown in section \ref{ext-const}.
		\item Now checking whether at point $i$ in $\rho$, $x \until_{I,\#b =  k \% n}y$ is true, is equivalent to checking that at point $i$ in $\rho'$ there exist a point $j$ in the future where $y$ is true and for all the points between $j$ and $i$, $x$ is true and the difference between the index values of the symbols from $B$ at $i$ and $j$ is $k\%n$.
		$\phi_{mark,a} {=} \wB \bigwedge \limits_{i \in \{1,\ldots n-1\}}(a \wedge b_i {\leftrightarrow}   [x \until_I (y \wedge b_j)])$
		where $j = k+i \% n$.
		\item The formula $\delta_1 \wedge \delta_2 \wedge \delta_3 \wedge \phi_{mark,a}$ is equivalent to $T$ modulo simple projections. 
	\end{enumerate}
    \end{proof} 
    
  Notice that in the above reduction, if we start with an $\mitl+\UM$ formula, we will obtain an $\mitl$ formula since we do not introduce any new punctual intervals.   
    \begin{lemma}
    	Satisfiability of $\mitl+\UM$ is $\mathsf{EXPSPACE}$-complete.
    	\end{lemma}
    	\begin{proof}
  Assume that we have a $\mitl+\UM$ formula $\phi$ with $|\phi|=m$, and hence
  $\leq m$ $\UM$ modalities. Let  $\phi$ be over the alphabet $\Sigma$.   The number of propositions 
  used is hence $2^{\Sigma}$,  and let 
  $K$ be the maximal constant appearing in the intervals of $\phi$.    
     Let $k_1\%n_1, \dots, k_m\%n_m$ be the 
  modulo counting entities in $\phi$. Let $n_{max}$ be the maximum of $n_1, \dots, n_m$. 
  Going by the construction above, we obtain $m$ temporal definitions $T_1, \dots, T_{m}$, 
  corresponding to the $m$ $\UM$ modalities.
  
   To eliminate 
  each $T_i$, we introduce $n_{max}$ formulae of the form $\phi_{mark,a}$, 
  evaluated on timed words over $2^{\Sigma} \cup B_1 \cup  \dots \cup  B_{m} $.
  This is enforced by $\delta_1, \delta_2, \delta_3$. The number of propositions 
  in the obtained $\mitl$ formula is hence $|2^{\Sigma}|.|B_1+B_2+ \dots +B_m|$.  
  The size of the new formula is 
  is $\mathcal{O}(m.n_{max})$, while the maximum constant appearing in the intervals 
  is same as $K$. Thus we have an exponential size (the size now is 
  $\mathcal{O}(m.2^{log n_{max}})$) $\mitl$ formulae with max constant as $K$. 
  The $\mathsf{EXPSPACE}$-hardness of $\mitl+\UM$ follows from that of $\mitl$. 
   Lemma \ref{expmitl} now shows that  satisfiability checking for $\mitl+\UM$ is $\mathsf{EXPSPACE}$-complete.
  \end{proof}

\subsection{Proof of Theorem \ref{mitl-ureg}.3 : $\mitl+\ureg$ is in $\mathsf{2EXPSPACE}$}
\label{app:th-ureg}
  	\begin{proof}
   Consider a $\ureg$ modality $a \ureg_{I,\re(\Ss)} b$, where $\re$ is a rational expression 
   over $\Ss$ and $a, b \in \Sigma$.    
    The size of a  $\ureg$ modality is $size(\re)+log(l)+log(u)$, where $l, u$ are the lower and upper bounds of the interval $I$, and 
 $size(\re)$ is the size of the rational expression $\re$. 
  We first flatten $\varphi$ by introducing witness propositions for each $\ureg$ modality obtaining 
 temporal definitions of the form $\wB(w \leftrightarrow \ureg_{I,\re(\Ss)} \psi)$. Flattening only creates a linear blow up
 in formula size. 
Assume that $\varphi$ is flattened.
Let $T_i=\wB(w_i \leftrightarrow a \uregm^{\Ss_i}_{I_i,\re_i(\Ss_i)}b_i)$
be a temporal definition, and let there be $t$ temporal definitions. Let $l_i, u_i$ be the bounds 
of the interval $I_i$.   The size of $\varphi$, $|\varphi|$ is then defined as 
   $\mathcal{O}(\sum_{i=1}^t(n_i+log(l_i)+log(u_i))$, where $n_i$ is the size of $\re_i$. Let $u$ be the maximum 
   constant appearing in the intervals $I_i$.
   
%  	All these counting should be done after flattening. Note that flattening does not affect the size of the formulae (asymptotically) :
%  	Let there be $t$ number of modalities (temporal definitions) in the flattened formulae.
%  	Let $\wB(w_i \leftrightarrow a \ureg^{S_i}_{\re_i, \langle l_i,u_i \rangle} b)$ be the form of the temporal definition. Thus the formulae size is 
%  	$\mathcal{O}(\Sigma  \Sigma_{i = 1}^t (n_i+log(l)+log(u))$, where $n_i$ is the size of $re_i$. Assume that $n$ is the size of the largest $re$. Let $u$ be the maximum timing constants used. 

  	Let us consider a temporal definition $T=\wB(w \leftrightarrow \ureg_{I,\re(\Ss)} \psi)$.
  	
\begin{enumerate}
\item We look at the number of propositions needed in obtaining the equisatisfiable $\mtl$ formula. 	
\begin{enumerate}
\item  The size of the rational expression $\re$ in $T$ is $n$. The DFA accepting $\re$ has $\leq 2^n$ states.
The transitions of this DFA are over formulae from $S$. Since we convert the formulae into ExNF, 
we also convert this DFA into one whose transitions are over $2^{\Ss}$. 
 Hence, the number of transitions 
in the DFA is $\leq 2^n \times 2^{|\Ss|}$. Let $\Ss'=2^{\Ss}$.
\item This DFA is simulated using the symbols $\Threads, \Merge$. There can be at most $2^n$ threads, and 
each thread be in one of the $2^n$ states. Thus, the number of propositions $\Th_i(q)$ 
is at most ${2^n}^{2^n}$. Given that there are $t$ temporal definitions, we need $t \times {2^n}^{2^n}$ extra symbols. 
\item Each integer point in the timed word is marked with a symbol $c_i$, $0 \leq i \leq u-1$ (see 
 the proof of lemma  \ref{elm-reg} in Appendix  \ref{app:pref-suf}). 
\item The number of propositions $\merge(i,j)$is $\leq 2^n \times 2^n$. 
 \item  Thus, the number of symbols needed is $2^{|\Sigma|} \times u \times t \times {2^n}^{2^n} \times (2^n \times 2^n)$.
 \end{enumerate}

%Counting the size of the propositions :  The no. of states of the DFA will be at most $2^n$ and the size of the transition is $2^n \times 2^{|S|}$.
%  	Each thread can be at one of the states. There are $2^n$ threads and each one can take $2^n$ values thus giving us $2^{n\times 2^n}$ possible sequences (or propositions.  There are $t$ different re's and thus the we need $t \times 2^{n\times 2^n}$ extra alphabets. We mark each integer point with one of the $c_0,\ldots c_u$(check whether its $u$ or $u-1$) thus $u+1$ more possible symbols. Thus the size of the propositions due to all the temporal definition is at most  $2^{n\times 2^n} \times u \times t$.
%  	
\item  Next, we count the size of the formulae needed while constructing  the equisatisfiable $\mtl$ formula. 
\begin{enumerate}
\item  For each temporal definition, we define the formulae  $\Next(\Th_i(q_x))$  for each thread $\Th_i$.
 The argument of $\Next$ can take at most $2^n$ possibilities ($2^n$ states of a DFA) on each of the $2^n$ threads.
  Thus, the  total number of $\Next(\Th(q))$ formulae is $2^n \times 2^n = \mathcal{O}(poly(2^n))$. 
  Note that each $\Next$ formulae simulates the transition  function of the DFA.
 $\Next(\Th_i(q'))$ is determined depending on the present state $q$ of the thread $\Th_i$, and the formulae (in $\mathcal{S'}$) that are true at the present point.    Thus, the size of each $\Next$ formulae is $2^n \times 2^{|S|}$. 
 Thus, the total size of all the $\Next$ formulae is $\mathcal{O}(poly(2^n)) \times \mathcal{O}(poly(2^n \times 2^{|S|})) = \mathcal{O}(poly(2^{n+|S|})$.
\item 	Next, we look at formulae $\mathsf{NextMerge(i,k)}$.
Note that both the arguments refer to threads, and hence can  take at most $2^n$ values. Thus,
 the total number of formulae is $2^n \times 2^n = \mathcal{O}(poly(2^n)$. Each $\mathsf{NextMerge}$ formulae 
 checks whether the states at the 2 threads $\Th_i, \Th_k$ are equal or not. 
 Thus, the  size of each formulae is $\mathcal{O}(2^n)$. The total blow up due to $\mathsf{NextMerge}$ formulae is hence, 
 $\mathcal{O}(2^n \times 2^n \times 2^n) = \mathcal{O}poly(2^n)$.
 \item Next, we look at formulae $\mathsf{MergeseqPref}(k_1)$.
 This formulae states all the possible merges from the present point to the integer point within the interval $(l-1,l)$. 
 There are at most $2^n$ merges possible, as the merge always happens from a higher indexed thread to a lower one. 
 The number of merges is equal to the nesting depth of the formula $\mathsf{MergeseqPref}(k_1)$.
  Note that the nesting depth can be at most $2^n$.
 % Thus, for each possible number of merges, $0<k<2^n$, we count the number of possible merges. 
 The number of propositions $\merge(i,j)$ is $2^n\times 2^n$. Let there be $k \leq 2^n$ merges 
 until we see the integer point in $(l-1,l)$. At each of these $k$ merges, we have $2^n \times 2^n$ possibilities, 
 the maximum possible number of propositions $\merge(i,j)$ ($i, j \leq 2^n$).  
    Hence,  the number of possible merge sequences
   we can generate is $(2^n\times 2^n)^ k \le (2^n \times 2^n)^{2^n}$. There are $2^n$ possible values of $k$ and the possible number of disjunctions of the formulae is at most $(2^n \times 2^n)^{2^n} \times 2^n = \mathcal{O}(poly(2^{poly(2^n)}))$. 
  \item  The counting for $\mathsf{MergeseqSuf}(k_1,k)$ is symmetric. 
  
\end{enumerate}
 \end{enumerate}
  Adding all the blow ups due to  various formulae $\Next(\Th(q))$, $\mathsf{NextMerge(i,k)}$, $\mathsf{MergeseqPref}(k_1)$ and 
 $\mathsf{MergeseqSuf}(k_1,k)$, we see the number to be doubly exponential 
 $\mathcal{O}(poly(2^{poly(2^n)}))$. Thus, we obtain an $\mitl$ formula of doubly exponential size, with doubly exponential number
  of new propositions. By applying the reduction as in lemma \ref{expmitl}, 
 we will obtain a formula in $\mitl[\until_{0,\infty},\since]$, which is still 
 doubly exponential, and which preserves the max constant. The $\mathsf{PSPACE}$ procedure 
 of $\mitl[\until_{0,\infty},\since]$ thus ensures that we have a 
   2$\mathsf{EXPSPACE}$ procedure for satisfiability checking for  $\mitl+\ureg$.  Arriving at a tighter complexity for this class is an interesting problem and is open.

	\end{proof}

\subsection{Proof of Theorem \ref{mitl-ureg}.4: $\mitl+\mcnt$ is $\bf{F}_{\omega^{\omega}}$-hard} 
\label{app:th-ack}
In this section, we discuss the complexity of $\mitl+\mcnt$, proving 
Theorem \ref{mitl-ureg}.4.
To prove this, we obtain a reduction from the reachability problem of Insertion Channel Machines with Emptiness  
Testing ($\mathsf{ICMET}$). 
We now show how to encode the reachability problem of $\icmet$ in  
 $\mitl+\mcnt$. 
 
\paragraph*{Recalling $\icmet$} 
A channel machine $\mathcal{A}$ consists of  a tuple  having a finite set of states 
$S$, a finite alphabet $M$ used to write on the channels, 
a finite set $C$ of channels, and a transition relation $\Delta \subseteq S \times Op \times S$
where $Op$ is a finite set of operations on the channels. 
These operations have the forms $c!a$, $c?a$ and $c=\epsilon$ which respectively write a message 
$a$ to the tail of channel $c$, read the first message $a$ from a channel $c$, and test if channel $c$ is empty.

A configuration of the channel machine $\mathcal{A}$ is a pair 
$(s, U)$ where $s$ is a state and $U$ is a tuple of length $|C|$ 
which describes the contents of all the $|C|$ channels. Each entry in this tuple is hence a string 
over the  alphabet $M$. We use $\mathsf{Conf}$ to denote the configurations of the channel machine. 
The configurations are connected to each other depending on the operations performed. In particular,
 \begin{enumerate}
 \item[(a)] From a configuration $(q,U)$, the transition 
 $(q,c!a,q')$  results in a configuration 
 $(q', U')$ where $U'$ is the $|C|$-tuple which does not alter 
 the contents of channels other than $c$, and appends $a$ to channel $c$. 
 \item[(b)] From a configuration $(q,U)$,
 the transition  $(q,c?a,q')$ 
results in a configuration 
 $(q', U')$ where $U'$ is the $|C|$-tuple which does not alter 
 the contents of channels other than $c$, and reads $a$ from the head of channel $c$.
  \item[(c)] From a configuration $(q,U)$, the transition 
  $(q,c=\epsilon,q')$ results in the configuration $(q', U)$ if channel $c$ is empty.
  The contents of all the channels are unaltered.  If channel $c$ is non-empty, then the machine 
  is stuck.
   \end{enumerate}
 If the only transitions allowed are as above, then we call  $\mathcal{A}$ an error-free channel-machine. 
 We now look at channel machines with insertion errors. These allow extra transitions between configurations as follows.
 \begin{enumerate}
 \item[(d)] If a configuration $(q, U)$ can evolve into $(q', V)$ using one transition 
 as above,  then we allow any configuration $(q, U')$, where $U'$ is a $|C|$-tuple of words 
 obtained by deleting any number of letters from any word in $U$, 
 to evolve into $(q', V')$ where $V'$ is obtained by adding any number of letters 
 to any word in $V$. Thus insertion errors are  created by 
 inserting arbitrarily many symbols into some word.
  \end{enumerate}
The channel machines as above are called $\icmet$. A run of an $\icmet$ is a 
sequence of transitions $\gamma_0 \stackrel{op_0}{\rightarrow} \gamma_1 \dots \stackrel{op_{n-1}}{\rightarrow}\gamma_n \dots$ that is consistent with the above operational semantics.

\subsection*{Reduction from $\icmet$ reachability to satisfiability of $\mitl+\mcnt$} 
Consider any $\icmet$ ${\cal C}=(S, M, \Delta, C)$, with 
set of states $S=\{s_0,\ldots,s_n\}$ and channels $C=\{c_1,\ldots,c_k\}$. Let  $M$ be a finite set of messages used for communication 
in the channels.

We encode the set of all possible configurations of $\cal C$, with a timed language over the alphabet 
$\Sigma = M_{a}\cup M_b \cup \Delta \cup S \cup \{H\}$, where $M_{a} = \{m_a | m \in M\}$ $M_b = \{m_b | m \in M\}$, 
and $H$ is a new symbol. 
 
\begin{enumerate}
	\item  The $j$th configuration for  $j \geq 0$  is encoded in the interval $[(2k+2)j, (2k+2)(j+1)-1 )$ where
	$k$ refers to number of channels.
The $j$th configuration begins at the time point $(2k+2)j$. At a distance   
	$[2i-1,2i]$ from this point, $1 \leq i \leq k$, 
the contents of the $i^{th}$ channel are encoded as shown in the point 7. The intervals of the form $(2i,2i+1)$, $0 \leq i \leq k+1$   from the start of any configuration are time intervals within which no action takes place. 
The current state at the $j$th configuration is encoded at $(2k+2)j$, and the transition 
that connects configurations $j, j+1$ is encoded at $(2k+2)j+(2k+1)$. 

%	\item 
%	At time $(2k+2)j+(2k-1)$, the current state $s_w$ of the $\icmet$ at configuration $j$ is encoded by the truth of the proposition $s_w$.
	
\begin{figure}[h]
\includegraphics[scale=0.45]{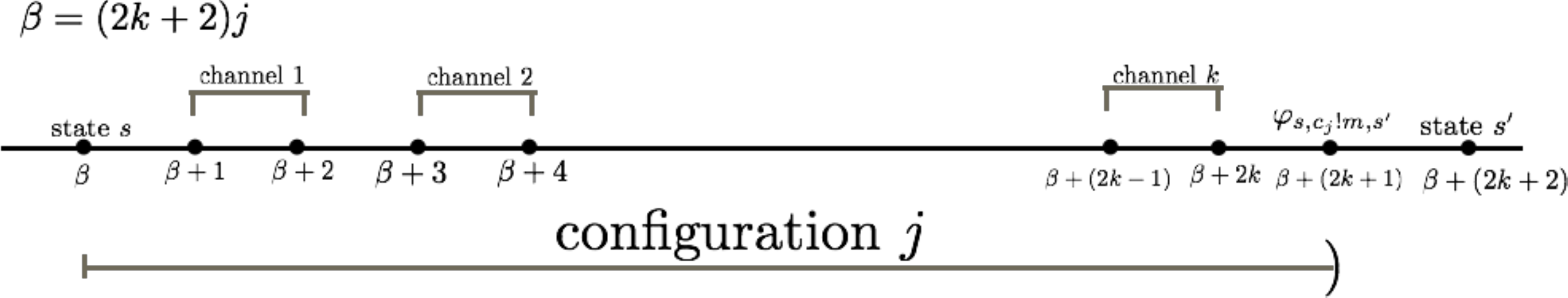}
\caption{Illustrating the $j$th configuration, with the current state encoded 
at $(2k+2)j$, and transition between configurations $j, j+1$ encoded at 
 $(2k+2)j+(2k+1)$, and the contents of channel $i$ encoded in the interval $(2k+2)j+[2i-1,2i]$.}	
\end{figure}

		\item  Lets look at the encoding of the contents of channel $i$ in the $j$th configuration.
Let $m_{h_i}$ be the message at the head of the channel $i$. Each message $m_i$ is encoded using consecutive occurrences of symbols $m_{i,a}$ and $m_{i,b}$. 
In our encoding of channel $i$, 	the first point marked $m_{h_i,a}$ in the interval $(2k+2)j+[2i-1,2i]$ is the head of the channel $i$ and denotes that $m_{h_i}$ is the message at the head of the channel.  The last point marked $m_{t_i,b}$ in the interval is the tail of the channel, 
	and denotes that message $m_{t_i}$ is the message stored at the tail of the channel.

	\begin{figure}[h]
\includegraphics[scale=0.45]{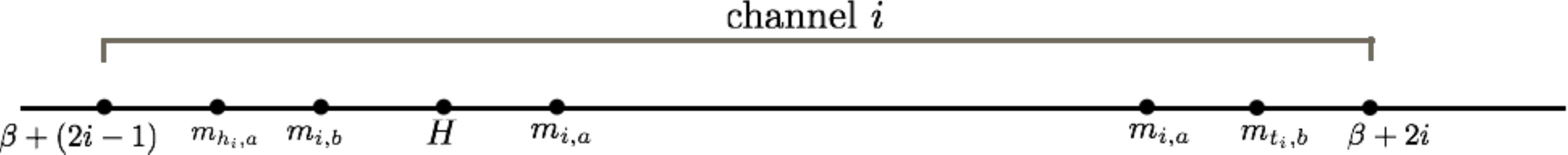}
\caption{Illustrating the channel contents with each message $m_i$ encoded as $m_{i,a}m_{i,b}$. $H$ is a separator for the head of the channel.}	
\end{figure}

	\item  
	Exactly at $2k+1$ time units after the start of the $j^{th}$ configuration, we encode the transition
	from the state at the $j^{th}$ configuration to the  $(j+1)^{st}$ configuration (starts at $(2k+2)(j+1)$).
		Note that the transition has the form $(s,c!m,s')$ or $(s,c?m,s')$ or $(s,c=\epsilon,s')$. 
	\item  We introduce a special symbol $H$, which acts as separator between the head
	of the message and the remaining contents, for each channel. 
	\item A sequence of messages $w_1w_2w_3\ldots w_z$ in any channel is encoded as a sequence 
	\\ $w_{1,a}w_{1,b} H  w_{2,a}w_{2,b}w_{3,a}w_{3,b}\ldots w_{z,a}w_{z,b}$. 
\end{enumerate}

Let $S = \bigvee_{i=0}^n s_i$ denote the states of the $\icmet$, 
$\alpha = \bigvee_{i=0}^m \alpha_i$, denote the transitions $\alpha_i$ of the form $(s,c!m,s')$ or $(s,c?m,s')$ or $(s,c=\epsilon,s')$. 
Let $action = true$ and let $M_a = \bigvee_{m_x \in M} (m_{x,a})$, $M_b = \bigvee_{m_x \in M} (m_{x,b})$, with $M = M_a \vee M_b$. 

\begin{enumerate}
	\item All the states must be at distance $2k+2$ from the previous state (first one being at 0) and all the propositions encoding transitions must be at the distance $2k+1$ from the start of the configuration.
\begin{quote}
$\varphi_{S}  {=}   ~s_0 \wedge \Box [S \Rightarrow \{\fut_{(0,2k+2]} (S) \wedge \Box_{(0,2k+2)}(\neg S)\wedge 
 \fut_{(0,2k+1]} \alpha \wedge \Box_{[0,2k+1)} (\neg \alpha) \wedge \fut_{(2k+1,2k+2)} (\neg \alpha)\}]$
\end{quote}
	\item  All the messages are in the interval $[2i-1,2i]$ from the start of configuration.  No action takes place at $(2i-2,2i-1)$ from the start of any configuration. 
	\begin{quote}
		$
		\varphi_{m} {=} ~\Box\{S {\Rightarrow} \bigwedge_{i=1}^k \Box_{[2i-1,2i]}(M {\vee} H) {\wedge}
		$
		 $\Box_{(2i-2,2i-1)} (\neg action)\}$
	\end{quote}
	
	\item  Consecutive source and target states must be in accordance with a transition $\alpha$. For example,  $s_j$ appears consecutively after $s_i$
	reading $\alpha_i$ iff $\alpha_i$ is of the form $(s_i,y,s_j) \in \Delta$, with $y \in \{c_i!m, c_i?m,c_i=\epsilon\}$.
	\begin{quote}
		$
		\varphi_{\Delta} {=} \bigwedge_{s,s' \in S} \Box\{(s {\wedge} \fut_{(0,2k+2]} s') {\Rightarrow} (\fut_{(0,2k+1]} {\bigvee \Delta_{s,s'}})\}
		$
		where $ \Delta_{s,s'} $ are possible $\alpha_i$ between $s,s'$.
	\end{quote}
	\item We introduce a special symbol $H$ along with other channel contents which acts as a separator between the
	head of the channel and rest of the contents. Thus $H$ has following properties
	\begin{itemize}
		\item There is one and only one time-stamp in the interval $(2i-1,2i)$  from the start of the configuration where $H$ is true. The following formula says that there is an occurrence of a $H$:
		\begin{quote}
			$
			\varphi_{H_1} {=}  \Box[(S {\wedge} \fut_{(2i-1,2i)} M) {\Rightarrow} (\bigwedge_{i=1}^{k} \fut_{(2i-1,2i)} (H))]
			$
		\end{quote}
		The following formula says that there can be only one $H$:
			$
			\varphi_{H_2} {=} \Box(H {\Rightarrow} \neg \fut_{(0,1)} H)
			$
		\item Every message $m_x$ is encoded by truth of proposition $m_{x,a}$ immediately followed by $m_{x,b}$. Thus for any message $m_x$, the configuration encoding the channel contents has a sub-string of the form $(m_{x,a}m_{x,b})^*$ where $m_x$ is some message in $M$.
		\begin{quote}
	    $\varphi_{m} {=} \Box [m_{x,a} {\Rightarrow} \nex_{(0,1]} m_{x,b} ] {\wedge}
	    \Box [m_{x,b} {\Rightarrow} \nex_{(0,1)} M_{a}$
	    $ {\vee} \nex(\bigvee \Delta \vee H) ] {\wedge}	 (\neg M_b \until M_a)$
	    \end{quote}
		\item If the channel is not empty (there is at least one message $m_a m_b$ in the interval $(2i-1,2i)$ corresponding to channel $i$ contents) then there is one and only one $m_b$ before $H$. 
		The following formula says that there can be at most one $m_b$ before $H$.
		\begin{quote} 
			$
			\varphi_{H_3} {=} \Box[\neg \{M_b \wedge \fut_{(0,1)}(M_b \wedge \fut_{(0,1)} H)\}]
			$
\end{quote}
			The following formula says that there is one $M_b$ before $H$ in the channel, if the channel is non-empty.
		\begin{quote}	$
			\varphi_{H_4} {=} \Box[S {\Rightarrow} \{\bigwedge_{j=1}^k(\fut_{[2j-1,2j]}(M_b) {\Rightarrow} $
			$\fut_{[2j-1,2j]}(M_b \wedge \fut_{(0,1)}H))\}]
			$
			\end{quote}
		Let $\varphi_H {=}   \varphi_{H_1} \wedge \varphi_{H_2} \wedge \varphi_{H_3} \wedge \varphi_{H_4}$.
		
	\end{itemize}

	\item Encoding transitions:
	\begin{itemize}
	\item[(a)] We first define a macro for copying the contents of the $i^{th}$ channel to the next configuration with insertion errors. 
	If there were some $m_{x,a},m_{x,b}$ at times $t,t'$, $m_{x,b}$ is copied to 
		$t''+2k+2$ (where $t'' \in [t,t')$), representing the channel contents in the next configuration. This is specified by means of an even count check.
		 From any 3 consecutive points $u,v,w$ such that $m_{x,a}$ and $m_{x,b}$ are true at $v$ and $w$, respectively, 
		we assert that  there are even (or odd) number of $m_{x,b}$ within $(0,2k+2)$ from both $v$ and $w$. 
		 This implies that there must be an odd number of $m_{x,b}$'s within time interval $[\tau_v+2k+2,\tau_w+2k+2]$. Thus, there must be at least one $m_{x,b}$ copied from the point $w$ to some point in the interval $[\tau_v+2k+2,\tau_w+2k+2]$. The rest of the even number of erroneous  $m_{x,b}$ in $[\tau_v+2k+2,\tau_w+2k+2]$, along with the arbitrary insertion errors within $[\tau_u+2k+2,\tau_v+2k+2]$ models the insertion error of the $\icmet$ (see Figure \ref{even-odd}). The formula $\mathsf{copy}_g$ is as follows.
		$
			 \Box_{[2i-1,2i]}[\bigwedge_{m_x \in M}(m_{x,a} {\wedge} \isEven_{(0,2k+2)}(m_{x,b})) 
			 {\Rightarrow} \nex(\isEven_{(0,2k+2)}(m_{x,b}))]\\
			{\wedge} \Box_{[2i-1,2i]}[\bigwedge_{m_x \in M}(m_{x,a} {\wedge} \neg \isEven_{(0,2k+2)}(m_{x,b})) 
			{\Rightarrow}\nex(\neg \isEven_{(0,2k+2)}(m_{x,b}))]
			$
	\begin{figure}
\includegraphics[scale=0.35]{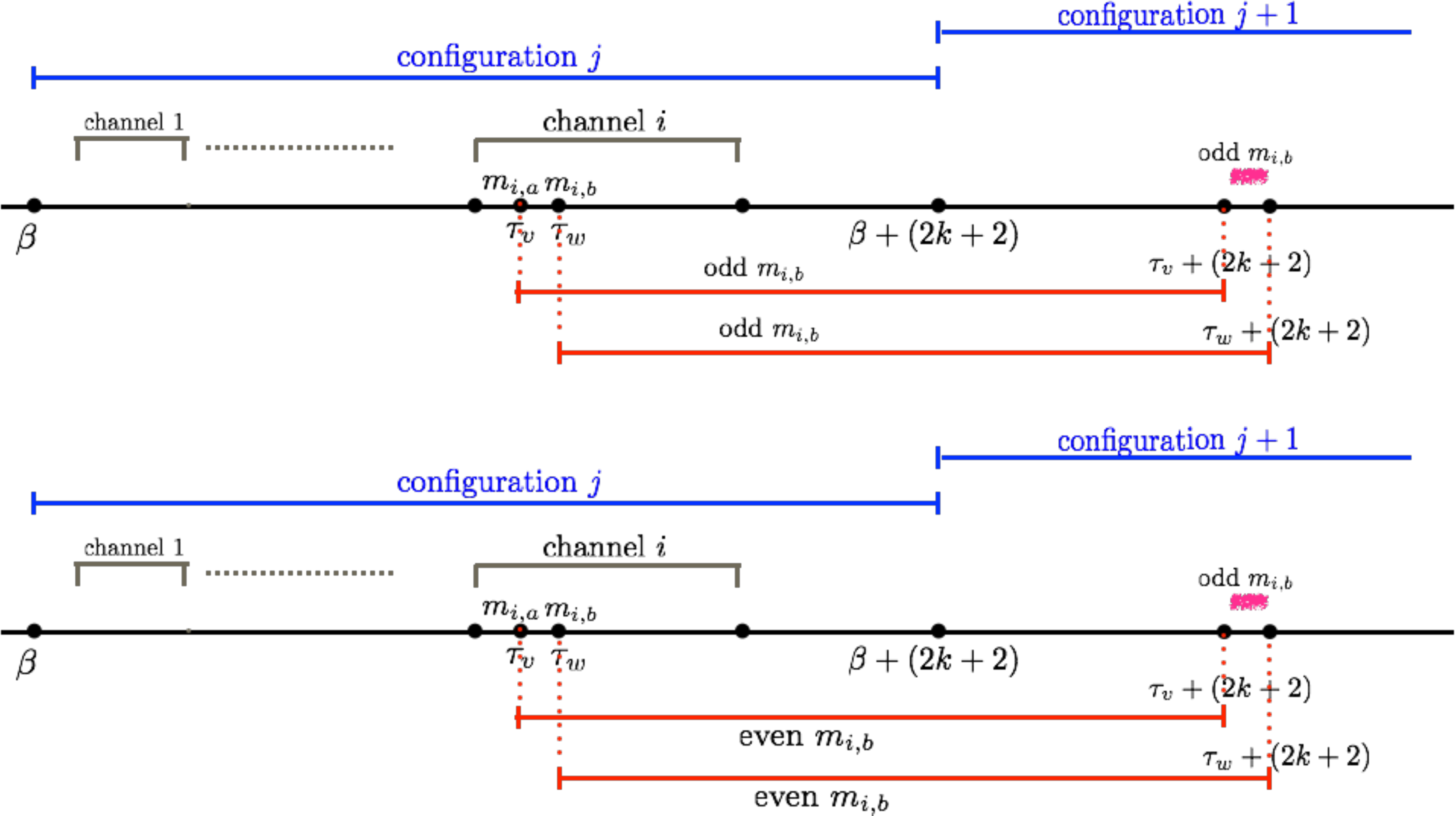}	
\caption{Copying channel contents from configuration $j$ to $j+1$. 
$\tau_v, \tau_w$ are consecutive time points labeled $m_{i,a}, m_{i,b}$. 
The points $\tau_v+(2k+2)$ and 
$\tau_w+(2k+2)$ may not necessarily be there in the word. However, there will be an odd number of $m_{i,b}$'s 
in the interval $[\tau_v+(2k+2),\tau_w+(2k+2)]$, if the number of $m_{i,b}$'s are even in 
both $[\tau_v, \tau_v+(2k+2)]$ and 
$[\tau_w,\tau_w+(2k+2)]$, or odd  in both $[\tau_v, \tau_v+(2k+2)]$ and 
$[\tau_w,\tau_w+(2k+2)]$. 
}
\label{even-odd}
\end{figure}

		\item[(b)] If the transition is of the form $c_i = \epsilon$. The following formulae checks that there are no events in the interval 
		$(2i-1,2i)$ corresponding to channel $i$, 
		while all the other channel contents are copied.    
		
		\begin{quote}
			$
			\varphi_{c_i=\epsilon} {=} S \wedge \Box_{(2i-1,2i)}(\neg action) {\wedge} \bigwedge \limits_{g =1}^{k} \mathsf{copy}_g
			$
		\end{quote}
		\item[(c)] If the transition is of the form $c_i!m_x$ where $m \in M$. An extra message is appended to the tail of channel $i$, and all the $m_a m_b$'s are copied to the next configuration. 
		$M_b \wedge \Box_{(0,1)}(\neg M))$ denotes the last time point of channel $i$; if this occurs at time $t$, 
		we know that this is copied at a timestamp strictly less than $2k+2+t$ (by 5(a)). Thus we assert that 
		truth of $\fut_{(2k+2,2k+3)}m_{x,b}$ at $t$. 
		\begin{quote}
			$
			\varphi_{c_i!m} {=} S {\wedge} \bigwedge \limits_{g =1}^{k} \mathsf{copy}_g
			{\wedge} \fut_{[2i-1,2i)}\{(M {\wedge} 
			\Box_{(0,1)}(\neg M)) {\Rightarrow} (\fut_{(2k+2,2k+3)}(m_{x,b}))\}$
			
		\end{quote}
		
		\item[(d)] If the transition is of the form $c_i?m$ where $m \in M$. The contents of all channels other than $i$ are copied to the intervals encoding corresponding channel contents in the next configuration. We also check the existence of a first message 
		in channel $i$; such a message has a $H$ at distance $(0,1)$ from it. 
				\begin{quote}
			$
			\varphi_{c_i?m_x} {=} S {\wedge} {\bigwedge \limits_{j \ne i,g=1}^k} \mathsf{copy}_g {\wedge} \fut_{(2i-1,2i)}\{m_{x,b} 
			{\wedge} \fut_{(0,1)}(H)\} 
			{\wedge} \\ 
			\Box_{[2i-1,2i]}[{\bigwedge_{m_x \in M}}(m_{x,a} {\wedge}
			 \isEven_{(0,2k+2)}(m_{x,b})
			 {\wedge} \neg \fut_{(0,1)} H) 
			{\Rightarrow} 
			\nex(\isEven_{(0,2k+2)}(m_{x,b}))]{\wedge} \\
			\Box_{[2i-1,2i]}[{\bigwedge_{m_x \in M}}(m_{x,a} {\wedge} \neg \isEven_{(0,2k+2)}(m_{x,b})
			 {\wedge} \neg \fut_{(0,1)} H) {\Rightarrow} \nex(\neg \isEven_{(0,2k+2)}(m_{x,b}))]
			$
		\end{quote}
		
	\end{itemize}
	\item Channel contents must change in accordance to the relevant transition.  Let $L$ be a set of labels (names) for the transitions. 
	Let $l \in L$ and $\alpha_l$ be a transition labeled $l$. 
	\begin{quote}
		$
		\varphi_C = \Box[S \Rightarrow \bigwedge_{l \in L}(\fut_{(0,2k+1]} (\bigvee \alpha_l \Rightarrow \varphi_l))]
		$
		\end{quote}
where $\varphi_l$ are the formulae  as seen in 5 above ($\varphi_{c_i?m_x}, \varphi_{c_i!m}, \varphi_{c_i=\epsilon}$).  
		\item Let $s_t$ be a state of the $\icmet$ whose reachability we are interested in. We check $s_t$ is reachable from $s_0$ :
		$
		\phi_{reach} = \fut(s_t)
		$
	
	Thus the formula encoding $\icmet$ is:
		$ \varphi_{S} \wedge \varphi_{\Delta} \wedge \varphi_{m} \wedge \varphi_{H} \wedge \varphi_{C} \wedge \varphi_{reach}$
	\end{enumerate}

This is a formula in $\mitl+\UM$, and we have reduced the reachability problem 
of $\icmet$ with insertion errors to 
checking satisfiability of this formula.

\section{Non-punctual 1-$\mathsf{TPTL}$ is $\mathsf{NPR}$}
\label{app:1-tptl-hard} 
In this section, we show that non-punctuality does not provide any benefits in terms of complexity 
of satisfiability for $\mathsf{TPTL}$ as in the case of $\mathsf{MITL}$. 
We show that satisfiability checking of non-punctual  $\mathsf{TPTL}$ is itself non-primitive recursive. 
This highlights the importance of our oversampling reductions from $\regmtl$ and $\regmitl$  
to $\mtl$ and $\mathsf{MITL}$ respectively, giving $\mathsf{RegMITL}$ an elementary complexity. 
It is easier to reduce $\regmitl[\ureg]$ to 1-variable, non-punctual, $\mathsf{TPTL}$ without using oversampling, but this  
gives a non-primitive recursive bound on complexity. Our reduction of $\regmitl[\ureg]$ to equisatisfiable $\mitl$ using oversampling, however 
has a $\mathsf{2EXPSPACE}$ upperbound.

\subsection*{Non-punctual $\mathsf{TPTL}$ with 1 Variable ($\optptl$)}
We study a subclass of $1-\mathsf{TPTL}$ called open $1-\mathsf{TPTL}$ and  denoted as $\optptl$. The restrictions 
are mainly on the form of the intervals used in comparing the clock $x$ as follows:  
\begin{itemize}
	\item Whenever the single clock $x$ lies in the scope of even number of negations, $x$ 
	is compared only with open intervals, and 
	\item Whenever the single clock $x$ lies in the scope of an odd number of negations, 
	$x$ is compared to a closed interval. 
\end{itemize}
Note that this is a {\bf stricter restriction} than non-punctuality as it can assert a property only within an open timed region. 
Our complexity result hence applies to $\tptl$ with non-punctual intervals. 
 Our hardness result uses a reduction from counter machines.

\subsection*{Counter Machines}
A deterministic $k$-counter machine is a  $k+1$ tuple ${\cal M} = (P,C_1,\ldots,C_k)$, where
\begin{enumerate}
	\item  $C_1,\ldots,C_k$ are counters taking values in  $\mathbb{N} \cup \{0\}$ (their initial values  are set to zero);
	\item  $P$ is a finite set of instructions with labels $p_1, \dots, p_{n-1},p_n$. 
	There is a unique instruction labelled HALT. For $E \in \{C_1,\ldots,C_k\}$, the instructions $P$ are of the following forms:
	\begin{enumerate} 
		\item  $p_i$: $Inc(E)$, goto $p_j$, 
		\item  $p_i$: If $E =0$, goto $p_j$, else go to $p_k$, 
		\item    $p_i$: $Dec(E)$, goto $p_j$,
		\item  $p_n$: HALT. 
	\end{enumerate}
\end{enumerate}
A configuration $W=(i,c_1,\ldots,c_k)$ of ${\cal M}$ is given by the value of the current program counter $i$ and values $c_1,c_2,\ldots,c_k$ of the counters $C_1,C_2,\ldots,C_k$.  A move of the  counter machine $(l,c_1,c_2,\ldots,c_k) \rightarrow (l',c_1',c_2',\ldots,c_k')$ denotes that configuration $(l',c_1',c_2',\ldots,c_k')$ is obtained from $(l,c_1,c_2,\ldots,c_k)$ by executing the $l^{th} $ instruction $p_l$.
If $p_l$ is an increment or decrement instruction, $c'_l=c_l+1$ or $c_l-1$, while $c'_i=c_i$ for $i \neq l$ and $p'_l$ is the respective 
next instruction, while 
if $p_l$ is a zero check instruction, then $c'_i=c_i$ for all $i$, and $p'_l=p_j$ 
if $c_l=0$ and $p_k$ otherwise.

\subsection*{Incremental Error Counter Machine ($\iecm$)}
An incremental error counter machine ($\iecm$) is a counter machine where a particular configuration can have counter values with arbitrary positive error.
Formally, an  incremental error $k$-counter machine is a $k+1$ tuple ${\cal M} = (P,C_1,\ldots,C_k)$ where $P$ is a set of instructions like above and $C_1$ to $C_k$ are the counters.  The difference between a counter machine with and without incremental counter error is as follows:
\begin{enumerate}
\item Let $(l,c_1,c_2\ldots,c_k) \rightarrow (l',c_1',c_2'\ldots,c_k')$ be a move of a counter machine without error when executing $l^{th}$ instruction.
\item The corresponding move in the increment error counter machine is 
$$(l,c_1,c_2\ldots,c_k) \rightarrow \{(l',c_1'',c_2''\ldots,c_k'') | c_i'' \ge c_i', 1 \leq i \leq k \}$$
 Thus the value of the counters are non deterministic. We use these machines for proving lower bound complexity in section \ref{open-tptl}.
\end{enumerate}
\begin{theorem}
	\label{theo:minsky}
	\cite{minsky} The halting problem for deterministic $k$ counter machines is undecidable for $k \ge 2$.
\end{theorem}
\begin{theorem}
	\label{theo:lazic}
	\cite{demriL06} The halting problem for incremental error $k$-counter machines is non primitive recursive for $k \ge 5$.
\end{theorem}

\subsection{Satisfiability Checking for $\mathsf{\optptl}$}
\label{open-tptl}

\begin{theorem}
Satisfiability Checking of  $\optptl[\fut,\nx]$ is decidable with  non primitive recursive lower bound over finite timed words. %and  it is undecidable over infinite timed words. 
\end{theorem} 
\begin{proof}
We encode the runs of  $k$ counter incremental error counter machine using $\optptl$ formulae with $\fut, \nx$ modalities. 
We will encode a particular computation of any counter machine using timed words. The main idea is to construct a $\optptl[\fut,\nx]$ formula $\varphi_{\iecm}$ for any given $k$-incremental counter machine $\iecm$ such that $\varphi_{\iecm}$ is satisfied by only those timed words that encode the halting computation of $\iecm$. 
Moreover, for every halting computation $\mathcal{C}$ of the $\iecm$, at least one timed word $\rho_C$ encodes $\mathcal{C}$
and satisfies $\varphi_{\iecm}$. 

We encode each computation of a $k$-incremental counter machine  $(P,C)$ where $P = \{p_1,\ldots,p_n\}$ is the set of instructions and $C = \{c_1,\ldots,c_k\}$ 
is the set of counters   using 
timed words over the alphabet $\Sigma_{\iecm}= \bigcup_{j \in \{1,\ldots,k\}} (S \cup F \cup \{a_j,b_j\})$ where $S = \{s^{p}|p \in {1,\ldots,n}\}$ and $F = \{f^{p}|p \in {1,\ldots,n}\}$ as follows:
The  $i^{th}$ configuration, $(p,c_1,\ldots,c_k)$ is encoded  in the timed region $[i,i+1)$ with the sequence 
\begin{quote}
	$s^{p} (a_1b_1)^{c_1} (a_2b_2)^{c_2} \ldots(a_kb_k)^{c_k} f^{p}$.
\end{quote}

\begin{figure}[h]
\includegraphics[scale=0.5]{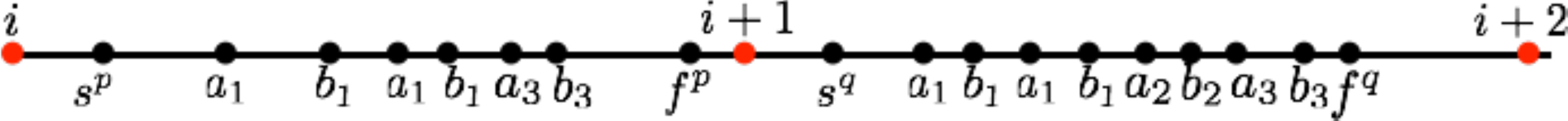}	
\label{conf1}
\caption{Assume there are 3 counters, and that the $i$th configuration is $(p, 2,0,1)$.  Let the instruction $p$ increment counter 2 and go to 
instruction $q$. Then the $i+1$st configuration is $(q,2,1,1)$. Note that the $i$th configuration 
is encoded between integer points $i, i+1$, while configuration $i+1$ is encoded between 
integer points $i+1, i+2$.}
\end{figure}

The concatenation of these time segments of a timed word encodes the whole computation.
Untiming our language yields the language 
	$$ (\mathcal{S} (a_1b_1)^* (a_2b_2)^* \ldots (a_kb_k)^*\mathcal{F})^*$$

where $\mathcal{S} = \bigvee \limits_{p \in \{1,2,\ldots,n\}} s^p$ and $\mathcal{F} = \bigvee \limits_{p \in \{1,2,\ldots,n\}} f^p$.

To construct a formula $\varphi_{\iecm}$, the main challenge is to propagate the behaviour from the time segment 
$[i,i+1)$ to the time segment $[i+1,i+2)$ such that the latter encodes the $i+1^{th}$ configuration of the $\iecm$ in accordance with the counter values of the $i^{th}$ configuration. 
The usual idea is to copy all the $a$'s from one configuration to another using punctuality. This is not possible in a non-punctual logic.
We preserve the number of $a$s and $b$s using the following idea:
\begin{itemize}
	\item Given any non last $(a_j,t)(b_j,t')$ before $\mathcal{F}$(for some counter $c_j$), of a timed word encoding a computation.
	We assert that the last symbol in $(t,t+1)$ is $a_j$ and the last symbol in 
	$(t',t'+1$) is $b_j$. 
	\item We can easily assert that the untimed sequence of the timed word is of the form
		$$ (\mathcal{S} (a_1b_1)^* (a_2b_2)^* \ldots (a_kb_k)^*\mathcal{F})^*$$
	\item The above two conditions imply that there is at least one $a_j$ within time $(t+1,t'+1)$. Thus, all the non last $a_j,b_j$ are copied to the segment encoding next configuration. 
	Now appending one $a_jb_j$,two $a_jb_j$'s or no $a_jb_j$'s  depends on whether the instruction was copy, increment or decrement operation. 
\end{itemize}

	$\varphi_{\iecm}$ is obtained as a conjunction of several formulae. Let $\mathcal{S},\mathcal{F}$ be a shorthand for 
$\bigwedge \limits_{p\in \{1,\ldots,n\}}s^{p}$ and  $\bigwedge \limits_{p\in \{1,\ldots,n\}}f^{p}$, respectively. We also define macros $A_j = \bigvee \limits_{w \ge j} a_w$ and $A_{k+1} = \bot$
		 We now give formula for encoding the machine. Let $\Cc=\{1,\ldots,k\}$ and $\Pp=\{1,\ldots,n\}$ be the indices of the counters and the instructions.  
	\begin{itemize}
		\item \textbf{Expressing untimed sequence}: The words should be of the form
		 $$ (\mathcal{S} (a_1b_1)^* (a_2b_2)^* \ldots (a_kb_k)^*\mathcal{F})^*$$
		This could be expressed in the formula below
		\begin{quote}

			$\varphi_1 = \bigwedge \limits_{j \in \Cc, p \in \Pp} 
			\wB[s^p \rightarrow \nex(A_1 \vee f^p)] \wedge \wB[a_j \rightarrow \nex(b_j)] \wedge \\
			~~~~~~~~~~~~~~~~~~~~~~~~~\wB[b_j \rightarrow \nex(A_{j+1} \vee f^p)] 
			\wedge 
			\wB[f^p \rightarrow \nex (\mathcal{S} \vee \wB(false))]$
		\end{quote}

	\item \textbf{Initial Configuration}: There is no occurrence of $a_jb_j$ within $[0,1]$. The program counter value is $1$.
		\begin{quote}
			$\varphi_2 = x.\{s^{1} \wedge \nex(f^{1} \wedge x \in (0,1))\} $
		\end{quote}
		\item \textbf{Copying $\mathcal{S},\mathcal{F}$}: Any $(\mathcal{S},u)$ (read as any symbol from 
		$\mathcal{S}$ at time stamp $u$) 
		 $(\mathcal{F},v)$ (read as (read as any symbol from 
		$\mathcal{F}$ at time stamp $v$)) has a next occurrence $(\mathcal{S},u')$, $(\mathcal{F},v')$ in the future such that 
		$u' - u \in (k,k+1)$ and $v' - v \in (k-1,k)$. 
		Note that this condition along with $\varphi_1$ and $\varphi_2$ makes sure that $\mathcal{S}$ and $\mathcal{F}$ occur only 
		within the intervals of the form $[i,i+1)$ where $i$ is the configuration number. Recall that $s^n,f^n$
		represents the last instruction (HALT). 
		\begin{quote}
			$\varphi_3 = [\wB x.\{(\mathcal{S} \wedge \neg s^{n}) \rightarrow \neg \fut(x \in [0,1] \wedge \mathcal{S}) \wedge \fut(\mathcal{S} \wedge x \in (1,2))\} \wedge \\
			 \wB x.\{(\mathcal{F} \wedge \neg f^{n}) \rightarrow \fut(\mathcal{F} \wedge x \in (0,1))\}]$
		\end{quote}
	Note that the above formula ensures that subsequent configurations are encoded in 
	smaller and smaller regions within their respective unit intervals, since consecutive 
	symbols from $\mathcal{S}$ grow apart from each other (a distance $>1$), while consecutive
	symbols from $\mathcal{F}$  grow closer to each other (a distance $<1$). See Figure \ref{conf2}.
	 \begin{figure}[h]
\includegraphics[scale=0.5]{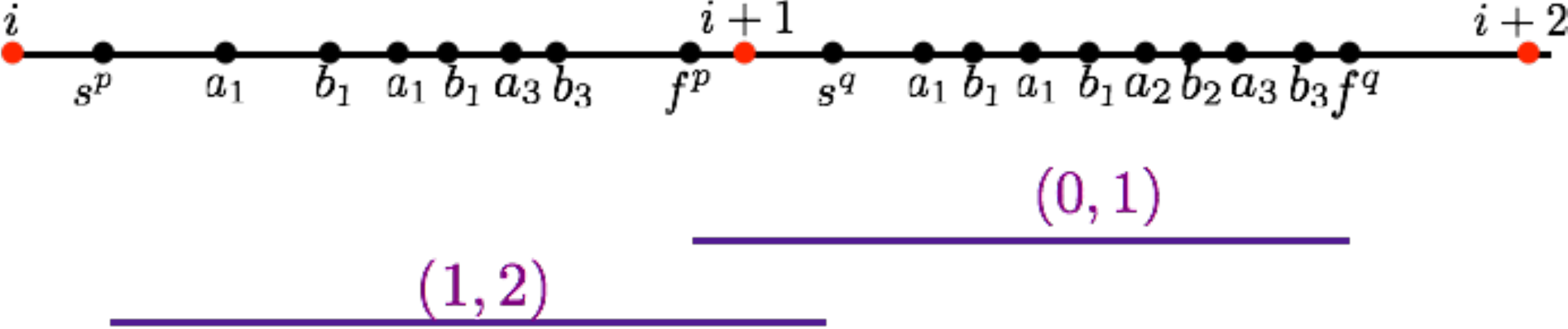}	
\caption{Subsequent configurations in subsequent unit intervals grow closer and closer.}
\label{conf2}
\end{figure}

		\item  Beyond $p_n$=HALT, there are no instructions
		\begin{quote}
		$\varphi_{4}\ =\  \wB[f^{n} \rightarrow \Box(false)]$
	\end{quote}
		
		\item At any point of time, exactly one event takes place. Events have distinct time stamps.
		\begin{quote}
		$\varphi_6\ =\ [\bigwedge \limits_{y \in \Sigma_{\iecm}}\wB[y \rightarrow
		\neg(\bigwedge \limits_{ \Sigma_{\iecm} \setminus \{y\}}(x))] \wedge \wB[\Box(false) \vee \nex(x \in (0,\infty))]
		$
		\end{quote}
		\item Eventually we reach the halting configuration $\langle p_n,c_1,\ldots,c_k \rangle$: $\varphi_6 = {\wF} s^{n}\\$
		
		\item Every non last $(a_j,t)(b_j,t')$ occurring in the interval $(i,i+1)$  should be copied in the interval $(i+1,i+2)$. We specify this 
		condition as follows:
		\begin{itemize}
		\item state that from every non last $a_j$   the last symbol within $(0,1)$ is $a_j$.
		Similarly from every non last $b_j$, the last symbol within $(0,1)$ is $b_j$.
		Thus  $(a_j,t)(b_j,t')$ will have a $(b_j,t'+1-\epsilon)$ where $\epsilon \in(0,t'-t)$. 
		
		 \begin{figure}[h]
\includegraphics[scale=0.5]{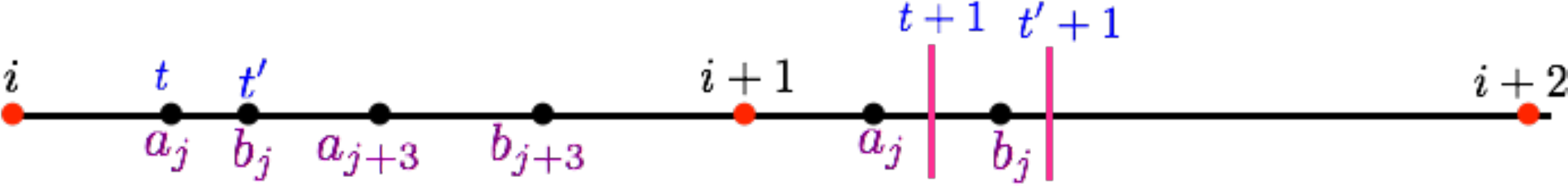}	
\caption{Consider a $a_jb_j$ where $a_j$ is at time $t$ and $b_j$ is at time $t'$. There are further $a,b$ symbols 
in the unit interval, like as shown above $a_{j+3}b_{j+3}$ occur after $a_jb_j$ in the same unit interval. 
Then the $a_j, b_j$ are copied such that the last symbol in the interval $(t, t+1)$ is an $a_j$ 
and the last symbol in the interval $(t', t'+1)$ is a $b_j$. There are no points 
between the  $a_j$ in $(i+1, i+2)$ and the time stamp $t+1$ as shown above. Likewise, 
there are no points between the $b_j$ in $(i+1, i+2)$ and the time stamp $t'+1$ as shown above. 
Note that the time stamp of the copied $b_j$ in $(i+1, i+2)$ lies in the interval $(t+1, t'+1)$.}
\label{conf3}
\end{figure}

		\item Thus all the non last
		$a_jb_j$ will incur a $b_j$ in the next configuration. $\varphi_1$ makes sure that there is an $a_j$ between two $b_j$'s. Thus this 
		condition along with $\varphi_1$ makes sure that the non last $a_jb_j$ sequence is conserved. Note that there can be some $a_jb_j$ which are arbitrarily inserted. These insertions errors model the incremental error of the machine.
	Any such inserted $(a_j,t_{ins})(b_j,t_{ins}')$ in $(i+1, i+2)$ is such that  
	there is a $(a_j,t)(b_j,t')$ in $(i, i+1)$ with $t_{ins}' \in (t+1,t'+1)$. 
	 Just for the sake of simplicity we assume that $a_{k+1}= false$.
		\end{itemize}
		
	Let $nl(a_j) = a_j \wedge \neg last(a_j)$, $nl(b_j) = b_j \wedge \neg last(b_j)$, $\psi_{nh} = \neg \fut (f^n \wedge x\in[0,1])$, \\
		$last(a_j) = a_j \wedge \nex(\nex (\mathcal{F} \vee A_{j+1})))$ and $last(b_j) =  b_{j} \wedge \nex(\mathcal{F} \vee A_{j+1})$. 
		\begin{quote}
		$\varphi_7 = \bigwedge \limits_{j \in \Cc} \wB x.[(nl(a_j) \wedge \psi_{nh}) \rightarrow \fut(a_j \wedge x \in 
			(0,1) \wedge \nex(x \in  (1,2)))]  \wedge \\
			 \wB x.[(nl(b_j) \wedge \psi_{nh}) \rightarrow \fut(b_j \wedge x \in  (0,1) \wedge \nex(x \in  (1,2)))]$
		\end{quote}
			\end{itemize}

We define a short macro 		
$Copy_{\Cc \setminus W}$: Copies the content of all the intervals encoding counter values except counters in $W$. Just for the sake of simplicity we denote 
		\begin{quote}
			        
$Copy_{\Cc  \setminus W} = \bigwedge \limits_{j \in \Cc \setminus W} \wB x.\{last(a_j) \rightarrow (a_j \wedge x \in (0,1) \wedge \nex(b_j \wedge x \in (1,2) \wedge \nex (\mathcal{F})))\}$

\end{quote}	
Using this macro we define the increment,decrement and jump operation.

\begin{enumerate}
	\item Consider the zero check instruction $p_g$: If $C_j=0$ goto $p_h$, else goto $p_d$. $\delta _1$ specifies the next configuration when the check for zero succeeds. $\delta_2$ specifies the else condition.
	\begin{quote}
		$ 
		\varphi^{g,j=0}_{8}\ =\  Copy_{\Cc \setminus \{\emptyset\}}
	\wedge \delta_1 \wedge \delta_2$
	\end{quote}
	\begin{quote}
		$\delta_1=	\wB[\{s^g \wedge ((\neg a_j) \until \mathcal{F})\} \rightarrow (\neg \mathcal{S}) \until s^h)]$
		
		$\delta_2 =   \wB[\{s^g \wedge ((\neg a_j) \until a_j )\} \rightarrow (\neg \mathcal{S}) \until s^d)]$.
		
	\end{quote}
	\item $p_g$: $Inc(C_j)$ goto $p_h$. The increment is modelled by appending exactly one $a_jb_j$ in the next interval just after the last copied $a_jb_j$
	\begin{quote}
		$ 
		\varphi^{g,inc_j}_{8}\ = Copy_{\Cc \setminus \emptyset} \wedge \wB(s^g \rightarrow (\neg \mathcal{S}) \until s^h )\wedge \psi^{inc}_{0} \wedge \psi^{inc}_{1}$
		\end{quote}
		\begin{itemize}
		\item The formula $\psi^{inc}_{0} = \wB [(s^g \wedge (\neg a_j \until f^g)) \rightarrow (\neg \mathcal{S} \until x.(s^h \wedge \fut(x \in (0,1) \wedge a_j))]$ specifies the increment of the counter $j$ when the value of $j$ is zero. 
		\item The formula
		$\psi^{inc}_{1}=\wB [\{s^g \wedge ((\neg \mathcal{F}) \until (a_j))\} \rightarrow  (\neg \mathcal{F}) \until x.\{last(a_j) \wedge \fut(x \in (0,1) \wedge (a_j \wedge \nex\nex (last(a_j) \wedge x \in (1,2))))\} ]$
	specifies the increment of counter $j$ when $j$ value is non zero by appending exactly one pair of $a_jb_j$ after the last copied $a_jb_j$ in the next interval.
		 	
		\end{itemize}

		\item $p_g$: $Dec(C_j)$ goto $p_h$. Let $second-last(a_j)= a_j \wedge \nex(\nex(last(a_j)))$. Decrement is modelled by avoiding copy of last $a_jb_j$ in the next interval. 
			\begin{quote}
			$ 
				\varphi^{g,dec_j}_{8}\ =\   Copy_{\Cc \setminus j} \wedge \wB(s^g \rightarrow (\neg \mathcal{S}) \until s^h )\wedge \psi^{dec}_0
		        \wedge \psi^{dec}_1
			$
			\end {quote}
			\begin{itemize}
			\item 	
			The formula $\psi^{dec}_0 = \wB [\{s^g \wedge (\neg a_j) \until f^g)\} \rightarrow \{(\neg \mathcal{S}) \until \{s^h \wedge ((\neg a_j) \until (\mathcal{F})\}]$ specifies that the counter remains unchanged if decrement is applied to the $j$ when it is zero. 
		\item 	The formula $\psi^{dec}_1 = \wB [\{s^g \wedge ((\neg \mathcal{F}) \until (a_j))\} \rightarrow  (\neg \mathcal{F}) \until x.\{second-last(a_j) \wedge \fut(x \in (0,1) \wedge (a_j \wedge \nex\nex ([A_{j+1} \vee \mathcal{F}] \wedge x \in (1,2))))\} ]$ decrements the counter $j$, if the present value of $j$ is non zero. It does that by disallowing copy of last $a_jb_j$ of the present interval to the next. 
			
			\end{itemize}

\end{enumerate}

The formula $\varphi_{\iecm}= \bigwedge \limits_{i \in \{1,\ldots,7\}} \varphi_i \wedge \bigwedge \limits_{p \in \Pp} \varphi^{p}_8$.

\end{proof}

\section{Details on Expressiveness}
\label{app:exp}
\begin{theorem}
 \begin{enumerate}
 \item $\mtl+\ureg \subseteq \mtl+\reg$
 \item $\mtl+\UM \subseteq \mtl+\mcnt$
 \end{enumerate}
\end{theorem}
 \begin{proof} 
 \begin{enumerate}
 \item 
 We first prove $\mtl+\ureg \subseteq \mtl+\reg$. 

%Let us recall the semantics of  $\phi_1  \ureg_{I,\re} \phi_2$.
% Given a timed word $\rho$, and  positions $i < j \in dom(\rho)$, let $\mathsf{Seg}(\Ss, i, j)$ denote 
%the untimed word over $\mathcal{P}(\mathsf{S})$
%obtained by marking  the positions $k \in \{i+1, \dots, j-1\}$ of $\rho$ with 
%$\psi \in \mathsf{S}$ iff $\rho,k \models \psi$.   
% $\rho,i \models \varphi_{1} \ureg_{I,\re(S)} \varphi_{2}$  $\leftrightarrow$  $\exists j {>} i$, 
%$\rho,j {\models}\ \varphi_{2}, \tau_{j} - \tau_{i} {\in} I$, $\rho,k\ {\models}\ \varphi_{1}$ ${\forall} i{<} k {<}j$ and, 
%$\mathsf{Seg}(\Ss, i, j) \in L(\re(S))$, where $L(\re(\Ss))$ is the language  
%of the rational expression $\re$ formed over the set $\Ss$.  This is same as saying that 
%$\rho,j {\models}\ \varphi_{2}, \tau_{j} - \tau_{i} {\in} I$ and, 
%$\mathsf{Seg}(\Ss, i, j) \in L(\re(S))$
%and $\mathsf{Seg}(\{\phi_1\}, i, j) \in \phi_1^+$. That is, 
%$\mathsf{Seg}(\Ss \cup \{\phi_1\}, i, j) \in L(\re(S)) \cap \phi_1^+$. 

Note that $\phi_1  \ureg_{I,\re} \phi_2$ is equivalent to $true \uregm^{\Ss'}_{I,\re'} \phi_2$, where $\re'$ is a regular expression obtained by conjuncting 
 	$\phi_1$ to all formulae $\psi$ occurring in the top level subformulae of $\re$, and $\Ss'=\Ss \cup \{\phi_1\}$.  
 	For example, if we had $a \ureg_{(0,1),(\reg_{(1,2)}[\reg_{(2,3)}(b+c)^*])} d$, then 
 	we obtain $true \ureg_{(0,1),(a \wedge \reg_{(1,2)}[\reg_{(2,3)}(b+c)^*])} d$. 
 	When evaluated at a point $i$, the conjunction  ensures that $\phi_1$ holds good at all 
 	the points between $i$ and $j$, where $\tau_j - \tau_i \in I$. 
 	 To reduce $true \uregm^{\Ss'}_{I, \re'} \phi_2$ to a $\reg_I$ formula,  we need the following lemma.

 	 \begin{lemma}
 		\label{reg-exp}
 		Given any regular expression $R$, there exist 
 	finitely many regular expressions $R^1_1, R^1_2, \dots, R^n_1, R^n_2$
 		such that $R=\bigcup_{i=1}^n R^i_1. R^i_2$. That is, 
 for any string $\sigma \in R$ and for any decomposition of $\sigma$ as  $\sigma_1.\sigma_2$, 
 there exists some $i\le n$ such that $\sigma_1 \in R^i_1$ and  $\sigma_2 \in R^i_2$.
 		\end{lemma}
 			\begin{proof}
Let $\mathcal{A}$ be the  minimal DFA for $R$. Let the number of states in $\mathcal{A}$ be $n$. 
 The set of strings that leads to some state $q_i$ from the initial state $q_1$ is definable by a regular expression $R^i_1$. 
 Likewise, the set of strings that lead from $q_i$ to some final state of $\mathcal{A}$ is also definable by some regular expression $R^i_2$. 
 Given that there are  $n$  states in the DFA $\mathcal{A}$, we have  $L(\mathcal{A}) = \bigcup_{i=1}^n R^i_1. R^i_2$. 
  Consider any string $\sigma \in L(\mathcal{A})$, and any arbitrary decomposition of $\sigma$ as  $\sigma_1.\sigma_2$. 
  If we run the word $\sigma_1$ over $\mathcal{A}$, we might reach at some state $q_i$. Thus $\sigma_1 \in L(R^i_1)$.
    If we read $\sigma_2$ from $q_i$, it should lead us to one of the final states (by assumption that $\sigma \in R$) . Thus $\sigma_2 \in L(R^i_2)$. 
  \end{proof}
 	
 		Lets consider $true \ureg_{I, \re'} \phi_2$ when $I=[l,u)$.\footnote{If $I=[l,l]$, then 
 		  $true \ureg_{I, \re'} \phi_2=\reg_{[0,l]}\re'.\phi_2$}
 		 	  	If $true \ureg_{[l,u),\re'} \phi_2$ evaluates to true at a point $i$, 
 	we know that $\phi_2$ holds good at some point $j$ such that $\tau_j-\tau_i \in [l,u)$, and that $[\mathsf{Seg}(\Ss',i,j)]^{\singl} \cap L(\re') \neq \emptyset$.  
 	By the above lemma,  for any word $\sigma \in L(\re')$, and   any decomposition 	  $\sigma = \sigma_1.\sigma_2$, there exist an $i\in \{1,2,\ldots,n\}$ such that  
 	 	$\sigma_1 \in L(R^i_1)$ and $\sigma_2 \in L(R^i_2)$. 
 Thus we decompose at a point $j'$ with every possible $R^k_1. R^k_2$  pair such that 
\begin{itemize}
\item  $\tau_{j'} \in \tau_i+[l,u)$,   $[\mathsf{TSeg}(\Ss,(0,l),i)]^{\singl} \cap L(R^k_1) \neq \emptyset$, 
 \item $[\mathsf{TSeg}(\Ss,[l,u),i)]^{\singl} \cap L(R^k_2).\phi_2.\Sigma^* \neq \emptyset$, 
where $\phi_2 \in \Ss$.
\end{itemize}
 Note that (i)$\phi_2$ holds good at the point $j$ such that $\tau_j \in [\tau_{i}+l, \tau_{i}+u)$,
  and, (ii)  the expression $R^k_2$ evaluates to true in $[l, \tau_j)$. We simply assert $\Sigma^*$ on the remaining part $(\tau_j,u)$
  of the interval. 
  
   Thus $true \ureg_{[l,u),\re'} \phi_2 \equiv \bigvee \limits_{i\in \{1,2\ldots,n\}}\reg_{(0,l)}R^i_1 \wedge \reg_{[l,u)}(R^i_2.\phi_2.\Sigma^*)$.
\item  	 
We first show that the $\UM$ modality can be captured by $\mcnt$.
	Consider  any formula $\phi_1 \UM_{I,\#\phi_3 = k \%n} \phi_2$. At any point $i$ this formulae is true if and only if there exists a point $j$ in future such that $\tau_j-\tau_i\in I$ and the number of points between $i$ and $j$ where $\phi_3$ is true is  $k\%n$, and $\phi_1$ is true at all points between $i$ and $j$. To count between $i$ and $j$, we can first count the behaviour $\phi_3$ from $i$ to the last point of the word, followed by the counting from $j$ to the last point of the word. 
	Then we check that the difference between these counts to be  $k\%n$.

	Let $cnt_{\phi}(x,\phi_3) = \{\phi \wedge \mcnt^{x \% n}_{(0,\infty)}(\phi_3)\}$.
	Using this macro,  $\phi_1 \UM_{I,\#\phi_3 = k \%n}\phi_2$ is equivalent to 
	$\bigvee_{k_1=0}^{n-1} [\psi_1 \vee \psi_2]$ where
	\begin{itemize}
	\item $\psi_1{=}\{cnt_{true}(k_1,\phi_3) \wedge(\phi_1 \until_I cnt_{\phi_2 \wedge \neg \phi_3}(k_2,\phi_3))\}$,
	 \item $\psi_2{=}{\{cnt_{true}(k_1,\phi_3) \wedge (\phi_1 \until_I cnt_{\phi_2 \wedge \phi_3}(k_2{-}1,\phi_3))\}}$, 
	\item $k_1 {-} k_2 {=} k$
	\end{itemize}
 	\begin{figure}[h]
 	\includegraphics[scale=0.35]{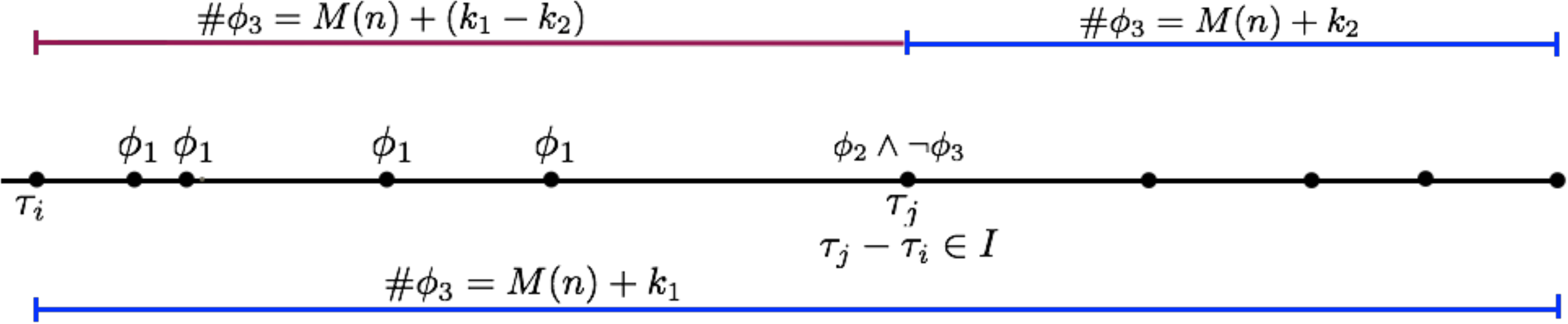}	
 	\caption{The case of $\psi_1$}
 	\end{figure}

 	\begin{figure}[h]
 	\includegraphics[scale=0.35]{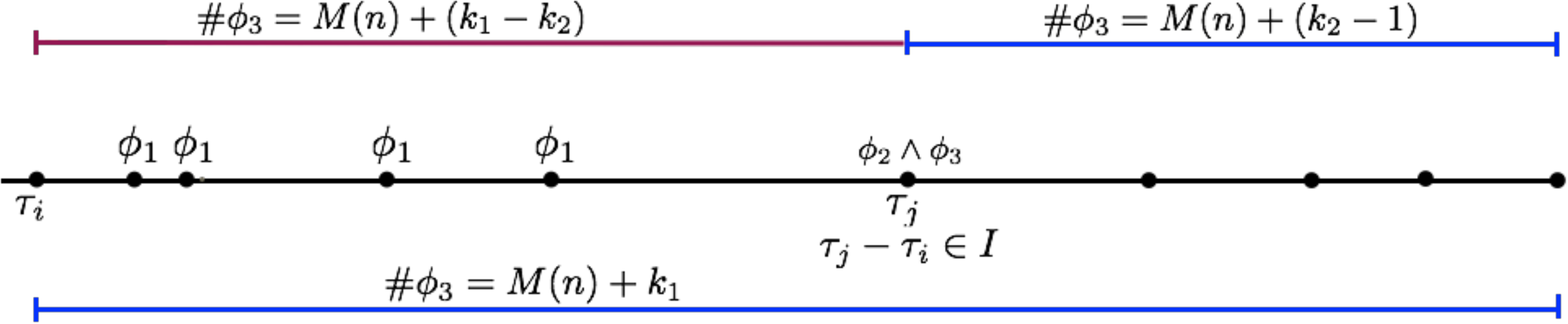}	
 	\caption{The case of $\psi_2$}
 	\end{figure}
 	
 	 The only difference between $\psi_1, \psi_2$ is that in one, $\phi_3$ holds at 
	position $j$, while in the other, it does not.  The $k_2-1$ is	 to avoid the 
	 double counting in the case $\phi_3$ holds at $j$.
	 \end{enumerate}
	 \end{proof}

\section{$\po$-1-clock ATA to  $1{-}\tptl$}
\label{app:ata-tptl-1}
In this section, we explain the algorithm which converts a  
 $\po$-1-clock ATA $\mathcal{A}$ into a  $1{-}\tptl$ formula 
 $\varphi$ such that $L(\mathcal{A})=L(\varphi)$. 
 \begin{enumerate}
\item {\bf{Step 1}}. Rewrite the transitions 
of the automaton. Each $\delta(s,a)$ can be written in an equivalent form 
as (i) $C_1\vee C_2$ or (ii) $C_1$ or (iii) $C_2$ where 
\begin{itemize}
\item $C_1$ has the form $s \wedge \varphi_1$, where $\varphi_1 \in \Phi( \downarrow s \cup \{a\} \cup X)$,
\item $C_2$ has the form $\varphi_2$, where $\varphi_2 \in \Phi(\downarrow s \cup \{a\} \cup X)$
 \end{itemize}
  In particular, if $s$ is the lowest location in the partial order, 
    then $\varphi_1, \varphi_2 \in \Phi(\{a\}  \cup X)$.  Denote this equivalent form 
    by $\delta'(s,a)$. 
    
  For the example above, we obtain 
 $\delta'(s_0,a)=(s_0 \wedge (a \wedge x.s_a)) \vee (a \wedge s_{\ell}), \delta'(s_0,b)=s_0 \wedge b$,
 $\delta'(s_a,a)=(s_a \wedge x<1) \vee (x>1)$
 $\delta'(s_{\ell})=(s_{\ell} \wedge b)$    
\item {\bf{Step 2}}. For each location $s$, construct $\Delta(s)$ which 
combines $\delta'(s,a)$ for all $a \in \Sigma$, by disjuncting them first, and again 
putting them in the  form in step 1.   
Thus, we obtain $\Delta(s)=\bigvee_a \delta'(s,a)$
which can be written as a disjunction $D_1\vee D_2$ or simply $D_1$ or simply $D_2$ where 
 $D_1,D_2$ have the forms $s \wedge  \varphi_1$ and $\varphi_2$ respectively, where  $\varphi_1, \varphi_2  \in \Phi( \downarrow s \cup \Sigma \cup X)$. 
 
 For the example above, we obtain 
 $\Delta(s_0)=(s_0 \wedge [(a \wedge x.s_a)) \vee b]) \vee (a \wedge s_{\ell})$, 
 $\Delta(s_a)=(s_a \wedge x<1) \vee (x>1)$
 $\Delta(s_{\ell})=s_{\ell} \wedge b$.
 
\item {\bf{Step 3}}. We now convert each $\Delta(s)$ into a normal form 
 $\N(s)$. $\N(s)$ is obtained  from $\Delta(s)$ as follows.
 \begin{itemize}
 \item If $s$ occurs in $\Delta(s)$, replace it  with $\nex s$.
 \item Replace each $s'$ occurring in each $\Phi_i(\downarrow s)$ with 
 $\nex s'$. 
 \end{itemize}
 Let $\N(s)=\N_1\vee \N_2$, where 
 $\N_1, \N_2$ are normal forms.  Intuitively, the states appearing on the right side of each transition 
 are those which are taken up in the next step. The normal form 
 explicitely does this, and takes us a step closer to 
 $1{-}\tptl$.
 
Continuing with the example, we obtain 
$\N(s_0)=(\nex s_0 \wedge [(a \wedge x.\nex s_a)) \vee b]) \vee (a \wedge \nex s_{\ell})$
 $\N(s_a)=(\nex s_a \wedge x<1) \vee (x>1)$
 $\N(s_{\ell})=\nex s_{\ell} \wedge b$.
 \item {\bf{Step 4}}.
\begin{itemize}
\item Start with the state $s_n$ which is 
the lowest in the partial order. 

Let 
$\N(s_n)=(\nex s_n \wedge \varphi_1) \vee \varphi_2$, 
  where $\varphi_1, \varphi_2 \in \Phi(\Sigma, X)$. 
Solving $\N(s_n)$, one obtains the solution 
$\Beh(s_n)$ as $\varphi_1 \weaku \varphi_2$   
if $s_n$ is an accepting location, and as 
$\varphi_1 \wU \varphi_2$ if $s_n$ is non-accepting. 
Intuitively, $\Beh(s_n)$ is the behaviour of $s_n$: that is, 
it describes the timed words that are accepted when we start in $s_n$.
In the running example, we obtain 
$\Beh(s_{\ell})=b \weaku \bot=\wB b$ and 
$\Beh(s_a)=(x <1) \wU  x>1$.

\item Consider now some $\N(s_i)=(\nex s_i \wedge \varphi_1) \vee \varphi_2$. 
First replace each $s'$ in $\varphi_i$ with $\Beh(s')$.
$\Beh(s_i)$ is then obtained as 
 $\Beh(\varphi_1) \weaku \Beh(\varphi_2)$   
if $s_i$ is an accepting location, and as 
$\Beh(\varphi_1) \wU \Beh(\varphi_2)$ if $s_i$ is non-accepting. 

Substituting $\Beh(s_a)$ and $\Beh(s_{\ell})$ in $\N(s_0)$, we obtain \\
$(\nex s_0 \wedge [(a \wedge x.\nex \Beh(s_a)) \vee b])	 \vee (a \wedge \nex \Beh(s_{\ell}))$, solving which, we get \\ 
$\Beh(s_0)=[(a \wedge x.\nex \Beh(s_a)) \vee b]  \weaku (a \wedge \nex \Beh(s_{\ell}))$. 

\item Thus, $\Beh(s_0)$ which represents all timed words which are accepted when we start at $s_0$ is given by 
$((a \wedge (x.\nex [(x <1) \wU  x>1])) \vee b) \weaku (a \wedge \nex \wB b)$. 
 The $1{-}\tptl$ formula equivalent to $L(\mathcal{A})$ is 
then given by $\Beh(s_0)$.    
\end{itemize}
\end{enumerate}

\subsection{Correctness  of Construction}
The above algorithm is correct; that is, the 1${-}\tptl$ formula 
$\Beh(s_0)$ indeed  captures the language accepted by the $\po$-1-clock ATA.

 For the proof of correctness, we define a 1-clock ATA with a $\tptl$ look ahead. That is, 
 $\delta: S \times \Sigma \rightarrow \Phi(S \cup X \cup \chi(\Sigma \cup \{x\}))$, where $\chi(\Sigma \cup \{x\})$ is a $\tptl$ formula over alphabet
  $\Sigma$ and clock variable $x$. We allow open $\tptl$ formulae for look ahead; that is, 
  one which is not of the form $x. \varphi$. All the freeze quantifications $x.$ lie within $\varphi$. 
   The extension now allows to take a transition  $(s,\nu) \rightarrow [\kappa \wedge \psi(x)]$, where $\psi(x)$ is a $\tptl$  formula,
    if and only if the suffix of the input word with value of $x$ being $\nu$ satisfies $\psi(x)$. We induct on the level of the partial order on the states.

Base Case: Let the level of the partial order be zero. 
Consider 1-clock ATA having only one location $s_0$. 
 Let the transition function be $\delta (s_0,a)= \mathcal{B}_a (\psi_a(x), X, s_0)$ for every $a \in \Sigma$. 
 By our construction, we reduce $s_0$ into  $\Delta(s_0) = \bigvee \limits_{a \in \Sigma} [\mathcal{B}_a (\psi_a(x), X, \nex(s_0) )]$. 
 Let $\Delta(s_0) = \bigvee (P_i \wedge \psi_i(x) \wedge X_i \wedge \nex s_0) \vee \bigvee (Q_j \wedge \psi_j(x) \wedge X_j)$. 
   $\delta(s_0,a) = s_0 \wedge X_1 \wedge \psi_1(x)$  specifies that the clock constraints $X_1$ are satisfied and the suffix satisfies the formula $\psi_1(x)$ on reading an $a$. Thus for this $\delta(s_0,a)$, we have $\nex s_0 \wedge X_1 \wedge \psi_1 (x) \wedge a$ as a corresponding disjunct in $\Delta$ which specifies the same constraints on the word from the current point onwards. Thus the solution to the above will be satisfied at a point with some $x = \nu$ if and only if there is an accepting run from $s_0$ to the final configuration with $x = \nu$.

If the $s_0$ is a final location, the solution to this is $\varphi = \bigvee (P_i \wedge \psi_i(x) \wedge X_i \wedge \nex s_0) \weaku \bigvee (Q_j \wedge \psi_j(x) \wedge X_j)$. If it is non-final, then it would be $\until$ instead of $\weaku$. Note that this implies that whenever $s_0$ is invoked with value of $x$ being $\nu$, the above formula would be true with $x = \nu$ thus getting an equivalent $1{-}\tptl$ formulae. 

Assume that for automata with $n-1$ levels in the partial order, we can construct an equivalent $1-\tptl$ formula as per our construction.
Consider an automaton with $n$ levels. Consider all the locations at the lowest level (that is, those location that can  only  call itself), $s_0,\ldots,s_k$. 
Apply the same construction. As explained above, the constructed formulae, while eliminating a location will be true at a point if and only if there is an accepting run starting from the corresponding location with the same clock value. Let the formula obtained for any $s_i$ be $\varphi_i$.

The occurrence of an $s_i$ in any $\Delta(s_{i<n})$ can be substituted with $\varphi_i$ as a look ahead. This gives us an $n-1$ level 1-clock ATA with 
$ \mathsf{ TPTL}$  look ahead. By  induction, we obtain that every 1-clock $\po$-ATA can be reduced to $1{-}\tptl$ formulae.

\section{Proof of Lemma \ref{lem:sf-1tptl}}
\label{app:sf-1tptl}
				\begin{proof}
						Let $\rho$ be a timed word  such that $\rho,i \models \reg_{I}\re$.   $\re$ can be either a simple star-free  expression 
						over $\Sigma$, or can be of the form  $\reg_{I'}\re'$ or $\reg_{I_1}\re_1+\reg_{I_2}\re_2$
						or $\reg_{I_1}\re_1.\reg_{I_2}\re_2$ or $(\reg_{I'}\re')^*$. Recursively, each of $\re'$, $\re_1, \re_2$ 
						also can be expanded out as above. The idea of the proof is to eliminate ``all levels'' of $\reg_I$ 
						starting from the inner most one, by replacing them
						 with $1{-}\tptl$ formulae using structural induction.  We first explain the proof in the case of $\reg_{I}\re$, 
						 where $\re$ is a star-free expression over $\Sigma$, and then look at general cases.
						 
We first consider the case when all intervals are bounded.

\begin{enumerate}
\item 	Given $\reg_{I}\re$, and a point $i$ in a word $\rho$, 
$\reg_I$ checks $\re$ at all points $j$ in $\rho$ 
		such that $\tau_j-\tau_i \in I$.
We first eliminate the interval $I$ from $\reg_{I}\re$
by imagining a witness variable 
		$w_I$ that evaluates to true at all points $j$ of $\rho$ such that 
		$\tau_j-\tau_i \in I$. $w_I$ is used to cover all points 
		distant $I$ from $i$.
		\item We eliminate the interval $I$ in $\reg_I$ 
 by rewriting  $\reg_I \re$ as 
  $\reg( (\neg w_I)^*.\re.(\neg w_I)^*)$, which, 
   when asserted at a point $i$,  checks the truth of the expression
   $(\neg w_I)^*.\re.(\neg w_I)^*$ in the suffix from $i$.
   If $w_I$ indeed captures $\tau_j-\tau_i \in I$, then indeed we are checking $\re$ in the interval $I$.  
Let $\varphi_{\re}$ be an LTL formula that is equivalent to $\re$. This is possible since $L(\re)$ is a star-free language. 
\item We next replace $w_I$ by using a freeze clock variable $x$ which checks $x \in I$ 
whenever we assert $w_I$. 
\begin{enumerate}						 						
\item We will look at the simplest case when $\re$ is a regular expression over $\Sigma$. 	
Let $\varphi_{\re}$ be the LTL formula equivalent to $\re$. 					 						
\begin{enumerate}
\item  Let $\re=a$ for $a \in \Sigma$. 
We expand the alphabet by allowing proposition $x \in I$ (and its negation $x \notin I$). 
The formula 
$\psi=((x \notin I) \until \{\varphi_{\re} \wedge [(x\in I) \wedge   (\nex(\Box (x \notin I)))]\})$
is then an LTL formula that says that there is a single point in the region $I$, and $a$ holds at that point. 
Then $\psi_{\re}=x. \psi$ is a 1-$\tptl$ formula that captures $\regm_Ia$.
\item Let $\re=a.b$ for $b \in \Sigma$. Then we inductively assume 
LTL formulae $\varphi_a$ and $\varphi_b$ that capture $a$ and $b$. 
As above, we allow the proposition $x \in I$. 
 Then the formula
$\psi_{\re}=x.\{(x \notin I) \until [(x \in I \wedge \varphi_a \wedge \nex(x \in I \wedge \varphi_b \wedge \nex(\Box (x \notin I))))]\}$
asserts the existence of two points in the interval $I$ respectively satisfying 
in order, $a$ and $b$. This argument can be extended to work for any finite concatenation 
$\re=a_1.a_2.\dots.a_n$. 
\item Let $\re=(\re_1)^*$ be rational expression over $\Sigma$. Let $\varphi_{\re_1}$ 
 be the LTL formula equivalent to $\re_1$. Then the formula 
 $x.[(x \notin I) \until( (x \in I) \wedge \Box[(x \in I) \rightarrow \varphi_{\re_1}])]$  
is a 1-$\tptl$ formula that asserts the formula $\varphi_{\re_1}$ at all points in $I$. 
 \item Let $\re=\re_1+\re_2$ be a rational expression over $\Sigma$. 
Let $\varphi_{\re_1}$ and $\varphi_{\re_2}$ be 
LTL formulae equivalent to $\re_1, \re_2$. Then 
$((x \notin I) \until \{(\varphi_{\re_1} \vee \varphi_{\re_2})\wedge [(x\in I) \wedge   (\nex(\Box (x \notin I)))]\})$
is then an LTL formula that says that there is a single point in the region $I$, and 
one of $\varphi_{\re_1}, \varphi_{\re_2}$ holds
 at that point. 
	\end{enumerate}

%
%
%we obtain the formula $\psi_{\re}=\Box(x \in I \rightarrow \varphi_{\re})$. 
%%we conjunct each symbol in  $\varphi_{\re}$ with $x \in I$. 
%Thus, if $\varphi_{\re}=a \until b$, then we obtain 
%$\psi_{\re}=\Box(x \in I \rightarrow (a \until b))$. 
%%(a \wedge (x \in I)) \until (b \wedge (x\in I))$. 
%Let $\zeta_{\re}$ be the LTL formula $((x \notin I) \until \{\psi_{re} \wedge [(x\in I) \wedge  
%  [(x \in I) \until (\nex(\Box (x \notin I)))]]\})$
%% [\psi_{\re} \wedge ((x \in I) \until (\nex(\Box (x \notin I)))])$. 
%Then $x.\zeta_{\re}$  is a 1-$\tptl$ formula over $\Sigma$ that is equivalent to $\reg_I \re$.

\item Finish the base case of the structural induction, where $\re$ was a rational expression over $\Sigma$, we now  
move on to general cases. 
\begin{enumerate}
\item Let us now consider the case when we have a formula 
$\regm_{I_1}[\regm_{I_2}\re]$. 

Then we first obtain as seen above, 
$x.\zeta_{\re}$ equivalent to $\regm_{I_2}\re$. Let $x.\zeta_{\re}=w$. 
Then, 
$\zeta_{out}=((x \notin I_1) \until \{w \wedge (x \in I_1) \wedge 
\nex (\Box (x \notin I_1))\})$ 
%$\zeta_{out}=((x \notin I_1) \until \{w \wedge [(x \in I_1) \wedge 
%[(x \in I_1) \until (\nex (\Box (x \notin I_1)))]]\})$ 
is 
an LTL formula which asserts the existence of a single point lying in the interval $I_1$ where $w$ is true. 
Then $x.\zeta_{out}$ is a 1-$\tptl$ formula over $\Sigma$ that is equivalent to 
$\regm_{I_1}[\regm_{I_2}\re]$.

\item If we have $\reg_{I}[\regm_{I_1}\re_1. \regm_{I_2}\re_2]$, then let
$w_1=x.\zeta_{\re_1}, w_2=x.\zeta_{\re_2}$. Then the formula
$x.\{(x \notin I) \until [(x \in I \wedge w_1 \wedge \nex(x \in I \wedge w_2 \wedge \nex(\Box (x \notin I))))]\}$
asserts the existence of two points in the interval $I$ respectively satisfying 
in order, $\regm_{I_1}\re_1$ and $\regm_{I_2}\re_2$.

%Let $w$ denote the effect of conjuncting $w_2$ to  the last part 
%of $w_1$: that is, 
%$x.(x \notin I_1) \until (\psi_{\re_1} \wedge [(x \in I_1) \until (\nex (w_2 \wedge (\Box (x \notin I_1))))])$. 
%Then the formula $(x \notin I) \until (w \wedge [(x \in I) \until (\nex  (\Box (x \notin I)))])$
%is a 1-$\tptl$ formula that is equivalent to $\reg_{I}[\regm^{\Ss_1}_{I_1}\re_1. \regm^{\Ss_2}_{I_2}\re_2]$. 

\item If we have $\reg_I[(\regm_{I_1}\re_1)^*]$, then let $w_1=x.\zeta_{\re_1}$.

Let $w=(x \notin I) \until( (x \in I) \wedge \Box[(x \in I) \rightarrow w_1])$.  
Then $x.w$ is a 1-$\tptl$ formula that asserts that at all points in $I$, the formula $\regm_{I_1}\re_1$
 evaluates to true. 

\item If we have $\reg_I(\regm_{I_1}\re_1+\regm_{I_2}\re_2)$, then let 
 $w_1=x.\zeta_{\re_1}$ and $w_2=x.\zeta_{\re_2}$. Then 
 %$x.\{(x \notin I) \until((w_1 \vee w_2) \wedge x \in I \wedge [(x \in I) \until (\nex (\Box(x \notin I)))])\}$  
$x.\{(x \notin I) \until((w_1 \vee w_2) \wedge x \in I \wedge (\nex (\Box(x \notin I))))\}$  
is a 1-$\tptl$ formula that checks that $\regm_{I_1}\re_1$ or $\regm_{I_2}\re_2$ evaluates to true at the single point 
in the interval $I$.

\end{enumerate}
\end{enumerate}
						
Note that the boolean combinations like conjunction, disjunction and unary operations like negation can be handled in a straightforward way, once we are done with the above. While encountering boolean combinations, we simply combine the 1-$\tptl$ formulae obtained 
so far. 						
\end{enumerate}

%		regular expression containing formulae of the form $\reg_{I'}\re'$.  
%As a first step, we introduce an atomic proposition 
%		$w_I$ which evaluates to true at all points $j$ in $\rho$ 
%		such that $\tau_j-\tau_i \in I$. 
%		Then it is easy to see that $\rho, i \models \reg_{I}\re$ iff 
%	$\rho, i \models \reg_I (\re \wedge w_I)$, since 
%$\reg_I$ covers exactly all points which are within the interval $I$ from $i$.   							As the next step, we replace $\re$, with an atomic proposition $w$ obtaining 
%		the formula $\reg_I[w \wedge w_I]$. Assume that $I$ is a bounded interval. $\reg_I[w \wedge w_I]$  
%		is equivalent to $\reg[(w \wedge w_I). (\neg w_I)^*]$, since 
%		$\reg[]$ covers the entire suffix of $\rho$ starting at point $i$. 				Now, replace $w_I$ with the clock constraint $x \in I$, and 
%		rewrite the formula as $x.[(w \wedge (x \in I)).\Box(\neg(x \in I))]$, which 
%		is in $1-\tptl$. Note that this step also preserves equivalence of teh formulae. 
%		Replacing $w$ with $\re$ now  eliminates one level of 
%		the $\reg$ operator in the above formula. Doing the same technique as above 
%		to $\re$ which has the form $\reg_{I'}(\re')$, will eliminate one more level of $\reg$ and so on. 
%		Continuing this process will result in a $1{-}\tptl$ formula which has 
%		$k$ freeze quantifications iff the starting $\sfmtl$ formula had $k$ nestings of the $\reg$ modality. 
		   
	In case $I$ is an unbounded interval, then we need not concatenate 
		$\nex (\Box(x \notin I))$
at the end, since the time stamps of all points in the suffix lie in $I$.  
 The rest of the proof is the same. \qed	\end{proof}

\section{Proof of Lemma \ref{lem:poata-sfmtl}}
\label{app:ata-1tptl}

\noindent{\bf The Main Idea}:
Let $\mathcal{A}$ be a $\po$-1-clock ATA with locations $S=\{s_0,s_1,  \dots,s_n\}$. 
Let $K$ be the maximal constant used in the guards $x \sim c$ occurring in the transitions. 
Let $R_{2i}=[i,i], R_{2i+1}=(i, i+1), 0 \leq i < K$ and $R^+_K=(K, \infty)$ be the regions 
$\mathcal{R}$ of $x$. Let $R_h \prec R_k$ denote that region $R_h$
precedes region $R_k$.
For each location $s$, $\Beh(s)$ as seen above (also Figure \ref{last-eg}) gives the timed behaviour starting at $s$, 
using constraints $x \sim c$ since the point where $x$ was frozen. 
In example \ref{eg1}, $\Beh(s_a){=}(x <1) \wU  (x>1)$, allows symbols $a,b$ as long as $x<1$ 
keeping the control in $s_a$, has no behaviour at $x=1$, and allows 
 control to leave $s_a$ when $x>1$.
 For any $s$, we ``distribute'' $\Beh(s)$ across regions by untiming it. In example \ref{eg1}, 
  $\Beh(s_a)$ is $\wB(a \vee b)$ for regions $R_0, R_1$, it is $\bot$ for $R_2$ and 
  is $(a \vee b)$ for  $R^+_1$.  Given any $\Beh(s)$, and a pair of regions $R_j \preceq R_k$, such that
   $s$ has a non-empty behaviour in region $R_j$, and control leaves $s$ in $R_k$,
    the untimed behaviour of $s$ between regions $R_j, \dots, R_k$ is written as LTL formulae 
   $\varphi_j, \dots, \varphi_k$. This results in a ``behaviour description'' (or $\BD$ for short)
    denoted $\BD(s, R_j,R_k)$ : this is a $2K+1$ tuple with $\BD[R_l]=\varphi_l$
    for $j \leq l \leq k$, and $\BD[R]=\top$ denoting ``dont care'' for the other regions. 
    Each LTL formula $\BD(s, R_j,R_k)[R_i]$ (or $\BD[R_i]$ when $s, R_j, R_k$ are clear)
        is replaced with a star-free rational expression
    denoted $\re(\BD(s, R_j,R_k)[R_i])$. Then   $\BD(s, R_j,R_k)$ 
    is transformed into a $\sfmtl$ formula $\varphi(s,R_j,R_k)= \bigwedge_{j \leq g \leq k}\reg_{R_g}\re(\BD(s, R_j, R_k)[R_g])$. 
The language accepted by the $\po$-1-clock ATA $\mathcal{A}$ is then given by  
$\bigvee_{0 \leq j \leq  k \leq 2K} \varphi(s_0,R_j, R_k)$ where $s_0$ is the initial location, and 
the word is accepted while in region $R_k$.   This disjunction allows all possible accepting behaviours 
from the initial location $s_0$.

Each location $s$ is associated with a set of $\BD$s.  Let $\BDset(s)$ denote the of $\BD$s that are associated with 
$s$. If $s$ is the lowest location in the partial order, then $\BDset(s)=\{\BD(s, R_i, R_j) \mid R_i \preceq R_j\}$.

\noindent{\bf Computing $\BD(s,R_i, R_j)$ for a location $s$ and pair of regions $R_i \preceq R_j$}. 
The proof proceeds by first computing $\BD(s,R_i, R_j)$ for locations $s$ which are lowest in the partial order, followed 
by computing $\BD(s',R_i, R_j)$ for locations $s'$ which are higher in the order.   
For any location $s$, $\Beh(s)$ has the form $\varphi_1 \weaku \varphi_2$ 
or $\varphi_1 \wU \varphi_2$, or $\varphi$, where $\varphi, \varphi_1, \varphi_2$ 
are disjunctions of conjunctions over $\Phi(S \cup \Sigma \cup X)$, where $S$ is the set of locations 
with or without the binding construct $x.$, and $X$ is a set of clock constraints 
of the form $x \sim c$. Each conjunct 
has the form $\psi \wedge x \in R$  where $\psi \in \Phi(\Sigma \cup S)$ and $R \in \mathcal{R}$. 
Let $\varphi_1=\bigvee(P_i \wedge C_i), \varphi_2=\bigvee(Q_j \wedge E_j)$ where $P_i,Q_j \in \Phi(\Sigma \cup S)$ and $C_i,E_j \in \mathcal{R}$.
Let $\mathcal{C}$ and $\mathcal{E}$ be a shorthand notation to represent any $C_k, E_l$.  

For $R_i \preceq R_j$, and a location $s$, $\BD(s,R_i,R_j)$ is empty if $\Beh(s)$ has no constraint $x \in R_i$ occurring in $\mathcal{C}, \mathcal{E}$, and 
if control cannot exit $s$ in $R_j$. 
 If $\Beh(s)$ has no $\until, \mathsf{W}$ modalities, then 
$\BD(s, R_i, R_i)$ is computed when $\Beh(s)=\bigvee(Q_j \wedge E_j)$ and there is some $E_l$ with 
$x \in R_i$. In this case, $\BD(s, R_i, R_i)[R_i]=Q_l$, and the remaining entries are $\top$ representing ``dont care''.
If $\Beh(s)$ has $\until, \mathsf{W}$ modalities, 
then $\BD(s,R_i,R_j)$ is computed when 
(1) there is a  constraint $x \in R_i$ in $\mathcal{C}$ or $\mathcal{E}$ (this allows us to start observing the behaviour in region $R_i$) 
(2) there is a  constraint $x \in R_j$ in some $\mathcal{E}$ (this allows us to exit the control 
 location $s$ while in region $R_j$).  
 If so,  the $\BD$ $\Beh(s,R_i,R_j)$ is a $2K+1$ tuple with 
  (i) formula $\top$ in regions $R_0, \dots, R_{i-1},R_{j+1}, \dots, R_{K}^+$ (denoting dont care),
    (ii)If $C_k=E_l=(x \in R_j)$ for some $C_k, E_l$, then the LTL formula in region $R_j$ is $P_k \until Q_l$ if $s$ is not an accepting location, 
  and is $P_k \weaku Q_l$ if $s$ is an accepting location, (iii)If no $C_k$ is equal to any $E_l$ for any $k,l$, and  
  if $E_l=(x \in R_j)$ for some $l$,   then the formula in region $R_j$ is $Q_l$. If $C_m=(x \in R_i)$ for some $m$, then 
  the formula for region $R_i$ is $\wB P_m$. If there is some $C_h=(x \in R_w)$ for $i < w < j$, then 
  the formula in region $R_w$ is $\wB P_h \vee \epsilon$, where $\epsilon$ signifies the fact that 
  there may be no points in regions $R_w$. If there are no $C_m$'s such that $C_m=(x \in R_w)$  
for $R_i \prec R_w \prec R_j$, then the formulae in region $R_w$ is $\epsilon$. 
We allow $\epsilon$ as a special symbol in LTL to signify that there is no behaviour 
in a region.

 \noindent{\bf $\BD(s,R_i, R_j)$ for location $s$ lowest in po}.  Let $s$ be a location that is lowest in the partial order. 
 The locations $s_{\ell},s_a$ in Example \ref{eg1} are lowest in the partial order, and 
$\Beh(s_{\ell})=b \weaku \bot=\wB b$, $\Beh(s_a){=}[(a \vee b) \wedge (x<1)]\wU [(a \vee b) \wedge (x>1)] $.
  In general, if $s$ is the lowest in the partial order, then  
   $\Beh(s)$ has the form $\varphi_1 \weaku \varphi_2$ 
or $\varphi_1 \wU \varphi_2$, or $\varphi$, where $\varphi, \varphi_1, \varphi_2$ 
are disjunctions of conjunctions over $\Phi(\Sigma \cup  X)$. Each conjunct 
has the form $\psi \wedge x \in R$  where $\psi \in \Phi(\Sigma)$ and $R \in \mathcal{R}$. 
 In example \ref{eg1}, the regions are $R_0=[0,0], R_1=(0,1), R_2=[1,1], R^+_1=(1, \infty)$.
$\Beh(s_{\ell}, R_1, R^+_1)=(\top, \wB b,\wB b \vee \epsilon,b \weaku \bot)$, and
  $\Beh(s_a, R_0, R^+_1)=(\wB(a \vee b), \wB(a \vee b)\vee \epsilon, \epsilon, (a \vee b))$. 
  If $\epsilon$ in the sole entry in a region, it represents that there is no behaviour in that region.
   If $\epsilon$ is a disjunct $\psi \vee \epsilon$, then it represents a possibility of no behaviour, or a behaviour $\psi$.

Using the $\BD$s of $s_a$, we can write the $\sfmtl$ formula that describes the behaviour of $s_a$.
This fomula is given by $\psi(s_a)=\varphi_{R_0}(s_a) \wedge  \varphi_{R_1}(s_a)
 \wedge \varphi_{R_2}(s_a) \wedge \varphi_{R^+_1}(s_a)$, where 
 each $\varphi_{R_i}$ 
 describes the behaviour starting from region $R_i$, while in location $s_a$. 
For a fixed region $R_i$, $\varphi_{R_i}(s_a)$ is  
$\bigwedge_{R_g \prec R_i} \reg_{R_g} \epsilon \wedge \reg_{R_i}\Sigma^+ \rightarrow
%\{\bigvee_{R_i \prec R_j, \Beh(s_a, R_i, R_j)}[\bigwedge_{R_i \preceq R_w \preceq R_j} \re(\Beh(s_a, R_i, R_j)[R_w] \}$. 
\{\bigvee_{R_i \prec R_j} \varphi(s_a, R_i, R_j)\}$, where 
$\varphi(s_a, R_i, R_j)$ is described above. $\reg_{R_g} \epsilon$ represents that there is no behaviour 
in $R_g$.  Recall that $\varphi(s_a, R_i, R_j)$ describes a possible behaviour of $s_a$ that starts at $R_i$ and ends in $R_j$.  
For instance,
 $\varphi_{R_0}(s_a)$ is
$\reg_{R_0} \Sigma^+ \rightarrow
\{(\reg_{R_0}(a+b)^* \wedge \reg_{R_1}[(a+b)^*+\epsilon] \wedge \reg_{R_2}\epsilon \wedge \reg_{R^+_1}(a+b)^*)\}$
while 
 $\varphi_{R_1}(s_a)$ is
$\reg_{R_0} \epsilon \wedge \reg_{R_1}\Sigma^+ \rightarrow
\{(\reg_{R_1}(a+b)^* \wedge \reg_{R_2}\epsilon \wedge \reg_{R^+_1}(a+b)^*)\}$. 
Similarly,   $\varphi_{R_2}(s_a)$ is empty since $s_a$ has no behaviour in $R_2$.
Finally,  $\varphi_{R^+_1}(s_a)$ is 
$\bigwedge_{R_g \prec R^+_1} \reg_{R_g} \epsilon \wedge \reg_{R^+_1}\Sigma^+ \rightarrow
\reg_{R^+_1}(a+b)^*$. In a similar manner, we can write the $\sfmtl$ formula $\psi_{s_{\ell}}$ that describes the behaviour 
of $s_{\ell}$ across regions.

%  \begin{enumerate}
%  \item 	
 \noindent{\bf $\BD(s,R_i, R_j)$ for a location $s$ which is higher up}. 
 If $s$ is not the lowest in the partial order, then $\Beh(s)$ 
 has locations  $s' \in \downarrow s$. $s'$ occurs as $\nex(s')$ or  $x.\nex(s')$ in 
 $\Beh(s)$.  We now elaborate the operations needed to combine $\BD$s. 

\noindent{\bf Boolean Combinations of $\BD$s.}
Let $s_1, s_2$ be two locations of the $\po$-1-clock ATA $\mathcal{A}$. 
Assume $\Beh(s_1)=\varphi_1 \wU \varphi_2$ or  
$\varphi_1 \weaku \varphi_2$ and 
$\Beh(s_2)=\psi_1 \wU \psi_2$ or  
$\psi_1 \weaku \psi_2$. 
We have already seen how to handle $x. \nex \Beh(s_1)$
or $x. \nex \Beh(s_2)$. So let us assume $s_1, s_2$ appear in $\Beh(s)$ as 
$\nex \Beh(s_1)$ and $\nex \Beh(s_2)$.

Consider $\BDset(s_1)$ and $\BDset(s_2)$, and consider any pair 
of $\BD$s, say  $\BD(s_1, R_i, R_j)$ and $\BD(s_2, R_i, R_j)$
from these respectively. The boolean operations are defined 
for each pair taken from $\BDset(s_1)$ and $\BDset(s_2)$. 
 
Take $\BD(s_1, R_i, R_j)$ and $\BD(s_2, R_i, R_k)$ respectively from 
 $\BDset(s_1)$ and $\BDset(s_2)$.
   We now define boolean operations $\wedge$ and $\vee$ on these $\BD$s.\\
%   such that $\Beh(\varphi_1,R) \wedge \Beh(\varphi_2,R) =
%  \Beh(\varphi_1 \wedge \varphi_2,R)$. 
The  $\BDset$ for $s_1 \wedge s_2$: 
   Consider $\BD_1=\BD(s_1,R_i,R_j)$ and $\BD_2=\BD(s_2,R_i,R_k)$, both which describe behaviours of $s_1, s_2$ starting in region $R_i$.    
   Assume $R_j \prec R_k$ (the case of $R_k \prec R_j$ is similar).   To obtain a 
    $\BD$ conjuncting these two,  starting  in region $R_i$, we do the following.
   Construct $\BD'$ by conjuncting 
   the entries of $\BD_1, \BD_2$ component wise. 
    This will ensure that we take the possible behaviour of $\Beh(s_1)$ at region $R_i$ and conjunct it with the possible behaviour of $\Beh(s_2)$ in the same region. $BD' \in \BDset(s_1 \wedge s_2)$.   
     In a similar way, we can also compute the $\BDset(s_1 \vee s_2)$.
  
\paragraph*{Elimination of $\nex \Beh(s')$ from $\BD(s, R_i, R_j)$} Given any $\Beh(s)$ of the form $[\bigvee_i(P_i \wedge C_i)] \wU [\bigvee_j(Q_j \wedge E_j)]$ or $[\bigvee_i(P_i \wedge C_i)] \weaku [\bigvee_j(Q_j \wedge E_j)]$
  with $P_i, Q_j \in \Phi(\Sigma \cup S)$, and $C_i, E_j$ are clock constraints of the form $x \in R$. 
   Assume that we have calculated  $\BD(s',R, R')$  for all $s' \in \downarrow s$ and all regions $R,R'$. 
       There might be some propositions of the form $\nx \Beh(s')$ as a conjunct in some entries of 
       $\BD(s,R_i, R_j)$.      This occurrence of $\nx \Beh(s')$ is eliminated by ``stitching'' 
       the behaviour of $s'$ with  $\BD(s,R_i, R_j)$ as follows:
  \begin{itemize}
    \item We consider three cases here, depending on how $\nex \Beh(s')$ occurs in $\BD_1=\BD(s, R_i, R_j)$. 
    As a first case, let $\BD_1=(X_0,\ldots, X_{g-1}, Q_j \wedge \nx(\Beh(s')), X_{g+1},\ldots, X_{2K})$.
  \begin{enumerate}
    \item  To eliminate 
    $\nex \Beh(s')$ from $\BD_1$, we first recall that $s' \in \downarrow s$ and that $\BD(s', R_k, R_l)$ has been computed 
    for all regions $R_k, R_l$. $\Beh(s)$ will not occur in any of these $\BD$s corresponding to $s'$. 
  \item The first thing to check is which region ($R_g$ or later) where the next point can be enabled, based on the behaviour
    of $s'$.  There are $2K-g+1$ possibilities, depending on which region $\ge g$ the next point lies with respect to  $Q_j \wedge \nx(\Beh(s'))$.
 \begin{itemize}
    \item Suppose the next point can be taken in $R_g$ itself. This means that from the next point, all the possible behaviours described by 
    any of the $\BD$'s $\BD(s',R_g, R_h)$ will  
  apply along with $\BD_1$. 
%  Thus, we first take a cross product  $\BD \times \Beh(F(T_j), R_g)$ which will give us pairs of sequences
%     of the form $[X_0,\ldots, X_{g-1}, Q_j \wedge \nx(T_j), X_{g+1},\ldots, X_{2K}],[Y_0,\ldots, Y_{2K}]$. 
    We define a binary operation $\combine$ which combines two $\BD$s,   $\BD_1=\BD(s, R_i, R_j)$ and 
    $\BD_2=\BD(s',R_g, R_h)$, producing a new $\BD$, $\BD_3=\combine(\BD_1, \BD_2)$. 
          To combine the behaviours from the point where $\Beh(s')$ is encountered, we substitute 
    $\nex \Beh(s')$ with the $\ltl$ formula asserted at region $R_g$ in $\BD_2$. If $\BD_2[R_g]$ represents 
    the $g$th component, then we  replace $\nex \Beh(s')$ in $\BD_1[R_g]$ with $\BD_2[R_g]$. 
    Thus, $\BD_3[R_g]=Q_j \wedge \BD_2[R_g]$. 
    For all $R_w \prec R_g$,  $\BD_3[R_w]=\BD_1[R_w]$. 
  For all $R_w$ such that $R_g \prec R_w$,  $\BD_3[R_w]=\BD_1[R_w] \wedge \BD_2[R_w]$. 
     
         \item Now consider the case when the next point is taken a region $>R_{g}$.
       The next point can occur in region $R_{g+1}$ or higher.  Let
    $b \in \{g+1,\ldots,2K\}$, and assume that the next point where $\Beh(s')$ has a behaviour is in $R_b$. 
    Then given $\BD_1=\BD(s, R_i, R_j)$ and $\BD_2=\BD(s', R_g, R_h)$ such that 
    $\BD_2[g+1], \dots, \BD_2[b-1]=\epsilon$, we obtain $\BD_3$ as follows.
     For all $R_w \prec R_g$, $\BD_3[R_w]= \BD_1[R_w]$. For $w=g$, $\BD_3[R_g]= Q_j \wedge \Box \bot$. The conjunction 
     with $\Box \bot$ signifies that the next point in $R_g$ is not available for $s'$, since $s'$ has no behaviour 
     in $R_g$.      For all $b>w>g$, $\BD_3[R_w] = \BD_1[R_w] \wedge \epsilon=\epsilon$. 
    This implies the next point from where the assertion $Q_j \wedge \nx(\Beh(s'))$ was made is in a region $\geq R_b$.
     For all $w\ge b$, $\BD_3[R_w] = \BD_1[R_w] \wedge \BD_2[R_w]$. 
     This combines the assertions of both the behaviours from the next point onwards.
\end{itemize}
  \end{enumerate}
  
    \item As a second case, consider $\BD_1=[X_0,\ldots, X_{g-1}, \wB (P_j \wedge \nx(\Beh(s'))), X_{g+1},\ldots, X_{2K}]$. Elimination 
    of $\nex \Beh(s')$ in this case is similar to case 1. 
   \item  As the third case, let $\BD_1 {=} [X_0,\ldots, X_{g-1}, P_i {\wedge} \nx(\Beh(s_1)) {\wU} Q_j {\wedge} \nx(\Beh(s_2)), 
    X_{g+1},\ldots, X_{2K}] $, and we have to eliminate both $\nx \Beh(s_1)$ and $\nx \Beh(s_2)$. 
   Either $Q_j \wedge \nx(\Beh(s_2))$ is true at the present point or, 
  $P_i \wedge \nx(\Beh(s_1))$ is true until some point in the future within the region $R_g$, at which point, $Q_j \wedge \nx(\Beh(s_2))$ 
  becomes true.    Thus, $\BD_1$ can be replaced with two $\BD$s 
  \begin{itemize}
  \item     $\BD'_1 {=} [X_0,\ldots, X_{g-1}, Q_j {\wedge} \nx(\Beh(s_2)), X_{g+1},\ldots, X_{2K}]$, and 
   \item  $\BD''_1 {=} [X_0,\ldots, X_{g-1}$, $P_i {\wedge} \nx(\Beh(s_1)) {\until} Q_j {\wedge} \nx(\Beh(s_2)), X_{g+1}, \ldots, X_{2K}]$. 
  \end{itemize}
Elimination of $\nex \Beh(s_2)$ is done from $\BD'_1$ as seen in case 1. 
   Consider $\BD''_1$ which guarantees that the next point from which the assertion
   $P_i \wedge \nx(\Beh(s_1)) \until Q_j \wedge \nx(\Beh(s_2))$ is made is within region $R_g$, 
   and that $\nx(\Beh(s_1))$ is called for the last time within $R_g$. 
    $\BD''_1$  has to be combined with any $\BD(s_1,R_g, R_h)$, 
    which has a starting behaviour of $s_1$ from region $R_g$.  
    $s_2$ can have an enabled transition  from any point either within region $R_g$ or a succeeding region. 
  \begin{itemize}
  \item Consider the case where $s_2$ has an enabled transition from within the region $R_g$. 
  In this case, we have to combine  $\BD''_1$ with  
   some  $\BD_3=\BD(s_1, R_g, R_h)$ and  
  with some  $\BD_4=\BD(s_2, R_g, R_j)$. 
  Let $\BD_3=(Y_{0},\ldots, Y_{2K})$ and let $\BD_4=(Z_0, \dots, Z_{2K})$.
 We now show to combine $\BD''_1, \BD_3$ and $\BD_4$ 
 obtaining a $\BD$ $(A_0,\ldots, A_{2K})$.
\begin{itemize}
\item   For every $w <g$, $A_w = X_w$. 
For $w = g$, $A_g$ is obtaining  by replacing $\nex \Beh(s_1)$ with $Y_g$  and $\nex \Beh(s_2)$ with $Z_g$ 
For all $w>g$, $A_w = X_w \wedge Y_w \wedge Z_w$. 
 \end{itemize}
 \item Now  consider the case where $s_2$ has an enabled transition from a region $R_b$ such that $R_g \prec R_b$. 
 In this case, $A_w=X_w$  for $w<g$. The main difference with the earlier case is that we have to assert that from the last point in $R_g$, the next point only occurs in the region $R_b$. Thus all the regions between $R_g$ and $R_b$ should be $\epsilon$ 
 in $(A_0,\ldots, A_{2K})$. That is, $A_w=\epsilon$ for $g < w<b$. 
 For $w=g$,  $A_g=(P_i \wedge  Y_{g}) \until (Q_j \wedge \Box \bot)$, where $P_i, Q_j$ are obtained from $\BD''_1[R_g]$. 
 Here again, conjuncting $\Box \bot$ with $Q_j$ signifies that the next point 
 is not enabled for $s_2$. Finally, for $w \geq b$, $A_w=X_w \wedge Y_w \wedge Z_w$. 
% 
% 
% 
% 
% 
% 
% 
% Take a cross product of  $\BD$ with $\Beh(F(U_i),R_g) {\times} \Beh(F(T_j),R_b)$.  This 
% gives us triplets of the form 
% $[X_0,\ldots, X_{g-1}, Q_j {\wedge} \nx(T_j), X_{g+1},\ldots, X_{2K}]$, 
% $[Y_{U,0},\ldots, Y_{U,2K}], [Y_{T,0},\ldots, Y_{T,2K}]$. The one difference in combining this triplet 
% as compared to the last one is that we have to assert that from the last point in $R_g$, the next point only occurs in the region $R_b$.  
% Thus all the regions between $R_g$ and $R_b$ should be conjuncted with $\wB \bot$.
%  We get a sequence  $[X'_0,\ldots, X'_{2K}]$ after combining,  such that 
%  \begin{itemize}
%  \item   For all $w<g$, $X'_w = X_w$. 
%  \item For $w=g$,  $X'_g=(P_i \wedge \nx Y_{U,g}) \until Q_j$. 
%  \item For $b>w>g$, $X'_w = X_w \wedge Y_{U,w} \wedge \wB \bot$. 
%  \item For $w \ge b$, $X'_w = X_w \wedge Y_{U,g} \wedge Y_{T,g}$.
%\end{itemize}
 
%   \item In case of formulae of the form  $[P_i \wedge \nx(U_i)] \weaku [Q_j \wedge \nx(T_j)]$ in $R_g$, 
%   we convert it into  $(\alpha_1 \until \alpha_2) \vee \wB \alpha_1$ where $\alpha_1=(P_i {\wedge} \nx(U_i))$ and $\alpha_2=(Q_j {\wedge} \nx(T_j))$. Then 
%   $\BD$ is $[X_0,\ldots, X_{g-1}, [\alpha_1 {\weaku} \alpha_2], X_{g+1},\ldots, X_{2K}]$
%   and  can be replaced by 2 $\BD$s
%   \begin{itemize}
%   \item    $\BD_1 {=} [X_0,\ldots, X_{g-1}, \alpha_1 {\wU} \alpha_2, X_{g+1},\ldots, X_{2K}]$
%     \item $\BD_2 {=} [X_0,\ldots, X_{g-1}, \wB(\alpha_1), X_{g+1},\ldots, X_{2K}]$. 
%    \end{itemize} 
%     For $\BD_1$ and $\BD_2$,  
%     we apply the operations defined previously. 
\end{itemize}
\end{itemize}
Note that elimination of $\nex Beh(s')$ from any $\BD$ in the set $\BDset(s)$ 
results in stitching some $\BD$ from $\BDset(s')$ to certain elements 
of $\BDset(s)$. At the end of the stitching, we obtain $\BDset(s)$ such that 
in each $\BD$ of $\BDset(s)$, $\nex \Beh(s')$  has been replaced.  

\paragraph*{Obtaining $\sfmtl$ Formulae}   
 Finally, we show that given a $\BD$ for $\Beh(s)$, we can construct an $\sfmtl$ formula, $\psi_s$, equivalent to $x.\nx(s)$.  
  That is, $\rho,i \models \psi_s$ if and only if $\rho,i,\nu \models x.\nx(\Beh(s))$, for any $\nu$. 
  Recall that $\Beh(s)$ is a 1-$\tptl$ formula, as computed in lemma \ref{aut-tptl-1}. 
  We give a constructive proof as follows:

    Assume $\rho,i,\nu \models x.\nx(\Beh(s))$. Note that according to the syntax of $\mathsf{TPTL}$, every constraint $x \in I$ checks the time elapse between the last point where $x$ was frozen. Thus satisfaction of formulae of the form $x.\phi$ at a point is independent of the clock valuation. 
  $\rho,i,\nu \models x.\nx(\Beh(s))$ iff  $\rho,i, \nu[x \leftarrow \tau_i] \models \nex \Beh(s)$. 
  We have precomputed $\BD(s,R_i, R_j)$ for all regions $R_i \preceq R_j$; and 
  $\BD(s,R_i, R_j)$  is guided by the 1-$\tptl$ formula $\Beh(s)$. The entry in region $R_i$
  of $\BD(s,R_i, R_j)$ depends on the behaviour allowed in region $R_i$ from location $s$;
  likewise, the entry in each region $R_g$ of $\BD(s,R_i, R_j)$ is 
  obtained by looking up $\Beh(s)$. In case $\Beh(s)$ does not admit any behaviour in a region $R_g$, then 
  the $g$th entry in $\BD(s,R_i R_j)$ is $\epsilon$.   
   Thus, $\rho,i,\nu \models x.\nx(\Beh(s))$ iff for all  $w \in{0,\ldots,2K}$, such that 
   $\Beh(s)$ has an allowed behaviour in region $R_w$, 
      $\rho,i+1,\tau_i \models (x \in R_w)$. 
     In addition, we also know that there is some $\BD(s,R_w, R_j)$ such that 
 $\BD[R_k]$ is the $\mathsf{LTL}$ formula that describes the behaviour in region $R_k$ of location $s$.

 % Let $\varphi(s,R_i,R_j)= \bigwedge_{i \leq g \leq j}\reg_{R_g}\re(\Beh(s, R_i, R_j)[R_g])$.
   Note that, $\rho,i+1,\tau_i {\models} (x \in R_w)$ is true, iff, $\rho,i \models \bigwedge \limits_{g \in \{1,\ldots,w-1\}}[\reg_{R_g}\epsilon] \wedge 
    \reg_{R_w}\Sigma^+$.  This is true iff 
         $\rho,i {\models} \bigvee\limits_{\BD=\BD(s,R_w,R_j)} \bigwedge \limits_{k \in \{1,\ldots,2K\}} \reg_{R_k}(\re(\BD[R_k]))$, where $\re(\BD[R_k])$ is a star-free rational expression equivalent to the LTL formula $\BD[R_k]$.

       Thus, $\rho,i,\nu {\models} x.\nx(\Beh(s))$, iff, $\rho,i {\models} (\psi_1 {\rightarrow} \psi_2)$ where  
       \begin{itemize}
       \item   $\psi_1{=}{\bigwedge \limits_{w \in \{0,\ldots,2K\}\setminus E}} \bigwedge \limits_{g \in \{1,\ldots,w-1\}} \reg_{R_g}\epsilon { \wedge}
    {\reg_{R_w}\Sigma^+}$ and 
       \item $\psi_2{=}\bigvee\limits_{\BD=\BD(s,R_w,R_j)} \bigwedge \limits_{k \in \{1,\ldots,2K\}} \reg_{R_k}(\re(\BD[R_k]))$. 
  \end{itemize}  
  where $E$ is the set of regions where $\Beh(s)$ has no behaviour.
    The $\sfmtl$ formula  $\psi_{s_0}$ is one which begins in the initial location $s_0$, stitches 
    the behaviours of locations $s_j$ that appear in a run from $s_0$ such that $L(\psi_{s_0})$ 
    is non-empty iff the language accepted by the $\po$-1-clock ATA $\mathcal{A}$ is non-empty, and 
     $L(\psi_{s_0})=L(\mathcal{A})$.

 Consider  the $\po$-1-clock ATA $\mathcal{A}=(\{a,b\}, \{s_0,s_a,s_{\ell}\}, s_0, \{s_0, s_{\ell}\},\delta)$ with transitions 
 $\delta(s_0,b)=s_0,\delta(s_0,a)=(s_0 \wedge x.s_a) \vee s_{\ell},$
 $\delta(s_a,a)=(s_a \wedge x<1) \vee (x>1)=\delta(s_a,b),$ and 
 $\delta(s_{\ell},b)=s_{\ell}, \delta(s_{\ell},a)=\bot$.       
  
Consider the subset of $L(\mathcal{A})$ consisting of timed words whose first 
symbol occurs at a time $>1$. We write a $\sfmtl$ formula that captures this subclass. 

Let us consider the formula we obtain if we consider allowed behaviours from
$s_0$ that begin in the region $R^+_1$; this is the subset of $\BDset(s_0)$ consisting of 
$\BD(s_0, R^+_1, R^+_1)=(\top, \top, \top, [(a\wedge x.\nex \Beh(s_a)) \vee b]\weaku 
(a \wedge \nex \Beh(s_{\ell})))$. We look at the $\sfmtl$ formula $\psi_{s_a}$ corresponding to 
$x.\nex \Beh(s_a)$, which is given by\\ 
$\reg_{R_0} \Sigma^+ \rightarrow \{\reg_{R_0}(a+b)^* \wedge \reg_{R_1}((a+b)^*+\epsilon) \wedge \reg_{R_2} \epsilon \wedge 
\reg_{R^+_1}(a+b)^*\}$ 
$\wedge$\\
$\reg_{R_0} \epsilon \wedge \reg_{R_1}\Sigma^+ \rightarrow \{\reg_{R_0}\top \wedge \reg_{R_1}(a+b)^* \wedge \reg_{R_2} \emptyset \wedge 
\reg_{R^+_1}(a+b)^*\}$ 
$\wedge$\\
$\reg_{R_0} \epsilon \wedge \reg_{R_1}\epsilon \wedge \reg_{R_2}\epsilon \wedge \reg_{R^+_1} \Sigma^+  \rightarrow 
\reg_{R^+_1}(a+b)^*$ \\

This formula $\psi_a$ is plugged in place of $x. \nex \Beh(s_a)$ in $\BD(s_0, R^+_1, R^+_1)$. We now combine  
$\BD(s_{\ell}, R^+_1, R^+_1) \in \BDset(s_{\ell})$  with  $\BD(s_0, R^+_1, R^+_1)[R^+_1]$ to obtain the combined behaviour 
of locations $s_{\ell}$ from the next point along with that of $s_0$. 
We know that $\BD(s_{\ell}, R^+_1, R^+_1)=(\top, \top, \top, b \weaku \bot)$. Thus, we obtain 
$\BD(s_0, R^+_1, R^+_1)$ after combining with   $\BD(s_{\ell}, R^+_1, R^+_1)$ and 
$\psi_{s_a}$ as $(\top, \top, \top,[(a\wedge \psi_{s_a}) \vee b]\weaku (a \wedge (b \weaku \bot)))$. 
Translating this into an $\sfmtl$ formula, we obtain the formula $\varphi_{R^+_1}(s_0)$\\
$\reg_{R_0} \epsilon \wedge \reg_{R_1}\epsilon \wedge \reg_{R_2} \epsilon 
\wedge \reg_{R^+_1} \Sigma^+ \rightarrow \reg_{R^+_1}\re([(a\wedge \psi_{s_a}) \vee b]\weaku (a \wedge (b \weaku \bot)))$.

$\varphi_{R^+_1}(s_0)$ is the formula which captures the subset of $L(\mathcal{A})$ which consists of timed words 
of the form $(a_1,t_1)(a_2,t_2) \dots (a_n, t_n)$ such that $t_1>1$. We can also write the formulae 
$\varphi_{R_0}(s_0)$, $\varphi_{R_1}(s_0)$, $\varphi_{R_2}(s_0)$, which capture respectively, the subset of words 
 of $L(\mathcal{A})$ which consists of timed words 
of the form $(a_1,t_1)(a_2,t_2) \dots (a_n, t_n)$ where $t_1=0$, $0< t_1<1$ and $t_1=1$ respectively. 
Thus, $L(\mathcal{A})$ is the union of the languages 
$L(\varphi_{R^+_1}(s_0)), L(\varphi_{R_0}(s_0)), L(\varphi_{R_2}(s_0))$ and 
$L(\varphi_{R_1}(s_0))$.

%  \end{enumerate}

\end{document}